\documentclass[11pt]{amsart}
\usepackage{amsfonts}
\usepackage{amsmath,amsthm}
\usepackage{latexsym}
\usepackage{latexsym}
\usepackage{array}
\usepackage{amssymb}
\usepackage{enumerate}
\usepackage{comment}
\usepackage{graphicx}
\usepackage{tikz}
\usepackage{color}

\DeclareFontFamily{U}{wncy}{}
\DeclareFontShape{U}{wncy}{m}{n}{<->wncyr10}{}
\DeclareSymbolFont{mcy}{U}{wncy}{m}{n}
\DeclareMathSymbol{\Sh}{\mathord}{mcy}{"58}

\usepackage[francais,english]{babel}
\usepackage{fdsymbol}
\usepackage{color}
\usepackage[latin1]{inputenc}

\numberwithin{equation}{section}
\theoremstyle{plain}
\newtheorem{theorem}{Theorem}[section]
\newtheorem*{theorem*}{Theorem}
\newtheorem{lemma}[theorem]{Lemma}
\newtheorem{proposition}[theorem]{Proposition}
\newtheorem{corollary}[theorem]{Corollary}

\newtheorem{problem}[theorem]{Problem}

\theoremstyle{remark}
\newtheorem{remark}[theorem]{Remark}
\newtheorem{example}[theorem]{Example}

\newtheorem*{lem*}{Lemma}
\newtheorem*{sublem*}{Sublemma}
\newtheorem*{remark*}{Remark}
\newtheorem*{NB*}{NB}

\newcommand{\dy}{  d_{x^\perp}y}
\newcommand{\ddy}{  d_{y^\perp}}
\newcommand{\tr}{  \|_{r}}

\newcommand{\Sv}{ \Sigma^v}
\newcommand{\Svm}{ \Sigma^{vm}_s}
\newcommand{\Sst}{{}^s\!\Sigma_*}
\newcommand{\Sod}{ \Sigma^1}

\newcommand{\nbh}{\text{neighbourhood}}
\newcommand{\id}{\text{id}}

\newcommand{\etl}{e^{-\tau \Lc}}
\newcommand{\etlt}{e^{-(\tau-t) \Lc}}

\newcommand{\nd}{n^{\le2}(0)}

\newcommand{\thb}{ \zeta }
\newcommand{\thet}{\theta }

\newcommand{\R}{ \mathbb{R} }

\newcommand{\C}{ \mathbb{C} }
\newcommand{\Z}{ \mathbb{Z} }
\newcommand{\N}{ \mathbb{N} }

\newcommand{\T}{ \mathbb{T} }

\newcommand{\gF}{\mathfrak F}
\newcommand{\fF}{{\frak F}}
\newcommand{\fn}{\mathfrak n}

\newcommand{\Del}{\Delta}

\newcommand{\bs}{ \bold{l} }

\newcommand{\bm}{ \bold{m} }

\newcommand{\ssum}{{\modtwosum}}

\newcommand{\cA}{ \mathcal{A} }

\newcommand{\cC}{ \mathcal{C} }

\newcommand{\Cc}{ \mathcal{C} }

\newcommand{\cY}{ \mathcal{Y} }

\newcommand{\cF}{ \mathcal{F} }

\newcommand{\EE}{ {\mathbb E}}

\newcommand{\Lc}{ \mathcal{L} }
\newcommand{\cZ}{ \mathcal{Z} }
\newcommand{\cM}{ \mathcal{M} }
\newcommand{\cN}{ \mathcal{N} }

\newcommand{\Rc}{ \mathcal{R} }

\newcommand{\cT}{ \mathcal{T} }

\newcommand{\cJ}{ \mathcal{J} }
\newcommand{\I}{ \mathcal{I} }

\newcommand{\om}{ \omega }
\newcommand{\nuu}{ {\nu'\, }}

\newcommand{\ga}{\gamma }

\newcommand{\s}{ \sigma }

\newcommand{\G}{ \Gamma }
\renewcommand{\phi}{ \varphi }
\newcommand{\oms}{ \omega^{12}_{3s}  }

\newcommand{\eps}{\varepsilon}

\newcommand{\de}{ \delta }
\newcommand{\al}{ \alpha }
\newcommand{\zz}{\mathfrak z}

\newcommand{\la}{ \lambda }

\newcommand{\dess}{\delta'^{12}_{3s}}
\newcommand{\des}{\delta^{12}_{3s}}
\newcommand{\deq}{\delta^{12}_{34}}

\newcommand{\dist}{{\operatorname{dist}}}

\newcommand{\be}{\begin{equation}}
\newcommand{\ee}{\end{equation}}
\newcommand{\ben}{\begin{equation*}}
\newcommand{\een}{\end{equation*}}

\newcommand{\ov}{ \overline }

\newcommand{\lan}{ \langle }

\newcommand{\ran}{ \rangle}

\newcommand{\p}{ \partial}

\newcommand{\wt}{ \widetilde }

\newcommand{\lbl}{\label}
\newcommand{\non}{\nonumber}
\newcommand{\qu}{\quad}
\newcommand{\qmb}{\quad\mbox}
\newcommand{\qnd}{\qmb{and}\qu}

\newcommand{\volna}{\thicksim}

\newcommand{\atw}[1]{a^{(2)}_{#1}}

\newcommand{\bo}[1]{\bar a^{(1)}_{#1}}

\newcommand{\Ga}{\Gamma}
\newcommand{\goa}{\mathfrak{a}}

\newcommand{\dep}{\delta'^{12}_{3s}}
\newcommand{\depp}{\delta'^{1'2'}_{3's}}

\title[Stochastic model for wave turbulence 1: kinetic limit]
{Formal expansions in stochastic model for wave turbulence 1: kinetic limit}
\author{Andrey  Dymov}
\address{Steklov Mathematical Institute of RAS, Moscow 119991, Russia 
	\& National Research University Higher School of
	Economics, Moscow 119048, Russia} \email{dymov@mi-ras.ru}
\author{Sergei Kuksin}
\address{Universit\'e Paris-Diderot (Paris 7), UFR de Math\'ematiques - Batiment Sophie Germain, 5 rue Thomas Mann, 75205 Paris,
 France  \& School of Mathematics, Shandong University, Jinan, PRC \& Saint Petersburg State University, Universitetskaya nab., St. Petersburg, Russia}
\email{ Sergei.Kuksin@imj-prg.fr}

\begin{document}

\begin{abstract}
We consider the damped/driven (modified) cubic NLS equation on a large torus with a properly scaled forcing and dissipation, and 
decompose its solutions to formal series in the amplitude. We study  the second order truncation of this series and prove that when the 
amplitude goes to zero and the torus' size goes to infinity the energy spectrum of the truncated solutions becomes close to a
solution of the damped/driven wave kinetic equation. Next we discuss higher order truncations of the series. 
\end{abstract}

\date{}

	\maketitle
	%\centerline { \today}
%	\newpage	
	\tableofcontents

	% \section{Introduction  }\label{s1}
	
	%%%\newpage
	
	\section{Introduction}\label{s2}

	\subsection{The setting}
	The wave turbulence (WT) was developed in 1960's as a heuristic tool to study  small-amplitude
	oscillations in nonlinear Hamiltonian PDEs and Hamiltonian systems on lattices. 
	We start with recalling basic concepts of the theory in application to the cubic non-linear Schr\"odinger equation (NLS).
	\smallskip

	\noindent {\it  Classical setting}.
	Consider the cubic NLS equation
	\be\label{nls}
	\frac{{{\partial}}}{{{\partial}} t}  u +i \Delta u- i\lambda \, |u|^2 u= 0\,,\qquad 
	%\Delta = {(2\pi )^{-2}}  \sum_{j=1}^d ({{\partial}}^2 / {{\partial}} x_j^2)\,,\;\;
	x\in {\mathbb{T}}^d_L={  \mathbb{R}   }^d/(L{  \mathbb{Z}   }^d)\,,
	\ee
	where 
	$
	\Delta = {(2\pi )^{-2}}  \sum_{j=1}^d ({{\partial}}^2 / {{\partial}} x_j^2)\,,
	$
	$d\ge2$,
	$L\ge1$ and $0<\lambda\le 1$. Denote by  $H$ the space  $ L_2({{\mathbb T}}^d_L;{{\mathbb C}})$, 
	given  the $L_2$--norm  with respect to the
	normalised Lebesgue measure:
	$$
	\|u\|^2 =\|u\|^2_{L_2(\T^d_L)}=
	\lan u, u\ran\,,
	\quad
	\lan u, v\ran =
	L^{-d} \Re\int_{{{\mathbb T}}^d_L} u \bar v\,dx\,.
	$$
	The NLS equation is a hamiltonian system in $H$ with two integrals
	of motion -- the Hamiltonian and  $\|u\|^2$. The equation with the Hamiltonian  $ \lambda \|u\|^4$ is
	$\ \frac{{{\partial}}}{{{\partial}} t}  u  - 2 i\lambda \|u\|^2u=0$ and its flow commutes with that of NLS. We modify the NLS 
	equation by subtracting  $\lambda \|u\|^4$ from its Hamiltonian, thus arriving at the equation 
	\be\label{newnls}
	\frac{{{\partial}}}{{{\partial}} t}  u +i\Delta u- i\lambda \,\big( |u|^2 -2\|u\|^2\big)u =0 , \qquad x\in {\mathbb{T}}^d_L.
	\ee
	This modification is used by mathematicians, working with hamiltonian PDEs, since it keeps the main features 	of the original equation, reducing some non-crucial 	 technicalities. It is also used by physicists, studying WT;  e.g.  see 
	 \cite{Naz11}, pp.~89-90.\footnote{Note in addition that if $u(t,x)$ satisfies  eq. \eqref{newnls}, then $u'= e^{ 2it\lambda \| u \|^2}u$ is a solution of \eqref{nls}.} Below we work with  eq.~\eqref{newnls}  and write its solutions as functions 
	$u(t,x)\in {{\mathbb C}} $ or as curves  $u(t)\in H$.

	The objective  of  WT is to study solutions of \eqref{nls} and \eqref{newnls} 
	when
	\be\label{limit}
	\text{
	$\lambda\to0\ $ \ and  \ $\ L\to\infty$}
	\ee
	{  on large time intervals}.

	We will write the Fourier series for an $u(x)$ as
	\begin{equation}\label{ku2}
	u(x)= L^{-d/2} {\sum}_{s\in{{\mathbb Z}}^d_L} v_s e^{2\pi  i s\cdot x},\qquad {{\mathbb Z}}^d_L = L^{-1} {{\mathbb Z}}^d\,,
	\end{equation}
	 where 
	  the vector   of Fourier  coefficients $v=\{v_s, s\in \Z^d_L\}$    is the Fourier transform of $u(x)$: 
	\be\label{vs}
	v= \hat u=\cF(u), \qquad 
	v_s=\hat u(s) = L^{-d/2}\int_{\T^d_L} u(x) e^{-2\pi i s\cdot x}\,dx \,.
	\ee
	% {\color{blue} 
	Given a    vector $v=\{v_s, s\in \Z^d_L\}$ we will regard the sum in  \eqref{ku2} as its inverse Fourier transform 
	 $(\cF^{-1}  v)(x)$,   which we will also write as $(\cF^{-1}  v_s)(x)$. \footnote{  
	 The symmetric form of the Fourier transform which we use -- with the same
	 scaling factor $ L^{-d/2}$ for the direct and the inverse transformations -- is convenient for the heavy calculation below in the paper.}
	  I.e.,
	$$
	u(x) =( \cF^{-1} v_s)(x)= ( \cF^{-1} v)(x)\,.
	$$
	Then
	$\ 
	\|u\|^2 = {\ssum}_s |v_s|^2 , %=: \|v\|^2= \|v\|^2_{L_2(\Z^d_L)}\,,
	$
	where for a complex sequence $(w_s, s\in \Z^d_L)$ we denote 
	$$
	\ssum_{s\in\Z^d_L}w_s = L^{-d} {\sum}_{s\in\Z^d_L} w_s\,
	$$
	(this equals to the integral over $\R^d$ of $w_s$, extended to a function, constant on the cells of the mesh in  $\R^d$
	of size $L^{-1}$). By $h$ we will denote the Hilbert space 
	$h=L_2(\Z^d_L;\C)$, given the norm
	$$
	\|w\|^2 =\ssum |w_s|^2\,. 
	$$
	Abusing a bit  notation we denote by the same symbol the norms in the spaces $H$ and $h$. This is justified by the fact 
	that the  Fourier transform \eqref{vs} defines an isometry
	$
	\cF: H\to h. %\quad u(x) \mapsto (v_s=\hat u(s))\,,
	$
	%whose inverse is $\cF$. 
	
	Equations \eqref{nls}, \eqref{newnls} and  other NLS equations on the torus $\T^d_L$ with fixed $L$ are intensively studied by 
	mathematicians, e.g. see the book \cite{Bour} and references in it. The limit $\lambda\to 0$ with $L$ fixed was rigorously treated in a number of publications, e.g. see \cite{HKM}.  But there are just a few   mathematical works, addressing the limit \eqref{limit}. 
	 In  the paper \cite{FGH}  $d=2$ and the
	limit \eqref{limit} is taken in a specific regime, when $L\to\infty$ much slower than $\lambda^{-1}$. The elegant description of the limit, obtained there, 
	 is far from the prediction of  WT, and rather should be regarded as a kind of averaging.  
	In  recent  paper \cite{BGHS} the authors study eq. \eqref{nls}  
	with random initial  data $u(0,x)=u_0(x)$ such that    the phases $\{\arg v_{0s}, s\in\Z^d_L\}$ of components of the vector 	$v_0=\cF (u_0)$  
	 are  independent  uniformly distributed  random variables. In the notation of our work  they 
	 prove that under the limit \eqref{limit}, if $L$ goes to infinity slower than $\lambda^{-1}$ but not too slow, then 
  for the values of time   of order $\lambda^{-1}L^{-\de}$, $\de>0$, the energy spectrum $n_s(\tau)$ approximately satisfies the WKE, 
  linearised on $u_0(x)$ and  scaled by the factor $\lambda$.      The authors of \cite{LS} start with eq.~\eqref{nls}, replace  there the space--domain
$\T_L^d$ by the discrete torus $\Z^d/ (L\Z^d)$,  modify the discrete Laplacian on  $\Z^d/ (L\Z^d)$ to a suitable operator, diagonal in the Fourier 
basis, and  study the obtained equation,
  assuming  that the distribution of the initial data $u(0,x)$ is given by the Gibbs measure of the equation. 
See \cite{Faou} for a related result.

	From other hand, there  are plenty of physical works on	equations \eqref{nls} and  \eqref{newnls} under the limit \eqref{limit};
	 many references  may be found in \cite{ZLF92, Naz11, NR}.  These  papers contain some different (but consistent) approaches
 to the limit. Non of them was ever rigorously justified, despite the strong  interest in physical and 
 mathematical communities to the questions,  addressed by these works. 
	
	\medskip
		
	\noindent
	{\it Zakharov-L'vov setting}. When  studying  eq. \eqref{newnls}, members of the WT community  talk
	about ``pumping  the energy to low modes  and dissipating it in  high modes".
	To make this rigorous,  Zakharov-L'vov \cite{ZL75} (also see \cite{CFG}, Section~1.2)
	suggested to consider the NLS equation, dumped by a (hyper) viscosity and driven by a random force: 
	 \be\label{nice_form}
	 \begin{split}
	 \frac{{{\partial}}}{{{\partial}} t}
u +i\Delta u - i\lambda (|u|^2 - 2\|u\|^2)u = -\nu\, \frak A (u) 
% (-\Delta +1)^{r_*} u
  +\sqrt\nu\,  \frac{{{\partial}}}{{{\partial}} t} \eta^\omega(t,x),\\
\eta^\omega( t,x)
	=\cF^{-1} \big(b(s)\beta_s^\omega( t)\big)(x) = L^{-d/2} {\sum}_{s\in\Z^d_L} b(s) \beta^\omega_s( t) e^{2\pi  i s\cdot x} .
\end{split}
\ee
Here $0<\nu \le1/2$, \ 
	$\{\beta_s(t), s\in{{\mathbb Z}}^d_L\}$ are standard independent complex Wiener processes,\footnote{
		i.e. $\beta_s = \beta_s^1 + i\beta_s^2$, where $\{\beta_s^j, s\in{{\mathbb Z}}^d_L, j=1,2  \}$ are
		standard independent real Wiener processes.} 	  $b$
	is a positive  Schwartz  function on ${{\mathbb R}}^d\supset\Z^d_L$ and $\frak A$ is the dissipative linear operator 
$$
\frak A(u(x)) = L^{-d/2} {\sum}_{s\in \Z^d_L} \ga_s v_s e^{2\pi is\cdot x},  \qquad v= \cF(u),\;\;
\ga_s =\ga^0(|s|^2), 
$$
where $\ga^0(y)$ is    a smooth increasing function of $y\in\R$  that has at most a polynomial growth at infinity together with its derivatives of all orders, such that $\ga^0\ge1$ and 
\be\label{gamma0}
%\ga^0 \ge1 \;\;\; \text{and}\;\; \; 
 c(1+y)^{r_*} \le \ga^0(y), \; % \le C (1+y)^{r_*}, 
   \; |\p^k  \ga^0(y)| \le  C (1+y)^{r_*-k}\;\; \text{for} \; 0\le k \le 3, 
  \; \forall\, y\ge0.
\ee
The exponent $r_*\ge0$ and $c,C$ are positive constants.\footnote{For example, if $\ga_s = (1+ |s|^2)^{r_*}$, then $\frak A= (1-\Delta)^{r_*}$.}

It is convenient to pass to the slow time $\tau= \nu t$ and re-write eq. \eqref{nice_form}  as 
		\begin{equation}\label{ku3}
	\begin{split}
	\dot u +&i \nu^{-1} \Delta u + \frak A( u)  =  i\rho\, \big( |u|^2 - 2\|u\|^2\big)u  + \dot \eta^\omega(\tau,x),  
	 \\
	&\eta^\omega(\tau,x)=
	L^{-d/2} {\sum}_{s\in\Z^d_L} b(s) \beta^\omega_s(\tau) e^{2\pi  i s\cdot x}, %=\cF\big(b(s)\beta_s^\omega(\tau)\big)(x). 
	\end{split}
	\end{equation}
	where $\rho = \lambda\nu^{-1}$.
	Here and below the upper-dot stands for $\p/ \p\tau$ and 
	$\{\beta_s(\tau)\}$ is another set of standard independent complex Wiener processes.
	
Solutions $u(t)$ and 	 $u(\tau)$ are random processes in the space $H$. 	If $r_*\gg1$, 
 then equations \eqref{nice_form} and 	\eqref{ku3} are well posed. In the context of equation \eqref{ku3}, 
 the objective of WT is to study its solutions  when
		\begin{equation}\lbl{limit'}
			%\la\to 0,\quad 
			\nu\to 0,\quad L\to\infty.
			\end{equation}
Below we show that the main characteristics of the solutions have non-trivial behaviour under the limit above if 
$\rho\sim \nu^{-1/2}$. We will choose 
\be\label{r_scale}
\rho= \nu^{-1/2}\eps^{1/2},
\ee
and will examine the limiting characteristics of the solutions 
when $\eps$ is a fixed small positive constant. We will show that $\sqrt\eps$ measures the 
properly scaled amplitude of the solutions, so indeed it should be small for the methodology of WT to apply.
The parameter $\la =\nu\rho = \nu^{1/2} \eps^{1/2}$ will be used only in order to discuss the obtained  results in the original scaling 
\eqref{nice_form}. Note that as $\la\sim\sqrt\nu$, then we study eq.~\eqref{nice_form}=\eqref{ku3} in the regime when it is a perturbation 
of the original NLS equation \eqref{newnls} by dissipative and stochastic terms much smaller then the terms of \eqref{newnls}.
 We will examine solutions of  \eqref{ku3} on the time-scale $\tau\sim1$. For  the original fast time $t$ in 
eq.~\eqref{nice_form} it corresponds to the scale $t\sim\nu^{-1}\sim\la^{-2}\!=${\it(size of the nonlinearity)}${}^{-2}$.  This  is exactly the time-scale usually considered in  WT, see in \cite{ ZLF92, Naz11, LS}.

	Applying  Ito's formula to a solution $u$  of \eqref{ku3} and denoting  $B =  {\ssum}_s b(s)^2$
	we arrive at the balance of energy relation: 
	\begin{equation}\label{ku8}
	{{\mathbb E}} \|u(\tau)\|^2 +2{ {\mathbb E} } \int_0^\tau  \| \frak A( u(s))\|^2 \, ds ={{\mathbb E}} \|u(0)\|^2 +2B \tau\,.
	\end{equation}
	The quantity ${{\mathbb E}} \|u(\tau)\|^2$ -- the ``averaged  energy per volume" of a
	 solution $u$ -- should be of order one,
	see \cite{ZLF92, Naz11, NR}. This agrees  well with \eqref{ku8}  since there   $B\sim \int{b}^2 \,dx \sim1$ as $L\to \infty$, and we
 immediately get from \eqref{ku8} that
	\be\non
	\EE \|u(\tau)\|^2 \le B+ (\EE\|u(0)\|^2 -B) e^{-2\tau}\,,
	\ee
	uniformly in $\nu$,  $\rho$ and $L$.	
	
		Using \eqref{ku2} and the relations  
	$$
	\widehat{u_1u_2}(s)  = L^{-d/2} {\sum}_{s_1}\hat u_1(s_1) \hat u_2(s-s_1), \qquad 
	\widehat{\bar u} (s) = \bar{\hat u}_{-s},
	$$
	we write  \eqref{ku3} as the system of equations 
	\begin{equation}\label{ku33}
	\dot v_s -i\nu^{-1}|s|^2v_s + \gamma_s v_s
	=  i\rho  L^{-d}   \Big({\sum}_{s_1, s_2} %{\sum}_{s_2} 
	 \dess v_{s_1} v_{s_2} \bar v_{s_3} -|v_s|^2v_s\Big)+b(s) \dot\beta_s \,,
	%e^{  i\nu^{-1}\tau\omega}\,,
	\end{equation}
	$ s\in{{\mathbb Z}}^d_L\,,$ 	where 
	$$
 \dess =
   \left\{\begin{array}{ll}
 1,& \text{ if $s_1+s_2=s_3+s$ and $\{s_1, s_2\} \ne \{s_3, s\}$}
 \,,
 \\
 0,  & \text{otherwise}.
\end{array}\right.
$$
Note that 
%\be\label{property}
	$\text{ if  $\{s_1, s_2\}\cap \{s_3, s\}\ne\emptyset$, then  ${\delta}'^{1 2}_{3 s} =0$. }$
	%\ee
 In  view of the factor  $\dess $,
	in  the double  sum in \eqref{ku33}  the index $s_3$ is a function of $s_1, s_2, s$.

Now using the interaction representation
	\be\label{interac}
	v_s = \exp({  i \nu^{-1}\tau |s|^2} ) \, \mathfrak{a}_s\,, \;\; s\in {{\mathbb Z}}^d_L\,,
	\ee
	we re-write  \eqref{ku33} as
	\begin{equation}\label{ku4a}
	\begin{split}
	&\dot \goa_s + \gamma_s \goa_s
	=  i\rho \Big(\cY_s(\goa,\nu^{-1}\tau) -L^{-d} |\goa_s|^2\goa_s\Big)
	+b(s) \dot\beta_s  \,,\quad s\in{{\mathbb Z}}^d_L\,,\\
	&\cY_s(\goa; t )=L^{-d} {\sum}_{s_1} {\sum}_{s_2} \dess \goa_{s_1} \goa_{s_2} \bar \goa_{s_3}
	e^{  i t \omega^{12}_{3s}}\,.
	\end{split}
	\end{equation}
	Here $\{\beta_s\}$ is an another set of  standard independent complex  Wiener processes, and
	\be\label{omega}
	\omega^{12}_{3s}= |s_1|^2+|s_2|^2-|s_3|^2 - |s|^2 = 2(s_1-s)\cdot (s-s_2)\,
	\ee
	(the second equality holds since $s_3=s_1+s_2-s$  due to  the  factor $\dess$).

\smallskip

Despite stochastic models of WT (including \eqref{nice_form}) are popular in physics, it seems that the only their mathematical
studies were performed by  the authors of this work and their collaborators,  and in paper  \cite{Faou}. 
In the same time the models, given by Hamiltonian chains of equations with stochastic perturbations, have 
received significant attention from the mathematical community, e.g. see \cite{EPRB,KOR, Dym16} and references in these 
works. The equations in the  corresponding works are related to the Zakharov--L'vov model, written as the perturbed chain of hamiltonian equations
\eqref{ku33}, but differ from it  significantly  since there, in difference with \eqref{ku33}, the interaction between the 
equations is local. This leads to rather different results, obtained by tools, rather different from those in our paper.

\subsection{The results.}

	The {\it energy spectrum} of a solution $u(\tau)$ of eq.~\eqref{ku3} is the function
\be\lbl{real_spect}
{{\mathbb Z}}^d_L \ni s\mapsto \cN_s(\tau) = \cN_s(\tau;\nu, L)
=   {{\mathbb E}} |v_s (\tau) |^2 =  {{\mathbb E}}|\goa_s (\tau)|^2.
\ee
Traditionally in the center of attention of  the WT community is 
 the limiting behaviour of the energy spectrum  $\cN_s$, as well as of  correlations of solutions $\goa_s(\tau)$ and $v_s(\tau)$ 
under the limit \eqref{limit'}.
Exact meaning of the latter   is unclear since no relation between the small parameters $\nu$ and
$L^{-1}$ is postulated by the theory. In \cite{KM16, HKM} it was proved  that for $\rho$ and $L$
fixed, eq.~\eqref{ku3} has a limit as $\nu\to0$, called the {\it limit of discrete turbulence,} see \cite{Kar, Naz11} and Appendix~\ref{app_discr_turb}. 
Next it was demonstrated in \cite{KM15} on the physical level of
rigour that if we scale $\rho$ as $\tilde\eps\sqrt L$, $\tilde\eps\ll1$, then the iterated limit $L\to\infty$ leads to a 
kinetic equation for the energy spectrum. Attempts to justify this  rigorously    so far  failed. Instead in
this work we specify the limit \eqref{limit}  as follows:
\be\label{assumption}
\begin{split}
&\nu\to0 \text{ and } L\ge\nu^{-2-\epsilon} \text{ for some }  \epsilon>0, \\
&\text{ or first $L\to\infty$ and next $\nu\to0$}
\end{split}
\ee
(the second option formally corresponds to the first one with $\epsilon=\infty$). 
To present the results  it is convenient  for  the moment  to regard 
 $\rho$ as an independent parameter, however  later we will choose it  to be of the 
form \eqref{r_scale} with  a fixed   small  positive constant.

Accordingly to \eqref{assumption}, everywhere in the introduction we assume that 
\be\label{L_is_big}
L\ge \nu^{-2-\epsilon} \ge 1 , \quad \epsilon >0. 
\ee

Let us supplement  equation \eqref{ku3}=\eqref{ku4a} with the initial condition
\be\label{in_cond}
u(-T)=0, 
\ee
for some $0<T \le +\infty$, and in the spirit of WT decompose a solution of \eqref{ku4a}, \eqref{in_cond} to formal series in $\rho$:
\be\label{decompa}
	\goa=\goa^{(0)}+\rho \goa^{(1)}+  \dots .
\ee
Substituting the series in the equation we  get that $\goa^{(0)}$ satisfies the linear equation 
	$$
	\dot \goa_s^{(0)} + \gamma_s \goa_s^{(0)}
	= b(s) \dot\beta_s  \,,\quad s\in{{\mathbb Z}}^d_L\,,
	$$
	so this is the Gaussian process
	\be\label{a0}
	\goa^{(0)}_s(\tau) = b(s) \int_{-T}^\tau e^{-\gamma_s(\tau-l)}d\beta_s(l),
	\ee
	while $\goa^{(1)}$ satisfies
	$$
	\dot \goa^{(1)}_s (\tau)+ \gamma_s \goa^{(1)}_s (\tau)
	=  i \cY_s(\goa^{(0)}(\tau),\nu^{-1}\tau)  - iL^{-d}|\goa^{(0)}_s(\tau)|^2\goa_s^{(0)}(\tau),\qquad \tau> - T\,,
	$$
	so
	\be\label{a1a}
		\begin{split}
	\goa^{(1)}_s(\tau) = iL^{-d} \int_{-T}^\tau  e^{-\ga_s(\tau-l)} 
			\Big( \sum_{s_1,s_2}&\delta'^{12}_{3s} (\goa^{(0)}_{s_1} \goa^{(0)}_{s_2}{\bar \goa}^{(0)}_{s_3})(l)	e^{i\nu^{-1} l \oms} \\
			& - |\goa_s^{(0)}(l)|^2\goa_s^{(0)}(l)\Big)
	\,dl\,
	\end{split}
	\ee
	is a Wiener chaos of third order (see \cite{Jan}). 
	Similarly, for $ n\ge1$
	\be\label{ana}
	\begin{split}
	\goa^{(n)}_s&(\tau) 
	= iL^{-d} \int_{-T}^\tau  e^{-\ga_s(\tau-l)}\,\sum_{n_1+n_2+n_3=n-1} \\ 
	& \Big(\sum_{s_1,s_2}
	 \delta'^{12}_{3s}\big( \goa_{s_1}^{(n_1)} \goa_{s_2}^{(n_2)} {\bar \goa_{s_3}}^{(n_3)}\big)(l)
	e^{i\nu^{-1} l \oms} \,- \,\big(\goa_s^{(n_1)}\goa_s^{(n_2)}\bar \goa_s^{(n_3)}\big)(l)\Big)\,dl
	\end{split}
	\ee
 is a Wiener chaos of order $2n+1$.

\medskip
		
\noindent
{\it  Quasisolutions}. 		
It is traditional in  WT to retain the quadratic in $\rho$ part of the decomposition \eqref{decompa} and analyse it, postulating that it
well approximates  small amplitude solutions. Thus motivated we start our analysis with the quadratic truncations of the series \eqref{decompa}, 
 which we call the {\it quasisolutions} and denote $\cA(\tau)$. So 
\be\lbl{A-intro}
\cA(\tau) = (\mathcal{A}_s(\tau), s\in\Z^d_L), \qquad \mathcal{A}_s(\tau) = \goa^{(0)}_s(\tau) +\rho \goa^{(1)}_s(\tau) +\rho^2 \goa^{(2)}_s(\tau)\,,
\ee
where $\goa^{(0)}$, $\goa^{(1)}$ and $\goa^{(2)}$ were defined above. 
The energy spectrum of a quasisolution $\cA(\tau)$ is 
\be\lbl{n_s-a}
\mathfrak{n}_s(\tau) = \EE |\cA_s(\tau)|^2,\quad s\in\Z_L^d, 
\ee
where 
\be\lbl{n_s_decomp}
\fn_s=\fn_s^{(0)} + \rho \, \fn_s^{(1)}+\rho^2 \fn_s^{(2)} 
+ \rho^3 \fn_s^{(3)} + \rho^4 \fn_s^{(4)}, \qu s\in\Z_L^d,
\ee
with
\be\label{n_s^ka}
\fn_s^{(k)} = \sum_{k_1+k_2=k,\, k_1, k_2\le2} \EE \goa^{(k_1)}_s \bar \goa^{(k_2)}_s, \quad 0\le k\le 4. 
\ee
In particular, $\fn_s^{(0)} = \EE |\goa_s^{(0)}|^2 $ and we easily derive from \eqref{a0} that 
\be\label{n_s_0}
\fn_s^{(0)} (\tau)= \EE |\goa_s^{(0)}(\tau)|^2 =B(s) \big( 1-e^{-2\ga_s(T+\tau)}\big), \quad
B(s) =  \frac{b(s)^2}{\ga_s}. 
\ee
So 
 $\fn_s^{(0)} $ extends to  a Schwartz function of $s\in \R^d$, uniformly in $\tau\ge-T$. Also it is not hard 
  to see that $\fn_s^{(1)} =0$, 
 see in Section \ref{s_series}. The coefficients $\fn_s^{(k)}$ with $k\ge2$ are  more complicated.
 
 The processes $\goa_s^{(r)}(\tau)$ and the functions $\fn_s^{(r)} (\tau)$, $r\ge0$,  depend on $\nu$ and $L$. When it will be needed
 to indicate this dependence, we will write them as 
 $\goa_s^{(r)}(\tau)=\goa_s^{(r)}(\tau;\nu,L)$, etc. The  dependence of the  objects on $T$  will not be  indicated.

 It was explained on the heuristic and half--heuristic level  in many physical works concerning various models of WT that
 the term $\rho^2 \fn^{(2)} (\tau)$ is the crucial non-trivial component of the energy spectrum, while the terms 
  $\rho^3 \fn^{(3)} (\tau)$ and  $\rho^4 \fn^{(4)} (\tau)$ make its  perturbations. 
  Our results justify
   this insight on the 
  energy spectra of quasisolutions. 
  Firstly we consider $\fn_s^{(2)} = \EE  |\goa_s^{(1)}|^2 +2\Re \EE \goa_s^{(2)} \bar \goa_s^{(0)}$. 
  The two terms on the right are similar. Consider the first one, $\EE  |\goa_s^{(1)}|^2$. The theorem below describes its asymptotic 
   behaviour under the limit \eqref{assumption}, where for simplicity we assume that $T=\infty$. 
   
   We set $B(s_1,s_2,s_3):=B(s_1)B(s_2)B(s_3)$ and denote by
  $C^\#(s)$ various  positive continuous  functions of $s$ which decay as $|s|\to\infty$
 faster than any negative degree of $|s|$. 
  %To indicate that a function $C^\#(s)$ depends on a parameter $p$ we write it as $C^\#(s;p)$.
 \be\label{constants}
 \begin{split}
&\text{The constants $C, C_1$ etc and the functions  $C^\#(s)$}\, \\ %C^{\#}(s;p)$} \\
&\text {never depend on $\nu, L, \rho, \eps$  and on the times $T,\tau, \quad$} \\
 \end{split} 
 \ee
 unless the dependence is indicated. 
 
 \medskip

  \noindent {\bf Theorem A.} 
{\it  Let in \eqref{in_cond} $T=\infty$. Then for any $\tau$ and 	any $s\in\Z^d_L$, 
\be\label{0_asympt}
	\begin{split}
	 \Big| \EE  |\goa_s^{(1)}(\tau)|^2 -
		 \frac{\pi \nu }{\ga_s}\int_{\Sigma_s} %\frac1{\sqrt{|x|^2+|y|^2}} 
		\frac{ B(s_1, s_2, s_1+s_2-s) }{{\sqrt{|s_1-s|^2+|s_2-s|^2}} \ }
		\, ds_1ds_2\!\!\mid_{\Sigma_s} \Big|\\
\le  \big(\nu^2 +L^{-2} \nu^{-2}\big)C^\#(s).
	\end{split}
	\ee
 Here 
 $\Sigma_s$ is the quadric  $ \{(s_1, s_2): (s_1-s) \cdot (s_2-s)=0\}$ and 
  $ds_1ds_2\!\!\mid_{\Sigma_s}$ is the volume element  on it, corresponding to the Euclidean 
 structure on $\R^{2d}$.  
 	If $d=2$, then the term $C^\#(s)\nu^2$ in the r.h.s. of \eqref{0_asympt}   should be replaced by $ C^\#(s;\aleph) \, \nu^{2-\aleph}$, where $\aleph$ is arbitrary positive number.}
 \medskip
 
 See Theorems~\ref{t_sum_integral}, \ref{t_singint} and Corollary~\ref{t_for_sum}. Due to \eqref{omega}, $\Sigma_s =\{(s_1, s_2, s_3): s_1+s_2 =s_3+s$ and $ \oms=0\}$. This quadric is the set of resonances for 
 eq.~\eqref{ku4}.

 Denote $F_s(s_1,s_2) := (s_1-s)\cdot (s_2-s) = -\tfrac12 \oms$. Then $|\partial F_s | $
 equals the divisor of the integrand in \eqref{0_asympt}. So the integral in \eqref{0_asympt} is  exactly what physicists call 
 the integral of $B$  over the delta-function of $F_s$ and denote 
$  \int B\delta(F_s)$, see \cite{ZLF92}, p.~67.  
For rigorous mathematical treatment of this object see \cite{Gelf}, Section~III.1.3, 
or \cite{Khin}, pp.~36-37.
As
$\
\delta(F_s) = \delta \big(-\tfrac12\, \oms \big) =  -2\delta( \oms ),
$
where $s_3:= s_1+s_2-s$, 
then neglecting the minus-sign (as physicists do) we may write  
 the integral from \eqref{0_asympt} as 
\be\label{nazar}
 \nu\frac{2\pi }{\ga_s} \int  B(s_1, s_2, s_3)\,\delta( \oms) \des\,ds_1 ds_2 ds_3
 \ee
(since $\int\dots\, \des\,ds_1 ds_2 ds_3 = \int\dots\mid_{s_3=s_1+s_2-s} ds_1 ds_2$).
 \medskip

 Theorem A and its variations play important role in our work since the terms, quadratic in $\goa^{(1)}$ and in
  its increments, as well as quadratic terms,  linear in  $\goa^{(0)}$, 
    $\goa^{(2)}$ and in their increments, play a leading role in the analysis of the energy spectrum  $\fn_s(\tau)$. It
     turns out that their asymptotic behaviour under the limit 
 \eqref{assumption} is described by integrals, similar to \eqref{0_asympt}.
 These results imply that 
 	\be\lbl{n^2_est}
 	\fn_s^{(2)}\sim \nu,
 	\ee 
 (recall \eqref{L_is_big}).
 	The terms  $\fn_s^{(3)}$  and  $\fn_s^{(4)}$ and their variations which appear in our analysis are 
 	smaller:
 	\be\label{n_34}
 	|\fn_s^{(3)}|,\, |\fn_s^{(4)}| \le C^\#(s) \nu^2. 
 	\ee
 	If $d=2$ the estimate for $\fn_s^{(3)}$ is slightly weaker: $|\fn_s^{(3)}|\leq C^\#(s)\nu^2\ln\nu^{-1}$.
  Upper bounds \eqref{n_34} may be obtained by annoying but rather straightforward analysis of certain integrals with oscillating exponents, based on results from
 	Sections~\ref{sec:quasisol} and \ref{s_9},  but instead we get them from a much  deeper result, presented below in Theorem~D. 
 	
 	 Estimates \eqref{n^2_est} and \eqref{n_34} justify the scaling \eqref{r_scale} since for such a choice of $\rho$  
 	\be\lbl{eps-expan}
 	\fn_s=\fn_s^{(0)}+ \eps \wt \fn_s^{(2)} + O(\eps^2)
 	\ee 
 	(if $\nu\ll1$), where $\wt \fn_s^{(2)}=\nu^{-1} \fn_s^{(2)}\volna 1$. 
 	Thus, the principal non-trivial component of $\fn_s$ is given by $\rho^2\fn_s^{(2)}=\eps \wt \fn_s^{(2)}\volna \eps$, 
while the other terms in \eqref{n_s_decomp}=\eqref{eps-expan} are smaller, of the size $O(\eps^2)$. This also shows that 
	 the small  parameter $\sqrt\eps$ measures the  properly scaled  amplitude  of the oscillations.

 \medskip
 
 Since  equation \eqref{ku4a} for the processes $\goa_s(\tau)$ fast oscillates in  time, then the task to describe the behaviour of
  $\fn_s(\tau)$  has obvious similarities with the problem, treated by the Krylov--Bogolyubov averaging (see in \cite{AKN}). 
  Accordingly
 it is natural to try to study the required  limiting behaviour following the suit of  the Krylov--Bogolyubov theory. That is, by considering the increments 
 $\fn_s(\tau+\theta) -\fn_s(\tau)$ with $\nu \ll \theta \ll1$ and passing firstly to the limit as $\nu\to0$ and next -- to the limit $\theta\to0$. That insight was exploited heuristically 
 in many works on WT (e.g. see \cite{Naz11}, Section~7), while  in \cite{KM16, HKM} it was rigorously applied to pass to the  limit of  discrete turbulence
 $\nu\to0$ with $L$ and $\rho$ fixed. In this work we also argue as it is customary in the classical averaging and analyse the increments $\fn_s(\tau+\theta) -\fn_s(\tau)$,
 using the asymptotical results like Theorem~A and estimates like \eqref{n_34}.  
 This analysis shows that due to the important role, played by the integrals like \eqref{0_asympt}, 
 the leading nonlinear contribution to the increments is described by the cubic wave kinetic integral operator 
  $K$,  sending a function $y(s)$, $s\in \R^d$, to the function 
  $$
 K_s(y(\cdot)) = 2\pi  
  \int_{\Sigma_s} \frac{ds_1\, ds_2\!\mid_{\Sigma_s}\, y_{1}y_{2}y_{3}y_{s}}{ \sqrt{|s_1-s|^2+|s_2-s|^2}}\left(
  \frac1{ y_{s}} +\frac1{ y_{3}} -\frac1{ y_{1}}-\frac1{ y_{2}} 
\right),
 $$
 where $y_s=y(s)$ and for $j=1,2,3$ we denote     $y_j=y(s_j)$ with $s_3 = s_1+s_2-s$. 
     Using the notation  \eqref{nazar}, $ K_s(y(\cdot))$ may be written as
 $$
  4\pi  
  \int
  y_{1}y_{2}y_{3}y_{s} %{ \sqrt{|s_1-s|^2+|s_2-s|^2}}
  \left(
\frac1{ y_{s}} +\frac1{ y_{3}} -\frac1{ y_{1}}-\frac1{ y_{2}}
\right) \,\delta( \oms) \, \des\,ds_1 ds_2 ds_3. 
 $$
  This is exactly the wave kinetic integral which appears in physical works on  WT to describe the 
   4--waves  interaction, see \cite{ZLF92}, p.~71 and \cite{Naz11}, p.~91.

% If $r_*=0$, then $\ga_s\equiv1$ and this  integral coincides with the kinetic integral, used to describe  WT for the 
% 4--waves  interaction, see \cite{ZLF92}, p.~71 and \cite{Naz11}, p.~91. 

    %%%%%%%%%

    The operator $K$ is defined in terms of the measure
 $
 \mu_s \!=\! \frac{ds_1\, ds_2}{ ({|s_1-s|^2+|s_2-s|^2})^{1/2}}    \!\mid_{\Sigma_s}$, %({|s_1-s|^2+|s_2-s|^2})^{-1/2}=$, 
   and our study of the operator  is based on the following useful disintegration of  $\mu_{0}$ (the measure $\mu_s$ with $s=0$), obtained in 
   Theorem~\ref{t_disintegr}:
 $$
 \mu_{0}(ds_1, ds_2) = |s_1|^{-1} ds_1\,d_{s_1^\perp} s_2, 
 $$
 where for a non-zero vector $s_1$ we denote by $d_{s_1^\perp} s$ the Lebesgue measure on 
   the hyper--space ${s_1^\perp}=\{ s: s\cdot s_1 =0\}$. Based on this disintegration, in Section~\ref{s_kin_int} we prove the result
 below, where for $r\in\R$ we  denote by $\Cc_r(\R^{d})$ the space of continuous complex functions on $\R^{d}$ with finite norm
 $
 | f|_r =\sup_x|f(x)| \lan x\ran^r.
 $

   \medskip
  
  \noindent {\bf Theorem B.} 
{\it For any  $r>d$ the operator $K$ defines a continuous 3--homo\-ge\-neous mapping }
 $
 K:\Cc_r(\R^d) \to \Cc_{r+1}(\R^d).
 $
 \medskip
 
 See Theorem~\ref{t_kin_int}.
Now
consider the damped/driven wave kinetic equation (WKE): 
  \be\label{wke}
 \dot{m}(\tau, s) =-2\ga_s {m}(\tau,s) +\eps K({m}(\tau, \cdot))(s) +2b(s)^2,\quad {m}(-T)=0.
\ee
 In Theorem~\ref{t_kin_eq} we easily derive from Theorem~B that for small $\eps$ this equation has a unique solution $m$. The latter can be 
 written as 
 $
 m =m^0(\tau,s) + \eps m^1(\tau,s),
 $
 where $m^0, m^1\sim1$,  $m^0$ is a solution of the linear equation  \eqref{wke}$\!{}\mid_{\eps=0}$ and equals $\fn^{(0)}$. 
 Analysing the increments $\fn_s(\tau+\theta) - \fn_s(\tau) $ using the 
 results, discussed above, in Section~\ref{sec:quasisol} we show that they  have approximately 
 the same form as the increments of 
 solutions for eq.~\eqref{wke}.  Next in Section~\ref{s_WKE},  arguing by analogy with 
  the classical averaging theory, we get the stated below 
   main result of this work (recall  agreement  \eqref{constants}): 
   \medskip

  \noindent {\bf Theorem C  (Main theorem)}.
{\it   The energy spectrum  $\fn_s(\tau)= \fn_s(\tau;\nu,L)$  of the quasisolution $\cA_s(\tau)$ of \eqref{ku4a}, \eqref{in_cond} 
	 satisfies the estimate $\fn_s(\tau)\leq C^\#(s)$ and	is close to the solution $m(\tau, s)$ of WKE \eqref{wke}. 
Namely,  under the scaling \eqref{r_scale} for 
any $r$ there exists $\eps_r>0$  such that for $0<\eps \le \eps_r$ we have 
 \be\lbl{TC-est}
| \fn_\cdot(\tau) - m(\tau, \cdot )|_r \le C_r\eps^2\qquad \forall\, \tau\ge-T, 
 \ee
 if $0<\nu\le \nu_\eps(r)$ for a suitable $ \nu_\eps(r)>0$,  and if $L$ satisfies \eqref{L_is_big}. 
  Moreover, the limit 
 $
  \fn_s(\tau;\nu,\infty)$ of $ \fn_s(\tau;\nu,L) $ as ${L\to\infty} $
 exists,  is a Schwartz function of $s\in\R^d$} and also satisfies the estimate above for any $r$. 
\smallskip

 Since the energy spectrum $\fn_s$ is defined for $s\in\Z^d_L$ with finite $L$, then the
  norm in \eqref{TC-est} is understood as $|f|_r=\sup_{s\in\Z^d_L}|f(s)| \lan s\ran^r$.

For $\eps=0$ eq. \eqref{wke} has the unique steady state $m^0$, $m^0_s= b_s^2/\ga_s$, which  is asymptotically stable. By
the inverse function theorem, for $\eps<\eps'_r$ $(\eps'_r>0)$, eq.~\eqref{wke} has a unique steady state $m^\eps \in\Cc_r(\R^d)$, close to $m^0$.
It also is asymptotically stable. Jointly with Theorem~C this result describes the asymptotic in time behaviour of the 
energy spectrum $\fn_s$:
\be\label{time_ass}
| \fn_\cdot(\tau) -m^\eps_\cdot|_r \le |m^\eps_\cdot|_r e^{-\tau-T} +C_r\eps^2, \qquad \forall\, \tau\ge -T. 
\ee
See in Section \ref{s_WKE}. 
\medskip

Due to Theorem A and some modifications of this result, the iterated limit
$ \ 
\lim_{\nu\to0} \lim_{L\to\infty} \nu^{-1} \fn_s^{(2)}(\tau; \nu, L) 
$
exists and is non-zero. It is hard to doubt  (however, we have not proved this yet) that a similar iterated limit also exists for $ \nu^{-2} \fn_s^{(4)}$  
(cf. estimate \eqref{n_34}). Then, in view of  \eqref{n_34} for $\fn_s^{(3)}$, 
under the scaling \eqref{r_scale} exists a limit
$
\fn_s(\tau; 0,\infty)=   \lim_{\nu\to0} \lim_{L\to\infty} \fn_s(\tau; \nu,L).
$
If so, then $\fn_s(\tau; 0,\infty)$ also satisfies the assertion of Theorem~C and obeys  the time--asymptotic \eqref{time_ass}.

\medskip

Theorems A--C are proved in Sections \ref{s_series}--\ref{s_WKE} and Sections \ref{s_singint}--\ref{s_proof_kinetic}, where Sections~\ref{s_singint}--\ref{s_proof_kinetic} contain a
demonstration of Theorem~A as well as of some lemmas, needed to prove Theorems~B,~C in Sections~\ref{s_series}--\ref{s_WKE}. 

Section \ref{s_en_spectra} presents the results of our second paper \cite{DK} on  formal expansions \eqref{decompa} 
in series in $\rho$. There we 
decompose the spectrum $\cN_s(\tau)$,  defined in  \eqref{real_spect}, in formal series,
\be\label{N_s}
\cN_s(\tau) = \fn_s^0(\tau) + \rho \fn_s^1(\tau) + \rho^2 \fn_s^2(\tau)+\dots .
\ee
 where $\fn_s^0=\fn_s^{(0)}$ and for $k\ge1$ 
 $\fn_s^k$ is given by the formulas \eqref{n_s^ka} with  the restriction
$k_1, k_2\le2$ being dropped, i.e. 
$
\fn_s^{k} = \sum_{k_1+k_2=k} \EE \goa^{(k_1)}_s \bar \goa^{(k_2)}_s, 
$
so  $\fn^j_s$ equals $ \fn_s^{(j)}$ for $j\le2$, but not for $j=3,4$. Still the estimates \eqref{n_34} remain true for $\fn^3_s$ and $\fn^4_s$. The bounds  on  $\fn_s^{0},\ldots,\fn_s^{4}$ suggest that $|\fn_s^k|\lesssim \nu^{k/2}$. The
 results of Section~\ref{s_en_spectra} put some light on this assumption. Namely, it is  proved there  that $\fn_s^k(\tau)$ may be approximated by a finite sum $\sum_{\cF_k} I_s^k(\cF_k)$ 
of integrals $I_s^k(\cF_k)$, naturally parametrised by a certain class of Feynman diagrams $\cF_k$. 
These integrals satisfy the following assertion, where  $\lceil x\rceil$ stands for the smallest integer $\ge x$. 
  \medskip
  
  \noindent {\bf Theorem D}.
{\it For each $k$, 

a) every integral $I_s^k(\cF_k)$ satisfies 
  \be\label{IS_est}
| I_s^k(\cF_k)| \le C^\#(s;k) \max(\nu^{\lceil k/2\rceil }, \nu^d),
\ee
if $L$ is so big that $L^{-2} \nu^{-2} \le \max(\nu^{\lceil k/2\rceil }, \nu^d)$. 
If $d=2$ and $k=3$ then the maximum in the r.h.s.  above should be multiplied by $\ln\nu^{-1}$.

b)  estimate \eqref{IS_est}  is sharp in the sense that if $k>2d$ 
 (so that $\max(\nu^{\lceil k/2\rceil }, \nu^d)$ $=\nu^d$),  then for some diagrams $\cF_k$ we have 
  \be\label{IS_est>}
| I_s^k(\cF_k)| \sim C^\#(s;k)\nu^d \gg C^\#(s;k) \nu^{\lceil k/2\rceil }. 
\ee}
\smallskip

Theorem D is a new result on the integrals with fast oscillating quadratic exponents (see the integral in \eqref{I_F}). We hope that the theorem
and its variations (cf. \cite{Dym} and the last section of \cite{DK}) will find applications outside the framework of WT. 

As  $\fn_s^k$ is  approximated by a finite sum of integrals $I_s^k(\cF_k)$, it 
also obeys  \eqref{IS_est}. Since $d\ge2$, then \eqref{IS_est} implies estimate \eqref{n_34}, needed to
prove Theorem~C. Also, by \eqref{IS_est}  
\be\label{true?}
| \fn^k_s| \lesssim  \nu^{k/2} \quad  \text{if}\quad k \le 2d.
\ee
Validity of inequality $| \fn^k_s|\lesssim \nu^{k/2} $ for $k>2d$ is a delicate issue. Assertion \eqref{IS_est>} implies that the inequality  
does not hold for all summands, forming $\fn^k_s$,  but still it may hold 
 for  $\fn_s^k = \sum_{\cF_k} I_s^k(\cF_k)$ due to  cancellations. And indeed,
we prove that  some cancellations do happen
in the sub-sum over a certain subclass $\gF_k^B$ of diagrams $\cF_k$ which give  rise to the biggest integrals $I_s^k(\cF_k)$. Namely, for any $\cF_k\in\gF_k^B$ we have \eqref{IS_est>} but
$\big|\sum_{\cF_k\in\gF_k^B} I_s^k(\cF_k)\big| \le C^\#(s;k) \nu^{k-1}\leq C^\#(s;k)\nu^{k/2}$. This 
does not imply the validity of \eqref{true?} for all $k$,  which remains an open problem:

\begin{problem}\label{c_1}
Prove that for  any $k\in\N$ 
\be\label{conjecture}
| \fn_s^k(\tau)| \le C^\#(s;k)\nu^{k/2} \quad\forall\,s,\ 
 \forall \tau\ge-T,
\ee
if $L$ is sufficiently big in terms of $\nu^{-1}$. 
\end{problem}

If the conjecture \eqref{conjecture} 
is correct, then under the scaling \eqref{r_scale}, for any $M\ge2$ the order $M$ truncation  of the series \eqref{N_s}, 
 namely 
$\ 
\cN_{s,M} (\tau) = \sum_{0\le k\le M} \rho^k \fn_s^k(\tau) , 
$
also meets the assertion of Theorem~C, i.e. satisfies the WKE with the accuracy $\eps^2$. 
It is unclear for us if $\cN_{s,M}$ satisfies the equation with better accuracy, i.e. if it 
better approximates a solution of \eqref{wke} than $\fn_s(\tau)$. 
On the contrary, if  \eqref{conjecture} fails in the sense that for some $k$ we have 
$ \|\fn^{k}_\cdot\| \gtrsim  C\nu^{k'}$  with  $ k'<k/2$, 
then under the scaling \eqref{r_scale}  the sum \eqref{N_s}  is not  a formal series in $\sqrt\eps$,  uniformly in $\nu$ and $L$.

\subsection{Conclusions.} 

$\bullet$ If in eq.  \eqref{ku3} $\rho$ is chosen to be $\rho=\nu^{-1/2} \eps^{1/2}$ with  $0<\eps\le1$, then the 
energy spectra $\fn_s(\tau)$ 
 of quasisolutions for the equation (i.e. of quadratic in  $\rho$ truncations of solutions $u$, decomposed in formal series in  $\rho$) under the limit
\eqref{assumption} satisfy the damped/driven wave kinetic equation \eqref{wke} with accuracy $\eps^2$. If we write the
 equation which we study  using the original fast time $t$ in the form \eqref{nice_form}, then the kinetic limit exists if there 
$\lambda\sim\sqrt\nu$. The time, needed to arrive at the kinetic regime is $t\sim\lambda^{-2}$. 

$\bullet$ If \eqref{conjecture} is true, 
then the energy spectra of higher order truncations for decompositions of solutions for \eqref{ku3} in series in $\rho$ also satisfy \eqref{wke}
at least  with the same accuracy $\eps^2$. 
%On the contrary, if \eqref{conjecture}  fails in the sense that for some $k$ we have \eqref{con_fail}, then \eqref{N_s}  is not
% a formal series in $\sqrt\eps$,  uniformly in $\nu$ and $L$.

$\bullet $ Similar, if the energy spectrum $\cN_s(\tau)= \EE |v_s|^2$ admits  the second order truncated Taylor decomposition in $\rho$ of the form 
$$
\cN_s(\tau) = \fn_s^{(0)} +\rho^2 \fn_s^{(2)}  + O(\rho^3\nu^{3/2}),
$$
then under the scaling \eqref{r_scale} and the limit  \eqref{assumption}
the spectrum $\cN_s(\tau)$ satisfies the WKE \eqref{wke} with accuracy $\eps^{3/2}$.   At this point
we recall that different  NLS equations and their damped/driven versions 
appear in physics as models for small oscillations in various media, obtained by neglecting in 
the exact equations  terms of higher order of 
smallness. So it is not  impossible that the kinetic limit holds for the energy spectra of the second order
 jets in $\rho$ of solutions $v(\tau)$ (which we call quasisolutions),   but not for the solutions themselves 
since the former are closer to the physical reality. 

$\bullet$ 
To prove our results we have developed in this work and in its second part  \cite{DK}  new  analytic and combinatoric  techniques.
 They apply to quasisolutions of  equations  \eqref{nice_form} under the WT limit \eqref{limit'} 
 and for this moment give no non-trivial information about the  exact solutions. 
   Still we believe that these techniques make a basis for further study of the 
damped/driven NLS equations under the WT limit, as well  as of other stochastic models  of WT.
 In particular, we hope that being applied to some other   
 models  they will give there  stronger and ``more final"  results. To verify this belief is our next goal.

	\subsection{Notation}
	\lbl{s:notation}
	 By $\R_+^n$ and  $\R_-^n$ we denote, respectively, the sets $[0,\infty)^n$ and $(-\infty,0]^n$. For a vector $v$ we denote by $|v|$ its Euclidean norm and by $v\cdot u$ its Euclidean scalar product with a vector $u$. We write
	$\lan v\ran =(1 +|v|^2)^{1/2}$. For a real number $x$, $\lceil x\rceil$ stands for the smallest integer $\ge x$. 
	We denote  by  $\chi_d(\nu)$ 
		the constant 
		\be\label{chi_d}
		\chi_d(\nu)=
		\left\{\begin{array}{ll}
			1,& d\ge3 \,,
			\\
			 \ln \nu^{-1} ,& d=2\,.
		\end{array}\right.
		\ee 	
The exponent $\aleph_d$ is zero if $d\ge3$ and is any positive number if $d=2$. 
	 	 For an integral $I= \int_{\R^N} f(z)\,dz$ and a domain $M\subset \R^N$, open or closed, 
		 we write 
		 \be\lbl{<I,M>}
		 \lan I, M\ran  =\int_M f(z)\,dz. 
		 \ee
		 Similar we write
	 $\lan |I|, M\ran =\int_M |f(z)|\,dz$.

	We denote
	$\ssum_{s_1,\ldots,s_k\in\Z^d_L}:=L^{-kd}\sum_{s_1,\ldots,s_k\in\Z^d_L}$.
	Following the tradition of WT, we often abbreviate $v_{s_j},\, \goa_{s_j}, \,\ga_{s_j},\ldots$ to $v_{j},\, \goa_{j},\, \ga_{j},\ldots$,
	and abbreviate the sums 
	$\sum_{s_1,\ldots,s_k\in\Z^d_L}$ and $\ssum_{s_1,\ldots,s_k\in\Z^d_L}$ to   $\sum_{1,\ldots,k}$ and $\ssum_{1,\ldots,k}$.
	 By $\des$ we denote the Kronecker delta of the relation $s_1+s_2=s_3+s$. 
	 % (equal one if the relation holds and vanishing  otherwise).

	 Finally,   $C^\#(\cdot), C_1^\#(\cdot),\dots $ stand for various  non-negative  continuous functions, fast decaying at infinity:
	%their norms  decay faster than the norms of their arguments in any negative degree.  That is,
	\be\label{S}
	0\le C^\#(x)\le C_N\lan x\ran^{-N}\quad \forall\, x\,,
	\ee
	for every $N$, with suitable constants $C_N$.
	By $C^{\#}(x;a)$ we denote a   function $C^{\#}$, depending on a
	 parameter $a$. Below we discuss some properties of the functions $C^\#$.
\smallskip

	{\it Functions $C^\#$}.	
	For any 
	 Schwartz function $f$,   $|f(x)|$  may be written as $C^\#(x)$.
	If $\Lc: \R^n \to \R^n$ is a linear isomorphism, then  the function $g(y) = C^\#(\Lc y)$ may be
	written as $C_1^\#(y)$. 
	Next, for  any $C^\#(x,y)$, where $(x,y)\in\R^{d_1+d_2}$, $d_1, d_2\ge1$,  there exist  $C_1^\#(x)$  and $C_2^\#(y)$
	such that
	$$
	C^\#(x,y)  \le C_1^\#(x) C_2^\#(y).
	$$
	Indeed, consider $C_0^\#(t) = \sup_{|(x,y)|\ge t} C^\#(x,y)$, $t\ge0$. This is a non-increasing continuous function,  satisfying \eqref{S}. Then
	$$
	C^\#(x,y)  \le   \sqrt{ C_0^\#(
		\tfrac1{\sqrt2}(
		|x|+|y|))}\, \sqrt{ C_0^\#(\tfrac1{\sqrt2}( |x|+|y|))} \le C_1^\#(x)  C_2^\#(y) \,,
	%	\sqrt{ F^\#(\tfrac1{\sqrt2}|x|)}\,\sqrt{ F^\#(\tfrac1{\sqrt2}(|y|)}\,.
	$$
	where 
	$C_{j}^\#(z) = \big( C_0^\#(\tfrac1{\sqrt2} |z|)\big)^{1/2}$, $j=1,2$. 	Finally,  for any $C_1^\#(x) $ and $C_2^\#(y)$ the function 
	$C_1^\#(x)  C_2^\#(y) $ may be written as $C^\#(x,y)$. 
	\bigskip

\noindent
{\bf Acknowledgments.} AD was supported by Russian Foundation for Basic Research
%RFBR 
 according to the research project 18-31-20031, and SK -- by  Agence Nationale de la Recherche through 
 the grant    17-CE40-0006.
% and the grant 18-11-00032 of Russian Science Foundation. 
 We thank Johannes~Sj\"os\-trand for discussion and an anonymous referee for careful reading of the paper and pointing out some flaws.

	\section{Formal decomposition of solutions % of $a$--equation 
	 in series in $\rho$ 
	}\label{s_series}
	
	\subsection{ Approximate $a$-equation}

	Despite that the term $L^{-d}|\goa_s|^2\goa_s$ is 
	 only a small perturbation in equation \eqref{ku4a}, it is rather inconvenient for our analysis.
	  Dropping it  we consider the  following more convenient  equation:
	\begin{equation}\label{ku4}
	\dot a_s + \gamma_s a_s
	=  i\rho \cY_s(a,\nu^{-1}\tau)
	+b(s) \dot\beta_s  \,,\quad a_s(-T)=0\,, \qu s\in{{\mathbb Z}}^d_L\,.
	\end{equation}
	Similar to  the process $\goa_s$, we decompose  $a_s$ to the formal series in~$\rho:$
	\be\label{decomp}
	a=a^{(0)}+\rho a^{(1)}+  \dots .
	\ee
	Here $a_s^{(0)}(\tau) =\goa^{(0)}_s(\tau)$ is  the Gaussian process  given by \eqref{a0}, 	while the processes  $a_s^{(n)}(\tau)$ with $n\ge1$, where 
	\be\label{a1}
			a^{(1)}_s(\tau) = iL^{-d} \int_{-T}^\tau  e^{-\ga_s(\tau-l)} 
			 \sum_{s_1,s_2}\delta'^{12}_{3s} (a^{(0)}_{s_1} a^{(0)}_{s_2}{\bar a}^{(0)}_{s_3})(l)	e^{i\nu^{-1} l \oms} 
			\,dl
	\ee
	and for $ n\ge1$
	\be\label{an}
	\begin{split}
			a^{(n)}_s&(\tau) 
			= iL^{-d} \int_{-T}^\tau  e^{-\ga_s(\tau-l)}\,\\
			&\times \sum_{n_1+n_2+n_3=n-1} 
			\sum_{s_1,s_2}
			\delta'^{12}_{3s} (a_{s_1}^{(n_1)} a_{s_2}^{(n_2)} {\bar a_{s_3}}^{(n_3)}\big)(l)
			e^{i\nu^{-1} l \oms} \,dl,
	\end{split}
	\ee
	are Wiener chaoses of order $2n+1.$
	 In Section~\ref{s:a-goa} we prove that the processes $a^{(n)}(\tau)$ and $\goa^{(n)}(\tau)$ are $L^{-d}$-close:
	\begin{proposition}\lbl{l:a-goa} 
		For any  $n\geq 0$,
		\be\lbl{a-goa}
		\EE|a_s^{(n)}(\tau)-\goa_s^{(n)}(\tau)|^2\leq L^{-2d}C^\#(s;n), \qquad s\in\Z^d_L,
		\ee
		uniformly in $\tau\geq -T$ and $\nu$.
	\end{proposition}

	Since $\EE|a_s^{(n)}(\tau)|^2,\,\EE|\goa_s^{(n)}(\tau)|^2\leq C^\#(s;n)$ for any $\tau,\,n$ (this follows from Theorem~\ref{l:est_a^ia^j} jointly 
	  with \eqref{a-goa}), Proposition~\ref{l:a-goa} and the Cauchy  inequality imply
	\begin{corollary}\lbl{c:a-goa}
		For any $m,n\geq 0$,
		\be\lbl{corr:a-goa}
		\big|\EE a_s^{(n)}(\tau_1)\bar a_s^{(m)}(\tau_2)-\EE\goa_s^{(n)}(\tau_1)\bar \goa_s^{(m)}(\tau_2)\big|\leq L^{-d}C^\#(s;m,n), \qquad s\in\Z^d_L,
		\ee
		uniformly in  $\tau_1,\tau_2$ and $\nu$.
	\end{corollary}
	
	 We define a quasisolution $A(\tau)$ of equation \eqref{ku4}  as in definition \eqref{A-intro},  where the processes $\goa^{(n)}$ are replaced by $a^{(n)}$, and define its energy spectrum as $n_s(\tau)=\EE |A_s(\tau)|^2$, cf. \eqref{n_s-a}.  Due to  Corollary~\ref{c:a-goa} 
	   it suffices to prove the results, formulated in the introduction, replacing in their  statements processes  $\goa^{(k)}(\tau)$ and $\mathfrak{n} (\tau)$ with   $a^{(k)}(\tau)$ and $n(\tau)$. Indeed,  for Theorem A this assertion holds 
	  since $d\geq 2$. For Theorem C it follows from \eqref{assumption}  as
	$|n_s-\mathfrak{n}_s|\leq \rho^4L^{-d}C^\#(s)\leq \nu^{-2}L^{-d}C^\#(s)$, where we recall that $\rho$ is given by \eqref{r_scale}. 
	 For other results the argument is similar.

	Below we work with the processes $a^{(j)}(\tau)$ and  $n(\tau)$,  and never with   $\goa^{(j)}(\tau)$ and  $\mathfrak{n}(\tau)$.

	\subsection{Nonlinearity $\cY_s$}
	The cubic nonlinearity $\cY$ in \eqref{ku4} 
	 defines the 3-linear over real numbers  mapping
	 $(u, v,w)\mapsto \cY(u, v,w;t)$, where 
	$$
	\cY_s(u,v,w; t) = L^{-d} \sum_{s_1, s_2}   \delta'^{1 2}_{3 s}\, u_{s_1} v_{s_2} \bar w_{s_3}
	 e^{it \omega^{12}_{3s}}\,,
	$$
	so $\cY_s(a;t) = \cY_s(a,a,a;t)$. 
	Often it will be  better to use the  symmetrisation 
	\be
	\cY_s^{sym}(u,v,w;t) =\frac{ L^{-d}}3\sum_{s_1, s_2}   \delta'^{1 2}_{3 s} 
\big(	u_{s_1} v_{s_2} \bar w_{s_3}+v_{s_1} w_{s_2} \bar u_{s_3}+w_{s_1} u_{s_2} \bar v_{s_3}\big)\, e^{it \omega^{12}_{3s}} \,.
	\ee
	Clearly \ $ \cY^{sym}(v,v,v) = \cY(v)$.  Besides, 
	\begin{equation*}
	\begin{split}
	\cY_s&(v^0+ \delta v^1+\delta^2 v^2) =\cY_s(v^0) + 3\delta\  \cY_s^{sym}(v^0,v^0,v^1)\\
	&+  3\delta^2\Big(  \cY_s^{sym}(v^0,v^1,v^1) +   \cY_s^{sym}(v^0,v^0,v^2)\Big) +O(\delta^3). 
	\end{split}
	\end{equation*}
	\smallskip
	
\subsection{Correlations of first terms in the formal decomposition}
	
	 Let us go back to the formal decomposition \eqref{decomp}.
Correlations of the processes $a^{(n)}_s(\tau)$ are important for what follows. Since $a^{(0)}=\goa^{(0)}$, then  by \eqref{a0} 
for any $\tau_1,\tau_2$,
\be\label{corr_a_in_time}
	\begin{split}
	& \EE a_s^{(0)}(\tau_1)  a_{s'}^{(0)}(\tau_2) \equiv  \EE \bar a_s^{(0)}(\tau_1) \bar a_{s'}^{(0)}(\tau_2)\equiv 0,
	\\ 
	& \EE a_s^{(0)}(\tau_1) \bar a_{s'}^{(0)}(\tau_2)
	= \delta^s_{s'} \,  \frac{b(s)^2}{\gamma_s} \big( e^{-\gamma_s|\tau_1-\tau_2|} -
	e^{-\ga_s(2T + \tau_1+ \tau_2)}\big).
	\end{split}
	\ee
	Indeed, the first relations are obvious. To prove the last we may  assume that 
	 $-T\le \tau_1\le\tau_2$. Then the l.h.s. vanishes if $s\ne s'$, while for $s=s'$ it equals
	$$
	 b(s)^2 \EE \Big(  \int_{-T}^{\tau_1} e^{-\gamma_s(\tau_1-l_1)}d\beta_s(l_1)\Big)
	\Big(  \int_{-T}^{\tau_1} e^{-\gamma_s(\tau_2-l_2)}d\bar\beta_s(l_2)
	+ \int_{\tau_1}^{\tau_2} e^{-\gamma_s(\tau_2-l_2)}d\bar\beta_s(l_2)
	\Big).
	$$
	The expectation of the second term  vanishes, and that of the first equals 
	$
	2  \, b(s)^2   \int_{-T}^{\tau_1} e^{-\gamma_s(\tau_1-l +\tau_2-l) }dl,
	$
which is the  r.h.s.
	of \eqref{corr_a_in_time} (recall that $d\beta\cdot d\bar\beta = 2dt$).

For the process $a^{(1)}$ we have 
\begin{lemma}\label{l_a1_1}
	For any $\tau_1, \tau_2$ and $s', s''$,
	
	\noindent
	i) $\EE a^{(1)}_{s'} (\tau_1) a^{(1)}_{s''}(\tau_2) =0$;\\
	ii) $\EE a^{(1)}_{s'}(\tau_1) \bar a^{(1)}_{s''}(\tau_2)=0$  if $s'\ne s''$;\\
	iii) $\EE  a^{(1)}_{s'}(\tau_1) a^{(0)}_{s''}(\tau_2)=\EE a^{(1)}_{s'}(\tau_1) \bar a^{(0)}_{s''}(\tau_2)=0$.
\end{lemma}
\begin{proof}
	Let us verify  ii). Due to \eqref{a1}, the expectation we examine is a sum over $s_1,s_2$ and $s_1',s_2'$ of  integrals of functions 
	$$
	\EE\Big(\delta'^{12}_{3s'} (a^{(0)}_1 a^{(0)}_2{\bar a}^{(0)}_3)(l) \,\delta'^{1'2'}_{3's''} (\bar a^{(0)}_{1'} \bar a^{(0)}_{2'}{ a}^{(0)}_{3'})(l')\Big),
	$$
	multiplied by some density-functions.
	Since each $a^{(0)}_j$ is a Gaussian process, then Wick's theorem applies. By \eqref{corr_a_in_time} and due to the
	the factor $\dess$,
	$a^{(0)}_1$ must be coupled with $\bar a^{(0)}_{1'}$ or with $\bar a^{(0)}_{2'}$,  $a^{(0)}_2$ -- with $\bar a^{(0)}_{2'}$ or $\bar a^{(0)}_{1'}$,
	and $\bar a^{(0)}_3$ -- with ${a}^{(0)}_{3'}$.  So $s'= s_1+s_2-s_3= s'_1+s'_2-s'_3=s''$ and ii) follows. The proof of i) and iii)
	is similar.
\end{proof}

By a similar argument it can be shown that $\EE a_{s'}^{(m)}a_{s''}^{(n)}=0$ for any $m,n$ and $s',s''$, while
$\EE a_{s'}^{(m)}\bar a_{s''}^{(n)}=0$ for any $m,n$ and $s'\neq s''$. 
 Moreover, $\Re\EE \goa^{(1)}_{s} \bar \goa^{(0)}_{s}=0$ (while Corollary~\ref{c:a-goa} implies only that this expectation is of the size $L^{-d}$), so, as it is claimed in the introduction, $\fn_s^{(1)}=0$.
We do not prove and do not use this observations.
	
\subsection{Second moments $\EE a^{(1)}_{s}{(\tau)} \bar a^{(1)}_{s}(\tau)$} 
\lbl{s:2nd moments}

In the center of attention of WT are the limiting, as $L\to\infty$, $\nu\to0$, correlations of solutions $a_s(\tau)$ (and  $v_s(\tau)$). 
Accordingly we should analyse limiting correlations of the processes $a^{(n)}_s(\tau)$. To give an idea what we should expect there, 
let us consider the second moments of the process $a^{(1)}_s(\tau)$.
 The tools, needed for this analysis, will be systematically used later.

We have
	\begin{equation*}
	\begin{split}
	\EE& |a^{(1)}_{s}{(\tau)}|^2  =L^{-2d} \int_{-T}^\tau dl_1  \int_{-T}^\tau dl_2 \, e^{\ga_s(l_1+l_2-2\tau)} \\
	&\times\sum_{1, 2}\sum_{1',2'} \EE \Big(\dess\delta'^{1'2'}_{3's} (a^{(0)}_1a^{(0)}_2\bar a^{(0)}_3)(l_1)  (\bar a^{(0)}_{1'}\bar a^{(0)}_{2'} a^{(0)}_{3'})(l_2)
	e^{i\nu^{-1}( l_1\oms -l_2\omega^{1'2'}_{3's} ) }\Big).
	\end{split}
	\end{equation*}
	By the Wick theorem
	$$
	\sum_{s_{1'}, s_{2'}}
	\EE\Big(\delta'^{12}_{3s} (a^{(0)}_1 a^{(0)}_2{\bar a}^{(0)}_3)(l) \,\delta'^{1'2'}_{3's} (\bar a^{(0)}_{1'} \bar a^{(0)}_{2'}{ a}^{(0)}_{3'})(l')\Big)=
	2 \prod_{j=1}^3 \EE a^{(0)}_j(l)  \bar a^{(0)}_j(l'),
	$$
	where we used that $\EE \bar a^{(0)}_j(l)  a^{(0)}_j(l')  = \EE a^{(0)}_j(l)  \bar a^{(0)}_j(l')$. 
	Then, recalling the notation $\ssum$ introduced in Section~\ref{s:notation},  by \eqref{corr_a_in_time} we get 
	\begin{equation*} %\label{a1_covar1}
	\begin{split}
	\EE |a^{(1)}_{s}{(\tau)}|^2
	&=2\ssum_{1,2}\dess \int_{-T}^\tau dl_1  \int_{-T}^\tau dl_2 \, B_{123} \\
	&\times \prod_{j=1}^3 \big( e^{-\ga_j|l_1-l_2|} -  e^{-\ga_j(2T+ l_1+ l_2)}\big)
	 e^{ \ga_s (l_1+l_2-2\tau) +i\nu^{-1}\oms(l_1-l_2)  }\,,
	\end{split}
	\end{equation*}
	where we denoted
	\be\lbl{B_123 def}
	B_{123} = B_1 B_2 B_3, \qquad 
	B_r =  b(s_r)^2/\ga_{s_r}\quad\text{for}\; r=1,2,3.
	\ee
	
	 To simplify the computations, we first assume that $T=\infty$.
In  this case $\EE |a^{(1)}_{s}{(\tau)}|^2$ does not depend on $\tau$ and equals 
	 $\Sigma_s$, where 
	\begin{equation}\label{a1_covar1}
	\begin{split}
	\Sigma_s
	=2\ssum_{1,2}\dess \int_{-\infty}^0 dl_1  \int_{-\infty}^0 dl_2  \,
	  B_{123} e^{-|l_1-l_2|(\ga_1+\ga_2+\ga_3) +\ga_s (l_1+l_2) +i\nu^{-1}\oms(l_1-l_2)  }\,.
	\end{split}
	\end{equation}
	Since for $a, b>0, c\in\R$ we have
	\be\non
	\int_{-\infty}^0 dl_1  \int_{-\infty}^0 dl_2 e^{-a|l_1-l_2| +b(l_1+l_2) +ic (l_1-l_2)} = \frac{a+b}{ b\big((a+b)^2+c^2\big)}\,,
	\ee
	then 
	\begin{equation}\label{a1_covar}
	\begin{split}
	\Sigma_s
	&=\frac{2\nu^2}{\ga_s}
	\ssum_{1,2}\dess  B_{123}  \,\frac{\ga_1+\ga_2+\ga_3+\ga_s}{(\oms)^2  +\nu^2 (\ga_1+\ga_2+\ga_3+\ga_s)^2}\,.
	\end{split}
	\end{equation}
	
	For the reason of equality \eqref{a1_covar}, below we call  expressions like those in the r.h.s. of \eqref{a1_covar1} ``sums", meaning that
	they become sums 	after the explicit integrating over $dl_j$.

	%%%%%%%%%%%%%%%%%
	%%%%%%%%%%%%%%%%%%
	%%%%%%%%%%%%%%%%%%%%
	%%%%%%%%%%%%%%%%%%%%%

	\section{Limiting behaviour of   second moments}		
\label{s_assymt}

In this section we study the asymptotic behaviour of the sum  $\Sigma_s $ (see \eqref{a1_covar1}=\eqref{a1_covar}) 
as $L\to\infty$ and $\nu\to 0$,
assuming that
$
L\gg \nu^{-1}\gg1.
$
The latter inequality will be specified later.

	\subsection{Approximation of the sums  $\Sigma_s $ by  integrals 
	}
	\label{s_sing_integral}
	Let us  naturally extend 
	 $\ga_s= \ga^0(|s|^2)$ to a function on $\R^d$ and denote
	$$
	\ga_1+\ga_2+\ga_3+\ga_s =: \Gamma(s_1, s_2, s_3,s), \qquad s_j,s\in\R^d.
	$$
	We also extend $B_s,\, s\in \Z^d_L$, to the  function
	$\ 
	B(s) = {b}(s)^2/\ga_s, \qquad s\in\R^d,
	$
	and extend $B_{123}$ to 
	$
	B(s_1, s_2, s_3):=  B(s_1)B( s_2) B(s_3)$, 
	$s_j\in \R^d$.
	Recalling \eqref{omega} we see that
	\be\label{Ss_extends}
	\text{
	 $\Sigma_s$  naturally extends to a function on $\R^d\ni s$,}
	\ee
	 both in the form  \eqref{a1_covar1} and \eqref{a1_covar}. Doing that we  understand the factor $\dess$ as
	the rule ``substitute $s_3= s_1+s_2-s\,$". In this case in  \eqref{a1_covar1} and \eqref{a1_covar} the indices $s_1, s_2$
	belong to $\Z_L^d$, while $s$ and $s_3$ are vectors in $\R^d$. 
	Here and in similar situations below we	will keep denoting the extended functions by
	the same letters. 
	
	Considering  \eqref{a1_covar} with $s\in \R^d$ we may replace there the sum $\ssum_{1,2}$ by the integral over $\R^d\times \R^d$,
	thus getting the function $I_s$, defined as 
	\be\label{I_s}
	\begin{split}
		I_s = \frac{\nu^2}{2\ga_s}
		\int_{\R^d\times\R^d}  ds_1\,ds_2\, \,
		\frac{
		\des\,  B(s_1, s_2, s_3)\Gamma(s_1, s_2, s_3,s)}{ ((s_1-s)\cdot 
		(s_2-s))^2 +(\frac12\, \nu \Gamma(s_1, s_2, s_3,s))^2}\\
		=: \nu^2 \int_{\R^d\times\R^d}  ds_1\,ds_2\, \,  
		 \frac{ \des\,  F_s(s_1,s_2)}{ ((s_1-s)\cdot (s_2-s))^2 +(\frac12\, \nu \Gamma(s_1, s_2, s_3,s))^2}\,.
	\end{split}
	\ee
	Here
	$ \des $ is the Kronecker delta of the relation $s_1+s_2=s_3+s$, so  the factor $\des$ in the integrands 
	 means that there  $s_3:=s_1+s_2-s$. By $F_s$ we denoted the positive Schwartz function on $\R^{3d}$	 
	\be\label{F}
	F_s(s_1,s_2) = \Gamma(s_1,s_2,s_3,s) B(s_1,s_2,s_3) /2\ga_s.
	\ee
 Reverting the transformation, used to get
\eqref{a1_covar} from \eqref{a1_covar1} we find that 
\begin{equation}\label{I_s1}
	\begin{split}
	I_s= 2\int_{\R^d\times \R^d}ds_1  ds_2& \int_{-T}^0 dl_1  \int_{-T}^0 dl_2\, \des\,
  B(s_1,s_2, s_3) \\
& \times e^{-|l_1-l_2|(\ga_1+\ga_2+\ga_3) +\ga_s (l_1+l_2) +i\nu^{-1}\oms(l_1-l_2)  }\,.
	\end{split}
	\end{equation}
	
	%%%%%%%%%%%%%%%%%%%%%%%

	 Our goal in this section 
	is to estimate the difference between the sum $\Sigma_s$  and the integral $I_s$, while in the next section we will study asymptotical behaviour of the latter as 
	$\nu\to0$. Moreover, we  consider a bit more general sums,  needed   in  the work \cite{DK}.

We will study sums with the summation index $(s_1, \dots, s_k)=:z \in \Z_L^{kd}$, $k\ge1$,
and integrals  with the integrating 
variable $z=(s_1, \dots, s_k) \in \R^{kd}$. 
 For $s\in \R^d$ let us consider the union
$	D_s=\cup_{j=1}^p D_s^j$ of $p\geq 0$ affine subspaces $D_s^j$ in $\R^{kd}$ (if $p=0$ then $D_s$ is empty), 
\be\non
 D_s^j=\big\{z=(s_1, \dots s_k) \in\R^{kd}:\, 
	c_0^j s + \sum_{i=1}^k c_i^j s_i=0\in\R^{d}\big\}, 
	\ee 
	where $c_i^j$ are some real numbers.
	We assume that the $k$-vectors $(c_1^j,\ldots,c_k^j)\ne 0\in\R^k$ for all $j$,
	so the subspaces $D_s^j$ have dimension $d(k-1)$ for every $s$.
Consider the sum/integral
	$$
	S_s = \int_{\R^m} \ssum_{\substack{z\in \Z^{kd}_L\setminus D_s}} G_s(z,\theta;\nu) \,d\theta, \qquad m\in \N\cup\{0\}, 
	\quad s\in \R^d, 
	$$ 
(if $m=0$, then there is no 
 integration  over $\R^m$), where $G_{s}$ is a measurable function  of $(z,\theta,\nu)\in \R^{kd}\times \R^m\times (0, 1/2]$, 
 $C^2$-- smooth  in  $z$  and satisfying 
\be\label{op1}
\big| \p^{\alpha}_{z}G_s(z,\theta;\nu)\big| \le \nu^{-|\alpha|} C^\#(s)C^\#(z) C^\#(\theta)\qquad \text{if} \quad 0\leq |\alpha| \le2,
\ee
for all values of the arguments. 
Our goal is to compare $S_s$ with the integral
$$
J_s =  \int_{\R^m} \int_{\R^{kd}} G_s(z,\theta;\nu)\,dz d\theta, \qquad s\in \R^d.
$$
\begin{theorem}\label{t_sum_integral}
Under the assumption \eqref{op1}, 
\be\label{op2}
		| S_s -J_s | \le C^\#(s)\nu^{-2} L^{-2}, \qquad s\in \R^d.
		\ee
	\end{theorem}
Theorem \ref{t_sum_integral} applies to the sum $\Sigma_s$ where, due to the factor $\de'$,  we take the summation over $s_1,s_2\neq s$.
Indeed,  in this case $D_s=D_s^1\cup D_s^2$ with $D_s^j=\{(s_1,s_2):\,s-s_j=0\}$.

\begin{proof}
		Denote by $\hat S_s$ the sum $S_s$, where $D_s$ is replaced by the empty set. 
	Firstly we claim that
	\be\lbl{no D_s}
	|S_s-\hat S_s|\leq C^\#(s)L^{-d}\leq C^\#(s)L^{-2},
	\ee 
	where in the last inequality we used that $d\geq 2$.
	Indeed, due to \eqref{op1} with $\al=0$, 
	$|S_s-\hat S_s|\leq C^\#(s)L^{-kd}
	\sum_{z\in  D_s\cap\Z^{kd}_L} C^\#(z),$ and the claim follows from the fact that the affine subspaces $D_s^j$ have dimension  $d(k-1)$. 

Now we consider a  mesh in $\R^d$, formed by cubes of size $L^{-1}$, centred at the points
	of the lattice $\Z_L^d$. For any $l\in \Z_L^d$ denote by $m(l)$ the cell of the mesh with the centre in $l$, and consider
	the measurable 	mapping
	$$
	\Pi:\R^d\mapsto \Z^d_L\,,\quad \Pi(x) =
	\left\{\begin{array}{ll}
	l, & \text{if $x\in$ interior of $m(l),\; l\in\Z^d_L$} \,,
	\\
	0, & \text{if $x\in\p m(l')$ for some $l'\in \Z^d_L$.
	}
	\end{array}\right.
	$$
	Let $\Pi^k=\Pi\times\ldots\times\Pi$, where the product is taken $k$ times.
Then 
$$
\hat S_s=  \int_{\R^m} \int_{\R^{kd}} G_s(z,\theta;\nu)\circ \big(\Pi^k \times\id\big)\,
dz d\theta.
$$
% Due to \eqref{op1} with $\alpha=0$, the integral above  will change at most by 
%$C^\#(s)L^{-d}$ if we make the domain of integration in $dz$ equal $\R^{kd}$.  Thus, 
Setting  $G_s^\Delta =G_s -G_s\circ(\Pi^k\times \id)$, we
see that 
$$
J_s - \hat S_s = \int_{\R^m} \int_{\R^{kd}}  G^\Delta_s(z,\theta;\nu)\,dz  d\theta. 
$$
	For any 
	$\bs=(l_1,\ldots,l_k)\in \Z^{dk}_L$ let us  denote $\bm(\bs) = m(l_1)\times\ldots\times m(l_k)$, restrict 
	 $G_s^\Delta$ to the cell
	$\bm(\bs)$ and write it as
	$$
	G_s^\Delta(z,\theta) = \p_{z}G_s(\bs ,\theta)\cdot (z-\bs)+ T_s(z,\theta), \quad z \in \bm(\bs).
	$$
	Here 
	$
	| T_s(z,\theta)  | \le C L^{-2} \sup_{\xi\in \bm(\bs)} |\p_\xi^2G_s(\xi,\theta) | 	$
	for $z\in \bm(\bs)$. 
	Then 
	$$
	\Big| \int_{\bm(\bs)}  G_s^\Delta(z,\theta) %-\tilde  F_s( \bs,\theta)) 
	dz\Big| \le  L^{-kd} L^{-2}\nu^{-2} C^\#(s) C^\#(\bs) C^\#(\theta)    ,
	$$
	and accordingly 
	\begin{equation}\lbl{d ne 1}
	| J_s - \hat S_s | \le \nu^{-2} L^{-2}  C^\#(s)   \int_{\R^m} \int_{\R^{kd}}    C^\#(z) C^\#(\theta)  \,dz d\theta \\
	\le C_1^\#(s)\nu^{-2} L^{-2}.  
	\end{equation}
	The theorem is proved. 
\end{proof}

\begin{remark}
	Theorem~\ref{t_sum_integral} remains true for $d=1$ if we replace  \eqref{op2} by the weaker estimate
	$$
	| S_s -J_s | \le C^\#(s)(\nu^{-2} L^{-2} + L^{-1}), \qquad s\in \R^d.
	$$
	Indeed, in the proof of the theorem  relation $d\geq 2$ was used only once, to get the second inequality in \eqref{no D_s}.
\end{remark}

\subsection{Limiting behaviour of the integrals $I_s$.}

Here we study the integral $I_s$, written in the form \eqref{I_s}, when  $\nu\to0$. 
With the notation $\Ga_s(s_1,s_2)= \tfrac12\Ga(s_1,s_2, s_1+s_2-s,s)$ the integral takes the form
\be\non
	\begin{split}
		I_s =
		\nu^2 \int_{\R^{2d}}  ds_1\, ds_2\, \,  
		 \frac{{F}_s(s_1, s_2)}{ ((s_1-s)\cdot (s_2-s))^2 +( \nu \Gamma_s(s_1, s_2))^2}.\\
		 %= \nu^2 \int_{\R^{2d}}  ds_1\, ds_2\, \,  
		% \frac{{F}_s(s_1, s_2)}{ ( \oms /2)^2 +( \nu \Gamma_s(s_1, s_2))^2},
	\end{split}
	\ee
	We study its asymptotical behaviour when $\nu\to 0$ in an abstract setting and do not use the explicit forms of the functions $F_s$ and $\Ga_s$. Instead we assume that 
they are  $C^2$--smooth and $C^3$--smooth  real functions,  correspondingly, 
 satisfying certain restrictions on their behaviour at 
infinity. Namely, it  suffices to assume that 
\be\label{B_condition}
|\p_{s_1, s_2,s}^\al{F}_s(s_1, s_2)|\le C_1^\#(s, s_1, s_2) \qquad \forall\,s_1, s_2, s\in\R^d,
\quad \forall\, |\al|\le 2,
\ee
and, for some real numbers $r_1\ge0$ and $K>0$, 
\be\label{Ga_condition}
|\Ga_s(s_1, s_2)|\ge K^{-1}, \quad  | \p_{s_1, s_2,s}^\al\Ga_s(s_1, s_2)|\le K \lan(s_1, s_2, s)\ran^{r_1-|\al|}\;\; 
\ee
for all $ s_1, s_2, s$ and all $|\al|\le 3$. 

Note that 
function $F_s$ from \eqref{F} satisfies \eqref{B_condition}, while function
$
\Ga_s= \frac12\Ga(s_1, s_2, s_1+s_2-s, s)
$
satisfies \eqref{Ga_condition} with $r_1=\max(2r_*, 3) $, 
as well as the functions 
 %$\Ga_s =\lan s\ran^{2r_*}$ and  
$
\Ga_s(s_1, s_2) =\ga^0(|s_i|^2),
$
  $i=1,2,3,4$, where $s_3=s_1+s_2-s$ and $s_4=s$.

In  the theorem below  we denote by  ${}^s\!\Sigma $ the quadric 
\be\label{sigmaS}
{}^s\!\Sigma =
\{(s_1,s_2):\oms=0\}, \qquad s_3=s_1+s_2-s.
\ee
The latter has a singularity at the point $(s,s)$, and we 
denote by  ${}^s\!\Sigma_*$ its smooth part 
${}^s\!\Sigma_*={}^s\!\Sigma\setminus \{(s,s)\}$.

\begin{theorem}\label{t_singint}
	As $\nu\to0$, the  integral $ I_s$, $s\in\R^d$, may be written as 
	\be\label{asymp_exp}
	I_s  = \nu I_s^0+  \nu^2I_s^\Delta, 
	\ee
	\be\label{p22}
	I_s^0 = \pi   \int_{{}^s\!\Sigma_*} \frac{{F}_s(s_1, s_2)}{\sqrt{|s-s_1|^2 +|s-s_2|^2}\ 
		\G_s(s_1, s_2) }\,ds_1\,ds_2\!\mid_{{}^s\!\Sigma_*}.
	%+ \nu^2 I_s^\Delta,
	\ee
	Here  
	%${}^s\!\Sigma_*=\{ (s_1,s_2)\in\R^{2d}\setminus \{(s,s)\}: (s_1-s)\cdot(s_2-s)=0\}$, \ 
	$ds_1ds_2\!\!\mid_{{}^s\!\Sigma_*}$ is the volume element on ${}^s\!\Sigma_*$, induced from $\R^{2d}$, 
	and the integral for $I_s^0$ converges absolutely.
	The functions $I_s^0$ and $I_s^\Delta$ satisfy the estimates 
	\be\non
	| I_s^0| \le  C^\#(s)\,,\qquad 
	|I_s^\Delta|\le 
		C^\#(s;\aleph_d) \, \nu^{-\aleph_d}, 
	\ee
	where $\aleph_d=0$ if $d\ge3$, while for $d=2$, $\aleph_d$ is any positive number.
\end{theorem}

The theorem is proved in Section \ref{s_singint}.
If $F_s$ and $\Gamma_s$ do not depend on $s$, then the theorem holds under related (but milder) restrictions on $F$
and $\Gamma$, and in that case 	$|I_s^\Delta|\le C \chi_d(\nu)$, where $\chi_d$ is defined in \eqref{chi_d},
 see  \cite{K}. 

Theorem \ref{t_singint} implies that 
\be\label{int_est}
|I_s| \le \nu C^\#(s).
\ee
In Appendix \ref{a_stat_phase}
 we show that this inequality may be obtained easier and under weaker restrictions on the functions $F_s$ and $\Ga_s$. This
observation is important since later in the text we use various generalisations of  inequality \eqref{int_est} in situations, where analogies 
of the asymptotic expansion \eqref{asymp_exp} are not known for us. 
	
Applying Theorems \ref{t_sum_integral} and \ref{t_singint} to the sum  $ \EE |a^{(1)}_{s}{(\tau)}|^2 =\Sigma_s$ in \eqref{a1_covar} and 
 recalling that $\Sigma_s$ was extended to a function on $\R^d$,  we get

	\begin{corollary}\label{t_for_sum}   Assume $T=\infty$. Then for $s\in \R^d$, 
	\begin{equation*}
\begin{split}
\Big| \EE |a^{(1)}_{s}{(\tau)}|^2  - \nu \frac{\pi }{\ga_s} \int_{{}^s\!\Sigma_*} %\frac1{\sqrt{|x|^2+|y|^2}} 
		\frac{B(s_1, s_2, s_1+s_2-s) }{{\sqrt{|s_1-s|^2+|s_2-s|^2}} \ 
		%\gamma_{s_1}\gamma_{s_2}\gamma_{s_1+s_2-s}
		}
		\, ds_1ds_2\!\!\mid_{{}^s\!\Sigma_*}\Big|\\
	\le  C^\#(s;\aleph_d) (L^{-2} \nu^{-2} +
	\nu^{2-\aleph_d} ). 
\end{split}
\end{equation*}
	\end{corollary}

It is convenient to pass in  \eqref{p22} from the variables $(s_1, s_2,s)$ 
to 
\be\label{p01}
(x,y,s) = (s_1-s, s_2-s, s),\qquad (x,y)=:z.
\ee
 Then  the 
quadric  ${}^s\!\Sigma$ becomes 
\be\non
 \Sigma=\{z: x\cdot y=0\}\subset \R^d_x\times \R^d_y =\R^{2d}_z.
\ee
 The locus of $\Sigma$ 
  is the point $(0,0)$, and the regular part is  $\Sigma_*=\{(x,y)\ne0: x\cdot y=0\}$. Now we write the integral $I_s^0$  as 
\be\label{new_int}
\int_{\Sigma_*} f(z) |z|^{-1}\,dz\mid_{\Sigma_*},
\ee
where $f=\pi F_s/\Gamma_s$.

Integrals of the form \eqref{new_int} 
are important for what follows. In next section we discuss some their properties.

\subsection{Integrals  \eqref{new_int} } 

Let us extend the measure $ |z|^{-1}\,dz\mid_{\Sigma_*}$ to a measure $\mu^\Sigma$ on $\Sigma$, where  $\mu^\Sigma(\{0\})=0$, and next extend
$\mu^\Sigma$ to  a Borel measure on $\R^{2d}$, supported by $\Sigma$, keeping for the latter the same name.   Then the 
 integrals \eqref{new_int} may be written as $\int_{\Sigma_*} f(z)\, \mu^\Sigma (dz)$, or as  $\int_\Sigma f(z)\, \mu^\Sigma (dz)$, or 
 as $\int_{\R^{2d}} f(z)\, \mu^\Sigma (dz)$.

For any real number $r$ let  $\Cc_r(\R^{2d})$ be the space of continuous complex 
functions on 
   $\R^{2d}$ with the    finite norm  
  \be\label{norm_m}
   |f|_r = \sup_z|f(z)| \lan z\ran^r.
   \ee
      
    \begin{proposition}\label{p_prop} 
 Integral \eqref{new_int} as a function of $f$  
 defines a continuous linear functional on the space $\Cc_r(\R^{2d})$ if $r>2d-2$.
 \end{proposition}
 
 The proposition is proved in Section \ref{s_proofprop}.  
 To study further the measure $\mu^\Sigma$ we consider  the  projection 
\be\label{Pi}
\Pi: \Sigma_* \to \R^d_x,\quad z=(x,y)\mapsto x.
\ee
It defines a fibering of $ \Sigma_*$, 
where the fiber $\Pi^{-1}0 = \{0\}\times \{\R^d_y\setminus \{0\}\}$ is singular, while for any non-zero $x$ the fiber $\Pi^{-1}x$ equals $ \{x\}\times x^\perp $, where $x^\perp$ is the orthogonal complement to $x$ in $\R^d_y$. 
So the restriction of $\Pi$ to  the domain $\ \Sigma^x= \Sigma \setminus (\{0\}\times \R^d_y)$ is  a smooth euclidean vector bundle over $\R^d \setminus \{0\}$.

 Let us us abbreviate to  $\mu$ the volume element   $ dz\mid_{\Sigma_*}$. Since the measure $\mu^\Sigma\mid_{\Sigma_*}$ 
  is absolutely continuous with respect to $\mu$, then  $\mu^\Sigma( \Sigma_* \setminus \Sigma^x)=  \mu( \Sigma_* \setminus \Sigma^x)=0$; so  to calculate the integrals \eqref{new_int} 
it suffices to know  the restriction of   $\mu^\Sigma$ to $\Sigma^x$.
The result below shows how to disintegrate the measures $\mu\mid_{\Sigma^x}$  and $\mu^\Sigma\mid_{\Sigma^x}$ 
with respect to $\Pi$, and allows to integrate explicitly over them.

\begin{theorem}\label{t_disintegr} The measures $\mu\mid_{\Sigma^x}$  and $\mu^\Sigma\mid_{\Sigma^x}$  disintegrate as follows:
\be\label{des1}
\mu(dz) = |x|^{-1} dx  |z| \,d_{x^\perp} y,
\ee
\be\label{desSigma}
\mu^\Sigma(dz) = |x|^{-1}  dx\,d_{x^\perp} y,
\ee
where $d_{x^\perp} y$ is the volume element on the space ${x^\perp} $ (the orthogonal complement to $x$ in $\R^d_y$). 
\end{theorem}

We recall that equality \eqref{des1} means that for any continuous function  $f$   on $\Sigma^x$  with compact support
\be\label{recall}
\int_{\Sigma^x} f(z)\mu(dz) =\int_{\R^d\setminus \{0\}}  |x|^{-1}   dx\int_{x^\perp}  f(z) \, |z|\,d_{x^\perp} y. 
\ee

\begin{proof}
It suffices to verify \eqref{recall} for all continuous functions $f$, supported by a compact set $K$, for any 
 $K\Subset (\R^d\setminus\{0\}) \times \R^d$. 
For $x'\in \R^d\setminus\{0\}$ and $m\in\N$ we denote $r'=|x'|>0$ and set
$
U_{x'} = \{x: |x-x' |<\tfrac12 r'\}$ and $
 U^m =\{y: |y|<m\}.
$
Since any $K$ as above can be covered by a finite system of domains $U_{x'} \times U^m$, it suffices to prove \eqref{recall} for 
 any set 
  $U_{x'} \times U^m =: U$ and any 
   $f\in C_0(U)$, where $C_0(U)$ is the space of continuous compactly supported functions on $U$.  

Now we construct explicitly a trivialisation of the linear bundle $\Pi$ over $U_{x'} $. To do this we fix in $\R^d$ a coordinate system such that
\be\label{holds}
x_1\ge\kappa>0\quad\text{for any}\quad x=(x_1,x_2,\dots,x_d) =: (x_1, \bar x)\in  U_{x'} .
\ee
  We denote $Q= \Pi^{-1}U_{x'} \subset \Sigma^x $
 and construct  a linear in the  second argument $\bar\eta$
coordinate mapping 
$
\Phi: U_{x'}\times \R^{d-1} \to Q$
of the form
 $$
 \Phi(x, \bar\eta) = \big(x,  \Phi_x(\bar\eta)\big), 
 %\qquad \Phi_x: \R^{d-1} \to x^\perp, \;\; 
 \qquad  \Phi_x(\bar\eta) =  ({\phi}(x, \bar\eta), \bar\eta).
 $$
 The function $\phi$  should be  such  that $\Phi_x(\bar\eta)\in x^\perp$. That is, it should satisfy
 $x\cdot \Phi_x(
 \bar \eta)= x_1 {\phi} + \bar x\cdot\bar\eta=0$.
 From here we find that 
 $\ 
 {\phi} = -\frac{\bar x \cdot\bar\eta}{x_1} 
 $. 
 Thus obtained mapping $\Phi_x$ is linear in $\bar\eta$, and the image of $\Phi$  is the set $Q$. 
In the coordinates $(x, \bar\eta) \in U_{x'}\times\R^{d-1}$ the hypersurface $\Sigma^x$ is embedded in $\R^{2d}$ as a graph of the function $\phi(x, \bar\eta)$. 
Accordingly  in these  coordinates 
  the volume element $\mu$ on $\Sigma^x$ reeds $\mu =\bar p(x,\bar\eta) dx\,d\bar\eta ,$ 
where\,\footnote{Indeed, denoting $\xi =(x,\bar\eta)$ we see that the first fundamental form $I^\xi$ of $\Sigma^x$ is
given by $I^\xi_{ij} =(\delta_{i,j} +\theta_i\theta_j)$, where $\theta=\p_\xi \phi \in\R^{2d-1}$. So for $X\in \R^{2d-1}$, 
$$
I^\xi(X,X) = \sum X_j^2 + \sum_{i,j} X_i\theta_i X_j\theta_j = \sum X_j^2 +(X\cdot\theta)^2. 
$$
Choosing in $\R^{2d-1} $ a coordinate system with the first basis vector $\theta/|\theta|$ we find that $ I^\xi(X,X)  = X_1^2(1+|\theta|^2) + \sum_{j\ge2} X_j^2$. So 
det$I^\xi = 1+|\p_\xi \phi |^2$, which implies the formula for the density $\bar p$.
}
$$
%\mu = p(x,\bar\eta) dx\,d\bar\eta , \quad
 \bar p(x,\bar\eta)= \big(1+|\p {\phi} (x,\bar\eta) |^2\big)^{1/2}=
\Big( 1+ x_1^{-2}\big( (x_1^{-1}\bar x\cdot \bar\eta)^2 +|\bar\eta|^2 + | \bar x|^2 \big)
\Big)^{1/2} .
$$
So 
$$
\int_{U\subset Q} f(z) \mu(dz)  = \int_{U_{x'}} \Big( \int_{\R^{d-1}} f(x, \Phi_x(\bar\eta)) \,\bar p(x, \bar\eta)\,d\bar\eta\Big)dx. 
$$
Passing from the variable $\bar\eta$ to $y= \Phi_x(\bar\eta)\in x^\perp$ we write the measure $\bar p(x, \bar\eta)d\bar\eta $ as
$p(z) d_{x^\perp} y$ with 
\be\non
p(z) = p(x,y)=\bar p(x, \Phi_x^{-1}y) |\det \Phi_x|^{-1}.
\ee
Then 
\be\label{rela}
\int_{U} f(z) \mu(dz)  =  \int_{U_{x'}} \Big( \int_{U^m \cap x^{\perp}} f(z) p(z) d_{x^\perp} y \Big) dx.
\ee
The smooth function $p$ in the integral above is defined on $U \cap \Sigma^x$ 
in a unique way and does not depend on the trivialisation of $\Pi$ over 
${U_{x'}}$, used to obtain it. Indeed, if $p_1(z)$ is another smooth function on $U \cap \Sigma^x$
 such that \eqref{rela} holds with $p:= p_1$, 
 then 
$$
\int dx  \int_{x^{\perp}} f(z) (p(z)-p_1(z))  d_{x^\perp} y =0\qquad \forall\,f\in C_0(U),
$$
which obviously implies that $p=p_1$.
 To establish \eqref{recall} it remains to verify that 
in \eqref{rela}
\be\label{remains}
p(x_*, y_*) = |x_*|^{-1} |z_*| \qquad \forall\, z_*= (x_*, y_*) \in U\cap\Sigma^x. 
\ee
To prove this equality 
let us choose in $\R^d$ euclidean coordinates with the first basis vector 
 $e_1= x_*/R_*$, $ R_*=|x_*| \ge \tfrac12 r'$. In these coordinates  condition \eqref{holds} holds, $x_*=( R_*,0)$ and
   $y_*=(0,\bar\eta_*)$, $\bar\eta_*\in\R^{d-1}$. Then
$$
\bar p(x_*, \Phi_{x_*}^{-1} y_*) = \bar p(x_*,   \bar\eta_*) =
 R_*^{-1} ( R_*^2+ | \bar\eta_*|^2)^{1/2} = |x_*|^{-1} |z_*|, 
$$
and \eqref{remains} follows since  $\det \Phi_{x_*}=1$.
This proves \eqref{des1}. Relation 
\eqref{desSigma} follows from  \eqref{des1} and the definition of the measure $\mu^\Sigma$. 
\end{proof}

Considering the projection $(x,y)\mapsto y$ instead of \eqref{Pi} we see that the measure 
$\mu^\Sigma$, restricted to  the domain $\Sigma^y=  \{(x,y)\in\Sigma: y\ne0\}$, disintegrates as
\be\label{Pii}
\mu^\Sigma\mid_{\Sigma^y}= 
dy\, |y|^{-1}  d_{y^\perp} x, \quad y\in \R^d\setminus\{0\}.
\ee

\begin{example}
Let us calculate the $\mu^\Sigma$--volumes of the balls $B_R^{2d} =\{|z|\le R\}$. Denoting by $A_n$ (by $V_n$) the area of the unit sphere
(the volume of the unit ball) in $\R^n$, % so that $A_n=nV_n$, 
 we have
$$
\mu^\Sigma(B_R^{2d} ) =\! \int_{|x|\le R} \frac1{|x|} \!\int_{ |y|\le \sqrt{R^2-|x|^2}}1\, dy=
A_dV_{d-1}\! \int_0^R r^{d-2} (R^2-r^2)^{(d-1)/2} dr.
$$
If $d=3$ this equals $ A_3 V_2 \tfrac14 R^4 = \pi^2 R^4$. If $d=2$, this equals %$A_2V_1 \int_0^R (R^2-r^2)^{1/2} dr= \pi^2R^2$. 
 $A_2V_1 \int_0^R \sqrt{R^2-r^2} dr$ $= \pi^2R^2$.

\end{example}

\section{Wave kinetic integrals and equations}\label{s_kin_int}
\subsection{Wave kinetic integrals}\label{ss_wki}
For a complex 
function $v(s)$, $s\in\R^d$, and $s_1, s_2, s_3, s_4\in \R^d$ with $s_4=s$ we denote
$
v_j = v(s_j),\; j=1,2,3,4$. 
In this section we study the wave kinetic integral $\big(Kv \big)(s)$, defined as follows:
\be\label{k1}
(Kv)(s) = 2\pi
 \int_{\Sst} \frac{ds_1\, ds_2\mid_{\Sst} \deq\,  v_1v_2v_3v_4}{ \sqrt{|s_1-s|^2+|s_2-s|^2}}\left(
\frac1{ v_4} +\frac1{ v_3} -\frac1{ v_2}-\frac1{ v_1}
\right),
\ee
where ${\Sst} $ is the quadric \eqref{sigmaS} without the singular point $(s,s)$, and $s_4=s$.
% and    $s_3$ should be expressed via $s_1, s_2$ and $s$  as $s_3=s_1+s_2-s$.
  Passing to the variable $z=(x,y)=(s_1-s, s_2-s)$ we write $K$ as an integral over $\Sigma_*$ with respect to the 
measure $\mu^\Sigma = |z|^{-1}d_{\Sigma_*}$ (see \eqref{new_int} and Proposition~\ref{p_prop}):
\be\label{k11}
\begin{split}
	(Kv) (s) =  
	2\pi\int_{\Sigma_*} & d\mu^\Sigma(z) \deq
	 \big( v_1v_2v_3 + 
	% \frac{\ga_4}{\ga_3}
	  v_1v_2 v_4 
	 -  v_1 v_3v_4-   v_2v_3v_4\big)(z)\\ &=: K_4(s) + K_3(s) + K_2(s) + K_1(s),
\end{split}
\ee
where now $v_1, v_2,v_3$ should be written as functions of $z$ and $s_4=s$ (note the minus--signs for $K_2$ and $K_1$).
\medskip

Now evoking Theorem \ref{t_disintegr} 
we will show that the wave 
kinetic integral $K$ defines 1-smoothing 
continuous operators in the
complex  spaces 
$\Cc_r(\R^d)$ (see \eqref{norm_m})
with $r$ not too small. 
%We recall that $\ga_s=\lan s\ran^{2r_*}, \  r_*>0$. 

\begin{theorem}\label{t_kin_int}
	If $v\in \Cc_r(\R^d)$ with 
	$r>d$, 
 then $K(v) \in \Cc_{r+1}(\R^d)$ and 
	\be\label{k2}
	|K(v)|_{r+1} \le C_r |v|_r^3,
	\ee
	 where $C_r$ is an absolute constant.
\end{theorem}

That is, the kinetic integral defines a continuous complex--homogeneous mapping of third degree 
$K:  \Cc_{r}(\R^d) \to  \Cc_{r+1}(\R^d)$ if $r>d$. 
We will derive the theorem's assertion from an auxiliary lemma, stated below and proved later in Section~\ref{s_proof_kin_int}.
 %The restriction $r>d-1$ is not optimal and may be relaxed by making the estimates in Section~\ref{s_proof_kin_int} more detailed. 

For $l=1,\dots, 4$ let $J_l(u^1,u^2,u^3,u^4)$ be the complex poly-linear operator of the third order, 
 which does not depend on $u^l$ and 
sends the quadruple of complex 
functions $(u^1(s), \dots, u^4(s))$, $s\in \R^d$, to the function $U_l(s)$, equal to the integral which
defines $K_l(s)$ without the factor $2\pi$
(see \eqref{k11}), where we substitute 
\be\lbl{u_j->s}
v_1 :=u^1(x+s), \quad v_2 :=u^2(y+s),\quad v_3:= u^3(x+y+s),\quad v_4:= u^4(s), 
\ee
in accordance to the relation between the coordinates $(s_1,s_2,s_3,s_4)$ and $(x,y)$. 
That is, 
\be\label{J4}
J_l(u^1,u^2,u^3,u^4)(s)= 
\int_{\Sigma_*}  d\mu^\Sigma(z) 
\prod_{\substack{1\leq i\leq 4 \\ i\ne l}} u^i, \qu l=1,\ldots,4,
\ee
where  $u^1\,\dots, u^4$ depend on the argument $z=(x,y)$ 
as in  \eqref{u_j->s}. Then for $l=1,\dots, 4$ 
\be\label{l1}
(K_lv) (s) = 2\pi \s_l  J_l(v,v,v,v)(s),
\ee
where $\s_1,\s_2:=-1$ and  $\s_3,\s_4:=1$. 
Theorem~\ref{t_kin_int} is an easy consequence of the following assertion:

\begin{lemma}\label{l_kin_int} Let $u^1, \dots, u^4 \in \Cc_r(\R^d)$ where $r>d$. 
	Then for $s\in \R^d$ and $1\le l\le4$ the integral, defining $J_l(u^1,u^2,u^3,u^4)(s) =:J_l(s)$
	converges absolutely and satisfies 
	\be\label{l3}
	| J_l|_{r+1} \le C_r \prod_{j\ne l} |u^j|_r.
	\ee
\end{lemma}

The lemma is proved in  Section~\ref{s_proof_kin_int}. To derive from it  Theorem~\ref{t_kin_int} we note that if the four 
functions $J_l(s)$ are proved to be continuous, then the theorem's assertion would follow from \eqref{l1} and \eqref{l3}. 
To establish the continuity of -- say -- function $J_4$, we note that for any $d$-vector $\xi$, $J_4(s+\xi)$ equals 
$ J_4(u^1_\xi, \dots, u^4_\xi)(s)$, where $u^j_\xi(\eta) = u^j(\eta+\xi)$. So
$$
|J_4(s) -J_4(s+\xi)| \le 
\int_{\Sigma_*}  d\mu^\Sigma(z) | u^1 u^2 u^3 - u^1_\xi u^2_\xi u^3_\xi |.
$$
If $|\xi|\le 1$, then  $ |u^j_\xi|_r \le 2^r |u^j |_r$
for each $j$, and by \eqref{l3} the integral in the r.h.s. is bounded uniformly in $|\xi|\le 1$. Since the integrand
converges to 0 with $\xi$ for each $z$, then 
$|J_4(s) - J_4(s+\xi)|\to 0$ as $\xi\to0$ by Fatou's lemma. So $J_4$ is a continuous function. For the same reason all other 
functions $J_j$ are continuous, and we have completed the derivation of  Theorem~\ref{t_kin_int}  from the lemma. 

The representation \eqref{l1}  together with  \eqref{l3} imply an estimate for 
increments of $K_l$: 

\begin{corollary}\label{c_7.1}
	If $v^1,v^2\in \Cc_r(\R^d)$ are such that   $|v^1|_r, |v^2|_r \le R$,  where $r>d$,  then
	\be\non
	| K_l(v^1) -K_l(v^2)|_{r+1} \le C_r R^2 |v^1-v^2|_r, \quad l=1, \dots, 4. 
	\ee 
\end{corollary}

\subsection{Wave kinetic equations}

Now we pass to the main topic of this section -- the wave kinetic equation:
\be\label{w_k_e}
\dot u(\tau,s) = -\Lc u + \eps K(u) +f(\tau,s), \qquad s \in \R^d,
\ee
\be\label{wk1}
u(0,s) = u_0(s),
\ee
where $0<\eps\le1$, 
$K(u)(\tau,s) = K(u(\tau, \cdot))(s)$ is the wave kinetic integral \eqref{k1} and $\Lc$ is the linear operator
\be\label{operL}
(\Lc u)(s) = 2\ga_s u(s), \quad s \in \R^d. 
\ee
This operator defines in the spaces $\Cc_r$ semigroups of contractions:
\be\label{normL}
\| \exp(-t\Lc)\|_{ \Cc_r(\R^d), \Cc_r(\R^d)} \le \exp(-2t), \qquad \forall\, t\ge0, \; \forall\, r.
\ee
We denote by $X_r$ the space of continuous curves $u: [0, \infty)\to \Cc_r(\R^d)$,  given the uniform norm
$\ 
\|u\tr = \sup_{t \ge0} | u(t)|_r. 
$

\begin{theorem}\label{t_kin_eq}
	If $r>d$, then
	
1) for any $u_0\in \Cc_r(\R^d)$, $f\in X_r$ and any $\eps$ the problem  \eqref{w_k_e}, \eqref{wk1} has at most one solution in $X_r$. 

2) If
	\be\label{wk4}
	|u_0|_r\le C_*,\; \| f\tr \le C_*
	\ee
	for some constant $C_*$, then there exist positive constants $\eps_*=\eps_*(C_*, r)$ and $R=R(C_*, r)$ such that if $0<\eps\le \eps_*$,
	then the problem \eqref{w_k_e}, \eqref{wk1} has a  unique solution 
	$u\in X_r$, and $\|u\tr \le R$.  Moreover, if $(u_{0 1}, f_1)$ and  $(u_{0 2}, f_2)$ are two sets of initial data, satisfying 
	\eqref{wk4}, and $u^1, u^2$ are the corresponding  solutions, then
	\be\label{wk6}
	\| u^1 -u^2\tr \le C(r) (|u_{01} - u_{02}|_r + \| f_1-f_2\tr). 
	\ee
\end{theorem}

The first assertion is obvious in view of the contraction property \eqref{normL}, 
since the mapping $K$ is locally Lipschitz by Corollary~\ref{c_7.1}, cf. the proof of Proposition~\ref{p_stab} below. 
The second  result follows elementary from Theorem~\ref{t_kin_int}. Details are given in  Section~\ref{proof_t_kin_eq}.
\medskip

Let $u^0(\tau,s)$ solves \eqref{w_k_e}, \eqref{wk1} with $\eps:=0$. Writing the solution $u$, constructed in  Theorem~\ref{t_kin_eq},
 as $u^0+\eps v$ we get for $v$
the equation
$\ 
\dot v = -\Lc v+K(u^0+\eps v), \;  v(0)=0.
$
So $\| v\|_r \le C(C_*, r)$ and Corollary \ref{c_7.1} implies that $K(u^0+\eps v) = K(u^0) + O(\eps)$. Accordingly, 
 $v=u^1+O(\eps)$,  where 
$$
\dot u^1 =  -\Lc u^1+K(u^0), \qquad u^1(0)=0.
$$
We have seen that the solution $u(\tau, s)$, built in Theorem \ref{t_kin_eq}, may be written as 
\be\label{u_present}
u(\tau, s)=u^0(\tau, s) +\eps u^1(\tau, s) + O(\eps^2), 
\ee
where $u^0$ and $u^1$ are defined above. 

\medskip
Let us fix any $r_0>d$ and denote
 %$r_0=\max(2r_*+1, d+1)$ and 
 $\eps(f)=\eps_*\big(\| f\|_{r_0}, r_0 \big). 
$

\begin{corollary}\label{c_kin_eq}
Let $u_0=0$, $f \in X_r\ \forall\,r$ and $0\le\eps\le\eps(f)$. Then the problem  \eqref{w_k_e}, \eqref{wk1} has a unique solution $u$ such that 
$u\in X_r$ $\forall\,r$. Its norms $\| u\|_r$ are bounded by constants, depending only on $r$ and $\| f\|_r$.
\end{corollary}
\begin{proof}
By the theorem the problem has a unique solution $u\in X_r$  with $r=r_0$. By Corollary~\ref{c_7.1}, $K(u) \in X_{r+1}$ 
and we get from the equation  \eqref{w_k_e} 
that also $u\in  X_{r+1}$. Iterating this argument we see that $u\in\cap X_r$. The second assertion follows from the theorem. 
\end{proof}

 Now let us assume that in \eqref{w_k_e} the function $f$ does not depend on $\tau$, so $f(\tau,s)=f_s$, where 
\be\label{fr}
f\in \Cc_r, \quad  r>d.
\ee
Then for $\eps=0$ the only steady state of \eqref{w_k_e}, i.e. a solution of the equation 
$-\Lc u +f=0$, is $u^0=\Lc^{-1} f$; it is asymptotically stable. By the  implicit  function theorem, there exists
 $\eps'_r$ such 
that for $0<\eps\le \eps'_r$ equation \eqref{w_k_e} has a unique steady state $u^\eps$, close to $u^0$, 
$
-\Lc u^\eps +\eps K(u^\eps) +f=0,
$
and  $| u^\eps - u^0|_r \le C_r\eps$.  

\begin{proposition}\label{p_stab}
Let $f$ and $r$  be as in \eqref{fr}  and \eqref{wk4} holds for some $C_*$. Let %also $r>d$}, 
$\eps\le \eps_*(C_*,r)$ and $u(\tau)$ be 
a solution of  \eqref{w_k_e}, \eqref{wk1}. Then there exists $\eps'=\eps'(C_*,r)$ 
such that if $0<\eps \le\min(\eps_*, \eps')$, then
\be\non
|u(\tau) - u^\eps|_r \le |u_0 - u^\eps|_r e^{-\tau} \quad\text{for}\quad \tau\ge0. 
\ee
\end{proposition}
\begin{proof}
 Denoting $w(\tau) =   u(\tau) -u^\eps $, we find 
$$
w(\tau) = e^{-\Lc \tau} w(0) +\eps\int_0^\tau e^{-\Lc (\tau-t) } \big( K(u(t)) -K(u^\eps) \big)dt.
$$
Since $|u^0|_r \le |f|_r\le C_*$, then we may assume that $|u^\eps|_r \le 2C_*$ since $\eps$ is sufficiently small. 
Moreover, by Theorem~\ref{t_kin_eq},  
$
|u(\tau) |_r \le R(C_*,r)
$
for all $\tau\ge0$. In particular, this implies that $| w(0)|_r \le C(C_*,r)$. Then, in view of Corollary \ref{c_7.1} and \eqref{normL}, we find
$$
| w(\tau) |_r \le e^{-2\tau} | w(0)|_r + \eps C_r R^2 \int_0^\tau e^{-2(\tau-t)} |w(t)|_r dt.
$$
This relation and Gronwall's lemma, applied to the the function $ e^{2t} |w(t) |_r$, imply that 
$$
|w(\tau)|_r \le |w(0)|_r e^{-\tau(2-\eps  C_r R^2)}. 
$$
Choosing $\eps$ so small that $\eps  C_r R(C_*,r)^2\le1$ we obtain the desired estimate.
\end{proof}

\section{Quasisolutions}
\lbl{sec:quasisol}

In this section we start to study  quasisolutions $A(\tau)=A(\tau;\nu,L)$ of eq.~\eqref{ku4} (where $r_*>0$), 
which are second order truncations of series \eqref{decomp}:
\be\lbl{2nd trunc}
A(\tau) = (A_s(\tau), s\in\Z^d_L), \qquad A_s(\tau) = a^{(0)}_s(\tau) +\rho a^{(1)}_s(\tau) +\rho^2 a^{(2)}_s(\tau)\,.
\ee
%We will examine them for large $L$ and small $\nu$. 
Our main goal  is to examine the energy 	spectrum of $A$, 
$$
n_s(\tau) =n_s(\tau;\nu,L)
= \EE |A_s(\tau)|^2,\quad s\in\Z_L^d, 
$$
when $L$ is large and $\nu$ is small and to show  that $n_s(\tau)$  approximately satisfies the   wave kinetic equation  (WKE) \eqref{wke}. 
The energy spectrum  $n_s$ is a polynomial in $\rho$ of degree four, 
\be\lbl{n_s-decomp}
n_s=n_s^{(0)} + \rho \, n_s^{(1)}+\rho^2 n_s^{(2)} 
+ \rho^3 n_s^{(3)} + \rho^4 n_s^{(4)}, \qu s\in\Z_L^d,
\ee
where $n_s^{(k)}(\tau) = n_s^{(k)}(\tau; \nu,L) $,
\be\non
n_s^{(k)} =\sum_{k_1+k_2=k,\, k_1, k_2\le2} \EE a^{(k_1)}_s \bar a^{(k_2)}_s, \quad 0\le k\le 4. 
\ee 
The term $n_s^{(0)}$ is given by
 \eqref{n_s_0} while by 
 %$n_s^{(0)}$ can be extended to a Schwartz function of $s\in\R^d$, uniformly in $\tau\ge-T$. By  
  Lemma~\ref{l_a1_1},
 \be\lbl{n_s^1}
n_s^{(1)} = 2\Re\EE \bar a_s^{(0)} a_s^{(1)}=0.
\ee
 Writing explicitly $n_s^{(i)}$ with $2\leq i\leq 4$, we find that 
\be\lbl{n^j_s}
n_s^{(2)}=\EE \big(|a_s^{(1)}|^2 + 2\Re \bar a_s^{(0)} a_s^{(2)}\big),
\;\;
n_s^{(3)}=2\Re\EE\bo{s}\atw{s}, \;\;
%\qnd
n_s^{(4)}=\EE|\atw{s}|^2.
\ee
We decompose
$$
n_s=n_s^{\leq 2}+n_s^{\geq 3},
$$
where 
$$
n_s^{\leq 2}=n_s^{(0)} +\rho^2 n_s^{(2)} 
\qnd
n_s^{\geq 3}=\rho^3 n_s^{(3)} + \rho^4 n_s^{(4)}.
$$

Let us extend  $n^{(0)}_s$ to the Schwartz function on $s\in\R^d$ given by \eqref{n_s_0}.
Iterating formula \eqref{an} we  write $a_s^{(n)}(\tau)$, $n\in\N$, as an iterative integral of polynomials of $a_{s'}^{(0)}(\tau')$,
$s'\in \Z^d_L$, $\tau'\le \tau$. Since each $a_{s'}^{(0)}(\tau')$ is a Gaussian random variable \eqref{a0}, 
 the Wick formula applies to every term
$\EE a_s^{(k_1)} \bar a_s^{(k_2)}$   and we see that 
 \begin{equation}\label{schw_ext}
\begin{split}
\text{for any $0\le k_1, k_2\le2$ and any $\nu, L$ the second moment  $\EE a_s^{(k_1)} \bar a_s^{(k_2)}$
}\\
\text{ naturally extends to a Schwartz function of $s\in\R^d$, }
\end{split}
\end{equation} 
cf. \eqref{Ss_extends}. 

The function   $n^{(0)}_s$ is of order one, and   $n^{(1)}_s\equiv 0$ by \eqref{n_s^1}.
 Consider the function $\R^d\ni s\mapsto  n^{(2)}_s(\tau)$.
It is made by two terms (see \eqref{n^j_s}). By Corollary~\ref{t_for_sum} the first may be written as 
$$
\EE |a_s^{(1)}(\tau; \nu,L) |^2 = \nu \Phi_1(s, \tau) +
 O(\nu^{2-\aleph_d}) C^\#(s;\aleph_d)
\quad \text{if}\quad L\ge \nu^{-2}, 
$$
where $\Phi_1$ is a	 Schwartz function of $s$, independent from $L$ and $\nu$. This relation was proved for $T=\infty$, but it 
remains true  for any finite $T$ due to  a similar argument. Similarly if $ L\ge \nu^{-2}$, then
$
\EE \Re \big(\bar a_s^{(0)} a_s^{(2)}(\tau; \nu,L) \big) = \nu \Phi_2(s, \tau) +
 O(\nu^{2-\aleph_d}) C^\#(s;\aleph_d),
$ 
so 
\be\label{decompp}
n^{(2)}_s(\tau; \nu,L) = \nu \Phi(s, \tau) +
 O(\nu^{2-\aleph_d}) C^\#(s;\aleph_d)
\quad \text{if}\quad L\ge \nu^{-2}, 
\ee
where  $s\mapsto \Phi (s, \tau) $ is  a Schwartz function 
 (we are not giving a complete proof of \eqref{decompp} since this relation is used only for
motivation and discussion). To understand the limiting behaviour of $n^{(3)}_s$ and $n^{(4)}_s$ we will use another result, proved 
in \cite{DK}, where $a_s^{(i)}(\tau)$ denote the terms of the series \eqref{decomp}. 
Recall that the function $\chi_d$ is defined in \eqref{chi_d}  and $\lceil\cdot\rceil$ --  in Notation.

\begin{theorem} \lbl{l:est_a^ia^j}
	%Assume that $d\geq 2$. 
For any $k_1,k_2\geq 0$ and $k:=k_1+k_2$,  we have 
	\be\lbl{main}	
	\big| \EE a_s^{(k_1)}(\tau_1)\bar a_s^{(k_2)}(\tau_2)\big|
	\leq C^{\#}(s;k)\big(\nu^{-2}L^{-2} + \max(\nu^{\lceil k/2 \rceil},\nu^d) \,\chi_d^k(\nu)\big)
	\ee
for any $s\in \Z^d_L$, 
	uniformly in $\tau_1,\tau_2\geq -T$, where $\chi_d^k(\nu)=\chi_d(\nu)$ if $k=3$ and $\chi_d^k(\nu)\equiv 1$ otherwise.
The second moment $ \EE a_s^{(k_1)}(\tau_1)\bar a_s^{(k_2)}(\tau_2)$ extends to a Schwartz function of $s\in\R^d\supset\Z^d_L$ which satisfies the same estimate \eqref{main}.
\end{theorem}

Note that for $k\le2$ the theorem's assertion follows from the preceding discussion since $d\ge2$ (and recall that  for  $k=1$ the l.h.s. 
of \eqref{main} vanishes by Lemma~\ref{l_a1_1}).  A  short direct proof of \eqref{main} with $k_1=k_2=1$ and $T=\infty$ is given in Addendum~\ref{a_stat_phase}. 
 In Section~\ref{s_en_spectra} we discuss a strategy of the theorem's  proof for any $k$,  given in \cite{DK}. 
By Theorem~\ref{l:est_a^ia^j}, for any $s\in\R^d$
\be\label{n_s_34}
|n^{(3)}_s|\le C^\#(s)\nu^2\chi_d(\nu), \quad |n^{(4)}_s| \le   C^\#(s) \nu^2 \quad \text{ if} \quad  L\ge \nu^{-2}. 
\ee
Choosing the parameter $\rho$ in the form \eqref{r_scale}, we will examine the energy spectrum $n_s$ under the limit \eqref{assumption}. Due to the discussion above, under this limit
$$
n_s^{(0)} \equiv B(s) \big( 1-e^{-2\ga_s(T+\tau)}\big), \;\;\  \rho n_s^{(1)} \equiv0 ,\;\;\  \rho^2 n_s^{(2)} \to\eps \Phi(s,\tau), \;\;\ 
 \rho^3 n_s^{(3)} \to0.
$$
Concerning the term $ \rho^4 n_s^{(4)}$, our results do not allow to find its asymptotic under  the limit, but only imply
that $ |\rho^4 n_s^{(4)} |\le \eps^2 C^\#(s) $. Accordingly our goal is to examine the energy spectrum $n_s$ under the limit 
 \eqref{assumption} and the scaling \eqref{r_scale} with precision $\eps^2 C^\#(s) $, regarding the constant 
 $\eps\le1$ (which measures the size of
 solutions for \eqref{ku3} under the proper scaling) as a fixed small parameter.

\subsection{Increments of the energy spectra $n_s^{\leq 2}$ and the reminder $n_s^{\geq 3}$.} \label{s_8.1}

We will show that the process $n_s^{\leq 2}$ approximately satisfies the  WKE  \eqref{wke},  while the reminder $n_s^{\geq 3}$ is small. 
This will imply that  $n_s(\tau)$  is an approximate solution of the WKE. We always assume \eqref{r_scale} and that  $L\ge1$, $0<\nu\le1/2$.

Now for  $u\in \Cc_r(\R^d)$, $ r>d$,  and for $\tau\in(0,1]$, 
we consider the kinetic integral 
 $K^\tau (u) =\big((K^\tau u)(s), s\in \R^d\big)$:
 \be\label{I_tau}
K^\tau (u) = \int_0^{\tau} e^{-t\Lc} K (u)\, dt, % +2 \sum_{j=1}^3  \int_0^{\tau/2} e^{-t\Lc} K_j(u)\, dt,
\ee
where the operator $K=K_1+\dots+K_4$ is defined in Section \ref{s_kin_int} and 
the linear  operator $\Lc$ is  introduced  in \eqref{operL}.  That is,
\be\lbl{kin int^tau}
(K^\tau u)(s) =  \frac{1-e^{-2\ga_s\tau}}{2 \ga_s} (Ku)(s) =
	 \frac{1-e^{-2\ga_s\tau}}{2\ga_s}\sum_{j=1}^4 (K_ju)(s).
\ee
 The result below is the main step in establishing the wave kinetic limit. 
There, using   \eqref{schw_ext}, we regard 
$n_s^{\leq 2}(\tau)$  as  a Schwartz  function of $s\in\R^d$. 

\begin{theorem}\lbl{th:kinetic_for_n^2}
	For any $0<\tau\le1$ we have 
	\be\label{n_increment}
	\begin{split}
		n^{\leq 2}(\tau)& = e^{-\tau \Lc} n^{\leq 2}(0) +2 \int_0^\tau e^{-t\Lc} {b}^2\, dt 
		+\eps K^\tau  (n^{\le2}(0) ) +\Rc,  \\
	\end{split}
	\ee
where $b^2 =\{ b^2(s) , s\in \R^d\}$ and   the reminder  $\Rc(\tau,s)$ satisfies 
	\be\lbl{R_s}
	|\Rc(\tau)|_r \le 
	 C_{r,\aleph_d} \,\eps \big( \nu^{1-\aleph_d}   +\nu^{-3}L^{-2} + \tau^2 + \eps\tau \big)\,,\quad \forall r,
	\ee
	 and $\aleph_d$ is defined as in Theorem~\ref{t_singint}. 
\end{theorem}

Proof of Theorem~\ref{th:kinetic_for_n^2} is given in Section~\ref{s:th_kinetic}. Since for any $\tau'\ge-T$ the process $\tau \to (A_s(\tau'+\tau), s\in \Z^d_L)$, is a quasisolution of the problem 
\eqref{ku4}, \eqref{in_cond} with $T:= T+\tau'$ and $\beta_s(\tau) := \beta_s(\tau+\tau')$, then the theorem applies to study the 
increments of $n^{\leq 2}$ from $\tau'$ to $\tau'+\tau$, for any $\tau'\ge-T$. 
That is, \eqref{n_increment} remains true if we replace $n^{\leq 2}(0)$ by
$n^{\leq 2}(\tau')$ and $n^{\leq 2}(\tau)$ by $n^{\leq 2}(\tau'+\tau)$.

We also need an estimate on the reminder $n_s^{\geq 3}$. It is  a part of the assertion below, which is an immediate
consequence of  \eqref{main}-\eqref{n_s_34} since $d\geq 2$:

\begin{proposition}\lbl{l:reminder}
If
$ L\ge  \nu^{-2}$, then for  $k=0,1,2,4$ we have
\be\lbl{n_s_est}
\big|n_s^{(k)}(\tau)\big|\leq C^\#(s)\nu^{\lceil k/2 \rceil}, \quad
%\qnd\quad \big|n_s^{(3)}(\tau)\big|\leq C^\#(s)\nu^2\chi_d(\nu),
	\ee
and $\big|n_s^{(3)}(\tau)\big|\leq C^\#(s)\nu^{2}\chi_d(\nu)$,	uniformly in $\tau\geq -T$.
So
\be\label{n_est}
%\big|n_s (\tau)|, \, 
\big|n_s^{\le 2} (\tau)\big| \leq C^{\#}(s)
%\big( \nu^{-4} L^{-2} +\sqrt{\nu}(\chi_d(\nu))^2 + \rho_\#^4 \big).
\ee 
and if $\nu(\chi_d(\nu))^{1/2}\leq\eps$, then 
	\be\lbl{n_s^3-est}
	\big|n_s^{\geq 3}(\tau)\big|\leq C^{\#}(s) \eps^2. %,\big((\rho^3+\rho^4) (\nu L)^{-2} + \rho^3 \nu^2(\chi_d(\nu))^2 +  \rho^4 \nu^2\big).
	\ee
\end{proposition}

In accordance with \eqref{assumption} we will study the energy spectrum $n_s(\tau)$ under the two limiting regimes: 
 {\it $L\gg \nu^{-1}$ when  $\nu\to 0$; }
 or
 {\it first $L\to \infty$ and then $\nu\to 0$. }
To treat the latter we will need the following result. 
\begin{proposition}\lbl{L to infty}
	For any $\nu\in (0,1/2]$, $\tau_1,\tau_2\geq -T$, $k_1,k_2\geq 0$ and $s\in\R^d$ 
	 the moment
	$\EE a_s^{(k_1)}(\tau_1;\nu,L)\bar a_s^{(k_2)}(\tau_2;\nu,L)$ admits a finite limit as $L\to\infty$. The limit is
	a Schwartz function of $s$. 
\end{proposition}
In particular, this result implies that any $n_s^{(k)}(\tau;\nu,L)$ converges, as $L\to\infty$, to a Schwartz function 
 $n_s^{(k)}(\tau;\nu,\infty)$ 
of $s\in\R^d$. 
The  proposition can be obtained  directly by iterating the Duhamel formula \eqref{an} and using Theorem~\ref{t_sum_integral} to replace 
the corresponding sum by an $L$--independent integral (cf. Section~\ref{s_sing_integral}, where the moments with $k_1=k_2=1$
are approximated by integrals $I_s$). 
We do not give  here a proof  since it  follows from a stronger result  in \cite{DK},  discussed in Section~\ref{s_en_spectra} (see there
\eqref{Mom a-a} and \eqref{I_F}).

\subsection{Proof of Theorem \ref{th:kinetic_for_n^2}.}
\lbl{s:th_kinetic}

It is convenient to decompose the processes $a^{(i)}_s(\tau)$, $\tau\ge0$, as
\be\lbl{a-a_del}
a_s^{(i)}(\tau) = c_s^{(i)}(\tau) + \Del a_s^{(i)}(\tau), \qquad i=0,1,2,\qu s\in \Z^d_L,
\ee
where 
$$
c_s^{(i)}(\tau)=e^{-\ga_s\tau}a_s^{(i)}(0)
$$
and $\Del a_s^{(i)}$ is defined via \eqref{a-a_del}. That is,   
$
c_\cdot(\tau) := c^{(0)}_\cdot(\tau) +\rho  c^{(1)}_\cdot(\tau) +\rho^2  c^{(2)}_\cdot(\tau)$
with 
$\tau\ge0$ is a solution of the linear equation \eqref{ku4}${}_{\rho=0, b(s)\equiv 0}$, equal $A_\cdot(0)$ at $\tau=0$, and 
$\Delta a_\cdot(\tau)$ equals $ A_\cdot(\tau) - c_\cdot(\tau)$. By \eqref{schw_ext},  for $0\le i,j\le2$ the functions 
 \begin{equation}\label{smooth_ext}
\begin{split}
&\EE c_s^{(i)} \bar c_s^{(j)} ,\quad \EE c_s^{(i)} \Delta \bar a_s^{(j)} ,\quad   \EE   \Delta  a_s^{(i)} \Delta \bar a_s^{(j)} 
\\
&\text{ naturally extend to  Schwartz functions of $s\in\R^d$. }
\end{split}
\end{equation} 
Due to \eqref{n_s^1} and \eqref{n^j_s}, 
$$e^{-2\ga_s\tau}n_s^{\leq 2}(0)
=\EE|c_s^{(0)}(\tau)|^2 + \rho^2\EE \big(|c_s^{(1)}(\tau)|^2 + 2\Re \bar c_s^{(0)}(\tau) c_s^{(2)}(\tau)\big),\quad \forall\, s\in \R^d.
$$ 
Also, 
 \begin{equation}\label{Delta_A^2}
\begin{split}
n_s^{\leq 2}&(\tau)-e^{-2\ga_s\tau}n_s^{\leq 2}(0)
=\EE\Big( |a_s^{(0)}(\tau)|^2  - |c_s^{(0)}(\tau)|^2 \\
	&+\rho^2  \big(|a_s^{(1)}(\tau)|^2  - |c_s^{(1)}(\tau)|^2
	 +2\Re\big(a^{(2)}_s\bar a^{(0)}_s(\tau)- c^{(2)}_s\bar c^{(0)}_s(\tau)\big)\Big).
\end{split}
\end{equation}

Writing explicitly  processes  
$
\Del a_s^{(i)}(\tau)
$,
$s\in\Z^d_L$, 
from eq.~\eqref{ku4a}, we find
\begin{equation}\lbl{deltas_a}
\begin{split}
	\Del a_s^{(0)}(\tau)&=
	b(s)\int_0^\tau e^{-\ga_s(\tau-l)}\,d\beta_s(l), 
	\\
	\Del a_s^{(1)}(\tau)&=
	i\int_0^\tau e^{-\ga_s(\tau-l)} \cY_s(a^{(0)},\nu^{-1}l)\,dl,
	\\ 
	\Del a_s^{(2)}(\tau)&=
	i\int_0^\tau e^{-\ga_s(\tau-l)} 3\cY_s^{sym}(a^{(0)},a^{(0)},a^{(1)};\nu^{-1}l)\,dl,
	\end{split}
\end{equation}
where $a^{(0)} = a^{(0)}(l)$ and 
 we recall that $\cY_s^{sym}$ is defined at the beginnig of Section~\ref{s_series}. 
Let us note that  to get explicit formulas for $c_s^{(i)}(\tau)$, $i=0,1,2$, it suffices to replace in the r.h.s.'s  of the relations in 
\eqref{deltas_a} the range of  integrating from $[0,\tau]$  to  $[-T,0]$.  For example,
$
c_s^{(0)}(\tau) = e^{-\ga_s\tau}  a_s^{(0)}(0) = b(s)\int_{-T}^0 e^{-\ga_s(\tau-l)}\,d\beta_s(l). 
$

Using that 
$\EE c_s^{(i)} \Del \bar a_s^{(0)}=\EE c_s^{(i)} \EE \Del \bar a_s^{(0)}=0$ for any $i$ and $s$, 
we obtain 
\be\lbl{aa0}
\EE \big(a_s^{(2)} \bar a_s^{(0)}(\tau)  - c_s^{(2)} \bar c_s^{(0)}(\tau) \big)
=\EE \Del a_s^{(2)} \bar a_s^{(0)}(\tau) ,
\ee
and from \eqref{a-a_del} we get that 
\begin{equation}\lbl{a^2-a^2}
\begin{split}
&|a_s^{(1)}(\tau)|^2  -  |c_s^{(1)}(\tau)|^2=|\Del a_s^{(1)}(\tau)|^2  + 2\Re \Del a_s^{(1)} \bar c_s^{(1)}(\tau),
	\\
&|a_s^{(0)}(\tau)|^2  -  |c_s^{(0)}(\tau)|^2=|\Del a_s^{(0)}(\tau)|^2.
	\end{split}
\end{equation}

Then, inserting \eqref{aa0} and \eqref{a^2-a^2} into \eqref{Delta_A^2},
we find
\be\non
n_s^{\leq 2}(\tau)-e^{-2\ga_s\tau}n_s^{\leq 2}(0)
= \EE\big|\Del a_s^{(0)}(\tau)\big|^2 + \rho^2 Q_s(\tau),\quad s\in\R^d, 
\ee
where 
\be\lbl{Q_s-def}
Q_s(\tau):=\EE |\Del a_s^{(1)}(\tau)|^2 + 2\Re\EE\big(\Del a^{(1)}_s(\tau) \bar c^{(1)}_s(\tau) + \Del a^{(2)}_s(\tau) \bar a^{(0)}_s (\tau)\big),
\ee
and we recall \eqref{smooth_ext}. 
Since 
$$\displaystyle{\EE\big|\Del a_s^{(0)}(\tau)\big|^2=\frac{b(s)^2}{\ga_s}(1-e^{-2\ga_s\tau})} =
2 \int_0^\tau e^{-t\Lc} {b}^2(s)\, dt,
$$
 then 
\be\non
n^{\leq 2}(\tau)-e^{-t\Lc}n^{\leq 2}(0)= 2 \int_0^\tau e^{-t\Lc} {b}^2\, dt + \rho^2 Q(\tau),
\ee
for $n^{\leq 2} = (n^{\leq 2}_s$, $s\in \R^d)$. 
So  the desired formula \eqref{n_increment} is an 
 immediate consequence  of the assertion below:
\begin{proposition}\lbl{l:N^2_s}
	We have	 
	\be\non
	\rho^2 Q_s(\tau) = \eps K^\tau(n^{\le2}(0))(s) +  \Rc(\tau,s), \quad s\in \R^d, 
	\ee
	where the reminder $\Rc$ satisfies \eqref{R_s}.
	\end{proposition}
\begin{proof} Below  we abbreviate  $n_s^{(0)}(0)$ to $n_s^{(0)}$.

Since $\rho^2\nu=\eps$, then  we should show that for any $r$,
\be\lbl{t-N2}
\big| Q(\tau) - \nu K^\tau (n^{\le2}(0))\big|_r \leq C_{r,\aleph_d}(\nu^{2-\aleph_d} 
+\nu^{-2}L^{-2} + \nu\tau^2 + \eps\nu\tau).
\ee
To this end, iterating  formula \eqref{an}, we will express the processes $\Del a_s^{(2)}$, $\Del a_s^{(1)}$ and $c_k^{(1)}$,
entering  the definition \eqref{Q_s-def} of $Q_s$,  through the processes $a_k^{(0)}$.
Then, applying the Wick formula, we will see that $Q_s$ depends on the quasisolution $A_s$ only through the correlations 
of the form $\EE a_k^{(0)}(l)\bar a_{k}^{(0)}(l')$. 
We will show that the main  input to $Q_s$ comes from
 those terms which depend only on the correlations with the times $l,l'$ satisfying $0\leq l,l'\leq \tau$. 
Then, approximating these correlations by their values at  $l=l'=0$, 
we will see  that 
 $Q_s(\tau) $ is close to a sum $\cZ_s$ from Proposition~\ref{l:N2} below, which depends only on $\tau$ and the energy 
spectrum $n_k^{(0)}(0)=\EE |a_k^{(0)}(0)|^2$. 
Next we will approximate the sum  $\cZ_s$ by its asymptotic  as $\nu\to 0$ and $L\to\infty$, which is given by the kinetic integral $\nu K^\tau (n^{(0)}(0))$. Finally, 
replacing  $\nu K^\tau (n^{(0)}(0))$ with  $\nu K^\tau (n^{\le2}(0))$ and estimating the difference of the two kinetic integrals we will  get \eqref{t-N2}.

%Now we pass to a detailed proof. 
 We will derive \eqref{t-N2} from the following result: 

\begin{proposition}\lbl{l:N2}
	We have
	\be\lbl{l-N2}
	\big| Q_s(\tau) - \cZ_s\big| \leq C^{\#}(s)(\nu^2\chi_d(\nu)+\nu^{-2}L^{-2}+\nu\tau^2), \quad s\in \R^d, 
	\ee
	where  
	\be\lbl{ZZ_s}
\begin{split}
\cZ_s:=2\ssum_{1,2}\dep 
		\big(\cZ^4n^{(0)}_1n^{(0)}_2n^{(0)}_3
		+\cZ^3n^{(0)}_1n^{(0)}_2n^{(0)}_s\\
		-2\cZ^1n^{(0)}_2n^{(0)}_3n^{(0)}_s\big)
	=: 2 S^1_s + 2S^2_s -4 S^3_s,
	\end{split}
\ee
	and the terms $\cZ^i=\cZ^i(s_1,s_2,s_3,s,\tau)$ have the following form: 
	\be\lbl{Z-edi}
	\cZ^4
	=\frac{\big|e^{i\nu^{-1}\oms\tau}-e^{-\ga_s\tau}\big|^2}{\ga_s^2+(\nu^{-1}\oms)^2},
	\qu
	\cZ^j = 2 \,\frac{1-e^{-\ga_s \tau}}{\ga_s}  
	 \frac{\ga_j}{\ga_j^2+(\nu^{-1}\om^{12}_{3s})^2}
	%=2 \frac{1-e^{-\ga_s \tau}}{\ga_j^2+(\nu^{-1}\om^{12}_{3s})^2}
	\quad
	\mbox{for}\; j=1,2,3.
	\ee
\end{proposition}

The proposition is proved in the next section.

To deduce the desired estimate \eqref{t-N2} from \eqref{l-N2} we  will 
 approximate  the sums $S^j_s$  ($j=1,2,3$) in  \eqref{ZZ_s}  by their asymptotic  as $\nu\to 0$ and $L\to\infty$:

\smallskip

{\it The  sum $S^1_s$.} It 
has the form \eqref{Sis_1} with $F_s(s_1,s_2)=n_{s_1}^{(0)}n_{s_2}^{(0)}n_{s_1+s_2-s}^{(0)}.$ So by 
 \eqref{Iso_diffe} and  Theorem~\ref{t_osc}, 
\be\lbl{cZ_s^s}
\Big|S^1_s
- \nu\frac{1-e^{-2\ga_s\tau}}{4\ga_s}  (K_4 n^{(0)})(s)
\Big| \le C^\#(s;\aleph_d)(\nu^{2-\aleph_d} + \nu^{-2} L^{-2})
\ee
for all $s\in\R^d$, 
where we recall that the integral $K_4$ is defined in \eqref{k11}.

{\it The  sum  $S^2_s$.}
Let us set $F_s(s_1,s_2):= \ga_3n_{s_1}^{(0)}n_{s_2}^{(0)}n_{s}^{(0)}$ and $\Ga_s(s_1,s_2)=\ga_3/2$.
Then the  sum takes the form
$$
S^2_s
=\frac{1-e^{-\ga_s\tau}}{2   \ga_s}\nu^2\ssum_{s_1,s_2}\dep\frac{ F_s(s_1,s_2)}{(\nu\Ga_s)^2 + (\oms/2)^2}.
$$
Applying  Theorems~\ref{t_sum_integral} and \ref{t_singint}  we get 
\be\lbl{cZ_s^3}
\Big|
S^2_s-  \nu \frac{1-e^{-\ga_s\tau}}{2\ga_s} (K_3n^{(0)})(s)
\Big| \le  C^\#(s;\aleph_d)(\nu^{2-\aleph_d} + \nu^{-2} L^{-2}).
\ee 
{\it The  sum $S^3_s$.}
Let us note that the function $\cZ^1$ equals to $\cZ^3$, if we there replace  $\ga_3$  by $\ga_1$. 
Then,  repeating the argument used to analyse the sum $S^2_s$, we get
\be\lbl{cZ_s^1}
\begin{split}
	\Big| -S^3_s
	%-\ssum_{s_1,s_2}\dep \cZ^1n_{s_2}^{(0)}n_{s_3}^{(0)}n_{s}^{(0)} 
-  \nu\frac{1-e^{-\ga_s\tau}}{2\ga_s} (K_1n^{(0)})(s)
\Big| 
\le  C^\#(s;\aleph_d)(\nu^{2-\aleph_d} + \nu^{-2} L^{-2}).
\end{split}
\ee

Using the symmetry of the integral $K_1$ with respect to the transformation $(s_1,s_2)\mapsto (s_2,s_1)$ in its integrand, we see that $K_1(n^{(0)})=K_2(n^{(0)})$.

 In \eqref{cZ_s^s} the sum $S_s^1$ is approximated by the integral $\frac{\nu}{2}K_4(n^{(0)})$, multiplied by the factor $\frac{1-e^{-2\ga_s\tau}}{2\ga_s}$ which also arises in \eqref{kin int^tau}, while for the sums $S_s^2$ and $S_s^3$ the corresponding factors are slightly different, see \eqref{cZ_s^3} and \eqref{cZ_s^1}. 
	To handle this difficulty we consider $K_j(n^{(0)}) =: \eta$, where $j$ is 1 or 3.
 By \eqref{n_s_0} and Lemma~\ref{l_kin_int}, 
$\ 
|\eta |_r \le C_r$ for all $r. 
$
Denote
$$
\xi_s = \Big( \frac{1-e^{-\ga_s\tau}}{\ga_s}  - \frac{1-e^{-2\ga_s\tau}}{2\ga_s} 
\Big) \eta_s, \quad s\in\R^d. 
$$
 Since
$\ 
\frac{1-e^{-\ga_s\tau}}{\ga_s}  - \frac{1-e^{-2\ga_s\tau}}{2\ga_s}  =\frac{(1-e^{-\ga_s\tau})^2}{2\ga_s}\leq \tau^2\ga_s$ and $\ga_s=\lan s\ran^{2r_*}$,
we find 
\be
%\label{xi_est}
\non
|\xi|_r \le \tau^2 | \eta|_{r+2r_*} \le    \tau^2 C_r.
\ee
This estimate allows to replace in \eqref{cZ_s^3} and  \eqref{cZ_s^1} 
$
\nu \frac{1-e^{-\ga_s\tau}}{2\ga_s} K_j(n^{(0)})
$
by
$
\nu \frac{1-e^{-2\ga_s\tau}}{4\ga_s} K_j(n^{(0)})
$
 with accuracy $\nu\tau^2$.
So recalling 
the definition of $K^\tau$ in \eqref{kin int^tau},  combining \eqref{cZ_s^s}, \eqref{cZ_s^3}, \eqref{cZ_s^1} and using that $K_1(n^{(0)})=K_2(n^{(0)})$, 
 for any $r$ we get 
\be\lbl{almots_kin'}
\big|\cZ_s - \nu (K^\tau n^{(0)})(s)\big|_{r}
\leq  
C_{r,\aleph_d}(\nu^{2-\aleph_d} + \nu^{-2} L^{-2})+C_r\nu\tau^2 . 
\ee
Finally, since by Proposition \ref{l:reminder} together with \eqref{n_s^1}
$
| n_s^{\le 2}(\tau) -n_s^{(0)}(\tau)| = \rho^2 |n_s^{(2)}(\tau)|\le C^\#(s)\eps,
$
then 
$
| K^\tau (n^{\le2}(0)) - K^\tau (n^{(0)}(0))|_{r+1} \le C_r\eps \tau
$
in view of  Corollary~\ref{c_7.1} and \eqref{I_tau}. 
This inequality jointly with estimates  \eqref{almots_kin'} and  \eqref{l-N2} 
imply  the desired relation \eqref{t-N2}.
\begin{comment}
\medskip

\noindent{\it Proof of \eqref{xi_est}.} Noting  that
$\ 
\frac{1-e^{-\ga_s\tau}}{\ga_s}  - \frac{1-e^{-2\ga_s\tau}}{2\ga_s}  =\frac{(1-e^{-\ga_s\tau})^2}{2\ga_s} =: m_s,
$
we set 
$$
	m_s^1=
	\left\{\begin{array}{ll}
		m_s,& \ga_s\tau\le1 \,,
		\\
		0,& \text{otherwise} \,,
\end{array}\right.
$$
and define 
$m^2_s = m_s -m^1_s$. Then $m^1_s \le \ga_s\tau^2/2$. So
$
|m^1 \xi |_r \le \tau^2|\xi |_{r+2r_*}/2. 
$
Now consider $m^2\xi$. Since $m^2_s \le 1/2\ga_s$ and $m_s^2=0$ if $\lan s\ran \le \tau^{-1/2r_*}$, then
\begin{equation*}
\begin{split}
|m^2 \xi |_r \le  \sup_{|s| \ge \tau^{-1/(2r_*)}} \lan s\ran^r | m_s^2 \xi_s| \le |\xi|_{r+2r_*} 
 \sup_{|s| \ge \tau^{-1/(2r_*)}} \lan s\ran^{-4r_*} =  |\xi|_{r+2r_*}  \tau^{2}.
	\end{split}
\end{equation*}
The obtained estimates on $|m^1 \xi |_r$ and  $|m^2 \xi |_r$ imply \eqref{xi_est}. 
\end{comment}
\end{proof}

\section{Proof of Proposition \ref{l:N2}} \lbl{sec:Q_s-prop}

To conclude the proof of Theorem~\ref{th:kinetic_for_n^2}, it remains to establish  Proposition~ \ref{l:N2}. The proof of the proposition is somewhat cumbersome since we have 
to consider a number of different terms and different cases. 
During the proof  we will often skip the upper index $(0)$, 
so by writing $a$ and $a_s$ we  will mean $a^{(0)}$ and  $a_s^{(0)}$.

We recall that $Q_s$ is given by  formula \eqref{Q_s-def} and first consider the 
 term $\EE \Del a^{(2)}_s(\tau) \bar a_s (\tau)$. 
Inserting the identity
$a^{(1)}(l)=c^{(1)}(l) + \Del a^{(1)}(l)$
into  formula  \eqref{deltas_a} for  $\Del a^{(2)}_s $, we obtain
$$
\EE \Del a^{(2)}_s (\tau) \bar a_s(\tau)= N_s + \wt N_s,
$$
where
\be\lbl{NN22}
N_s:=i\, \EE \Big(\bar a_s(\tau)
\int_{0}^\tau e^{-\ga_s(\tau-l)} 3\cY_s^{sym}(a,a,\Del a^{(1)};\nu^{-1}l)\,dl
\Big)
\ee
and
$$
\wt N_s:= i\,\EE \Big(\bar a_s(\tau)
\int_{0}^\tau e^{-\ga_s(\tau-l)} 3\cY_s^{sym}(a,a,c^{(1)} ;\nu^{-1}l)\,dl
\Big).
$$
Thus, 
\be\non
Q_s= \EE |\Del a_s^{(1)}(\tau)|^2 + 2 \Re N_s + 2\Re \EE\Del a^{(1)}_s (\tau)\bar c^{(1)}_s(\tau) +2\Re \wt N_s,\quad
s\in\R^d. 
\ee
So we have to analyse the four terms in the r.h.s. above.

\subsection{The first term of $Q_s$} \label{ss_2}
First we will show that the term $\EE  |\Del a^{(1)}_s(\tau)|^2$ 
can be approximated by the first sum $2S_1$ from \eqref{ZZ_s}. Indeed,  due to \eqref{deltas_a}, we  have
\be\lbl{E Del a}
\EE  |\Del a^{(1)}_s(\tau)|^2
=\EE\int_0^\tau dl \int_0^\tau dl'\,
e^{-\ga_s(2\tau-l-l') } \cY_s(a,\nu^{-1}l)
\ov{\cY_s(a,\nu^{-1}l')}.
\ee
Writing the functions $\cY_s$ explicitly and applying the Wick theorem, in view of \eqref{corr_a_in_time} we find
\begin{align}\non
	\EE  |\Del a^{(1)}_s(\tau)|^2=
	2\ssum_{1,2}
	\dep\int_0^\tau dl &\int_0^\tau dl'\,
	e^{-\ga_s(2\tau-l-l') + i\nu^{-1}\oms( l - l')}
	\\\lbl{N^2,1_wick}
	&\EE a_1(l) \bar a_{1}(l') \,
	\EE a_2(l) \bar a_{2}(l') \,
	\EE \bar a_3(l) a_{3}(l').
\end{align}
 Denoting $g_{s123}(l,l',\tau)= -\ga_s(2\tau-l-l')+i\nu^{-1}\oms (l-l')$ and computing the time integrals of the exponent above, we obtain
$$
\int_0^\tau dl \int_0^\tau dl' \,
e^{g_{s123}(l,l',\tau)}=\cZ^4,
%e^{-\ga_s(2\tau-l-l')+i\nu^{-1}\oms (l-l')}=\cZ^4,
$$
where $\cZ^4$ is defined in \eqref{Z-edi}.
Together with the sum in  \eqref{N^2,1_wick} we consider a sum obtained from the latter, without factor 2, by replacing
the processes $a_k(l)$, $a_k(l')$ by their value at zero $a_k(0)$:
\be\lbl{proc->proc at 0}
\ssum_{1,2}\dep
\cZ^4 \,
\EE |a_1(0)|^2\EE|a_2(0)|^2 \EE|a_3(0)|^2 = 
\ssum_{1,2}\dep
\cZ^4 \,
n_1^{(0)}n_2^{(0)}n_3^{(0)}=S^1_s. 
\ee
Our goal in this section is to show that
	\be\lbl{N21-hat}
	\big| \EE  |\Del a^{(1)}_s(\tau)|^2 - 2S^1_s\big| \le C^\#(s) (\nu^{-2} L^{-2} + \tau^2\nu).
	\ee
Due to \eqref{N^2,1_wick} and \eqref{corr_a_in_time},  
\be\lbl{Del aaa}
\EE  |\Del a^{(1)}_s(\tau)|^2 = 2\ssum_{1,2} \dess \int_0^\tau dl \int_0^\tau dl'\,
e^{g_{s123}(l,l',\tau)}
B_{123} h_{123}(l,l'),
\ee
where $B_{123}$ is the function defined in \eqref{B_123 def},   extended from $(\Z^d_L)^3$ to $(\R^d)^3$, and 
$
h_{123}(l,l')=\prod_{j=1}^3\big(e^{-\ga_j|l-l'|}-e^{-\ga_j(2T+l+l')}\big)
$
 (also viewed as a function on $(\R^d)^3$). Let  $f_{123}$ denotes the increment of the function $h_{123}$, that is
$$
  f_{123}(l,l')=h_{123}(l,l')-h_{123}(0,0).
$$
It is straightforward to see that
\be\lbl{f_123 est}
|\p^{\al}_{s_1,s_2,s_3} f_{123}(l,l')| \le C_{|\al|} \lan (s_1,s_2,s_3) \ran^{m_{|\al|}}\, \tau\qquad \forall |\al|\geq 0,
\ee
for appropriate constants $C_k,m_k>0$, uniformly in $0\leq l,l'\leq \tau$.
Since $n_1^{(0)}n_2^{(0)}n_3^{(0)}=B_{123}h_{123}(0,0)$, from \eqref{Del aaa} and \eqref{proc->proc at 0} we see that 
$
\EE  |\Del a^{(1)}_s|^2  -  2S^1_s = 2S_s^{\Delta} 
$ with 
$$
S_s^{\Delta} = \ssum_{1,2} \dess \int_0^\tau dl \int_0^\tau dl'\,
e^{g_{s123}(l,l',\tau)} B_{123} f_{123}(l,l').
$$
Now in the expression above we replace $\ssum_{1,2}$
by $\int ds_1\,ds_2$ and denote the obtained integral by $I_s^\Delta$,
\begin{equation}\lbl{I^Del_s}
\begin{split}
I_s^\Delta =  \int_0^\tau dl &\int_0^\tau dl' \int_{\R^{2d}}ds_1\, ds_2 \,  
e^{g_{s123}(l,l',\tau)}
\\
&B_{123}
f_{123}(l,l')
, \qquad s_3= s_1+s_2-s.
\end{split}
\end{equation}
Since $B_{123}$ with $s_3=s-s_1-s_2$ is a Schwartz function of $s,s_1,s_2$ and the function $f_{123}$ satisfies \eqref{f_123 est}, then Theorem~\ref{t_sum_integral} applies and we find
\be\label{f_1}
|S_s^\Delta - I_s^\Delta| \le C^\#(s) \nu^{-2}L^{-2}.
\ee
To establish \eqref{N21-hat} it remains to prove that 
$|I_s^\Delta|\le C^\#(s)\tau^2\nu$. To do that we divide the 
external integral (over $dldl'$) in  \eqref{I^Del_s} to two integrals:

\smallskip

\noindent {\it Integral over $|l-l'|\le\nu$}. In view of \eqref{f_123 est} with $\al=0$, in this case  
the internal integral (over $ds_1ds_2$)  is bounded by $ C^\#(s)\tau$, so $I_s^\Delta$ is bounded by $C_1^\#(s)\tau^2\nu$.

\smallskip

\noindent {\it Integral over $|l-l'|\ge\nu$}. 
 If $\tau<\nu$ than this integral vanishes, so we assume that $\tau\geq\nu$.
Since $\oms = 2(s_1-s)\cdot(s-s_2)$ is a non-degenerate quadratic form  (with respect to the variable $z=(s_1-s,s_2-s)\in\R^{2d}$), for any $l,l'$ from the considered domain the integral over  $ds_1ds_2=dz$  in \eqref{I^Del_s} has the form \eqref{I_la} with $\nu:=\nu |l-l'|^{-1}\leq 1$ and $n=2d$. 
In view of \eqref{f_123 est}, estimate
\eqref{stph_est1} together with \eqref{sob} implies that this integral is
 bounded by $C^\#(s) \tau \nu^d |l-l'|^{-d}$. So  
\begin{equation*}
\begin{split}
|I^\Del_s|\leq C^\#(s) \tau  \nu^d\int_0^\tau dl& \int_0^\tau dl'\, |l-l'|^{-d}\chi_{\{|l-l'| \ge\nu\}}\\
&\le C_1^\#(s) \tau^2  \nu^d  \int_\nu^\tau
x^{-d}\,dx\le C_2^\#(s)\tau^2 \nu.
\end{split}
\end{equation*}

We saw that  $|I_s^\Delta|\le C^\#(s)\tau^2 \nu$. This relation and \eqref{f_1} imply \eqref{N21-hat}.

\subsection{The second term of $Q_s$}  \label{ss_3}
To study the term  $2\Re N_s$ we  use the same strategy as above. Namely, 
expressing  in \eqref{NN22}  the function $3\cY_s^{sym}$ via   $\cY_s$, we write $N_s$ as 
$
N_s= N^1_s + 2N_s^2, \; s\in\R^d, 
$
where
\begin{align*}
	N^1_s&=i\,\EE \Big(\bar a_s(\tau) \int_0^\tau e^{-\ga_s(\tau-l)} 
	\cY_s(a,a,\Del a^{(1)};\nu^{-1}l)\, dl\Big),
	\\
	N^2_s&=i\,\EE \Big(\bar a_s(\tau) \int_0^\tau e^{-\ga_s(\tau-l)} 
	\cY_s(\Del a^{(1)},a,a;\nu^{-1}l)\, dl
	\Big).
\end{align*}
We will show that the terms $2\Re N_s^1$ and  $4\Re N_s^2$ 
can be approximated by the second and the third sums from \eqref{ZZ_s}.

{\it Term $N_s^1$.} Let us start with the term $N_s^1$: writing explicitly the function $\cY_s$ and then $\Del \bar a_3^{(1)}$ we get
\begin{align}\non
	N^1_s&=
	i\,L^{-d}\sum_{1,2}\dep \int_0^\tau dl\, e^{-\ga_s(\tau-l)+i\nu^{-1}\oms l} 
	\EE \big(a_1(l)a_2(l) \Del\bar a_3^{(1)}(l)\bar a_s(\tau)\big)
	\\\lbl{N_s-odin}
	&=L^{-2d}\sum_{1,2}\sum_{1',2'}\dep\de'^{1'2'}_{3'3}
	\int_0^\tau dl\,\int_{0}^l dl'\, 
	e^{-\ga_s(\tau-l)+i\nu^{-1}\oms l} e^{-\ga_3(l-l')-i\nu^{-1}\om^{1'2'}_{3'3} l'} 
	\\\non
	&{}\qu \qu \qu \times \EE \big(a_1(l)a_2(l)
	\bar a_{1'}(l')\bar a_{2'}(l') a_{3'}(l')
	\bar a_s(\tau)\big).
\end{align}
By the Wick theorem, we need to take the summation only over $s_{1'},s_{2'},s_{3'}$ satisfying
$s_{1'}=s_1$, $s_{2'}=s_2$, $s_{3'}=s$ or $s_{1'}=s_2$, $s_{2'}=s_1$, $s_{3'}=s$.
Since in the both cases we get 
$\de'^{1'2'}_{3'3}=\dep$
and 
$\om^{1'2'}_{3'3}=\oms$,
we find
\begin{align}\lbl{N221-start}
	N^1_s=
	2&\ssum_{1,2}\dep
	\int_0^\tau dl\,\int_{0}^l dl'\, 
	e^{-\ga_s(\tau-l)-\ga_3(l-l')+i\nu^{-1}\oms (l-l')}  
	\\\non
	&\times
	\EE a_1(l)\bar a_{1}(l')\,
	\EE a_2(l)\bar a_{2}(l')\,
	\EE a_s(l')\bar a_{s}(\tau).
\end{align}
Replacing  in \eqref{N221-start} the processes $a_k^{(0)}(l)$, $a_k^{(0)}(l')$ and $a_k^{(0)}(\tau)$ by their value at zero, 
we get instead of $N_s^1$  the sum
$$
\hat N^1_s:=2 \ssum_{1,2}
\dep
T_{12}\,
n_1^{(0)} n_2^{(0)} n_s^{(0)},
$$
where $T_{12}$ denotes the integral of the exponent above
\be\lbl{T_12}
T_{12}:=
\int_0^\tau dl\,\int_0^l dl'\, 
e^{-\ga_s(\tau-l)-\ga_3(l-l')+i\nu^{-1}\oms (l-l')}.
\ee
Arguing as in Section \ref{ss_2}  we find
\be\lbl{N221-hat}
| N_s^1 -  \hat  N_s^1| \le C^\#(s) (\nu^{-2} L^{-2} + \tau^2\nu).
\ee
The term $2\Re \hat N_s^1$ is not equal to the second sum from \eqref{ZZ_s} yet. 
To extract the latter from the former we
write the integral over $dl'$ in \eqref{T_12} as 
$\int_0^l = \int_{-\infty}^l -\int_{-\infty}^0 $. We get 
\be\lbl{Tbeau}
T_{12}=\cT_{12}+\cT_{12}^r,
\ee
where
\be\non
\cT_{12}=
 e^{-\ga_s\tau} \int_0^\tau dl \, e^{l(\ga_s -\ga_3 +i\nu^{-1}\oms)} \int_{-\infty}^l dl' e^{l'(\ga_3 - i\nu^{-1}\oms)}=
\frac{1-  e^{-\ga_s\tau}} {(\ga_3 -i \nu^{-1}\oms) \ga_s}
\ee
and
\be\lbl{T_12^r}
\cT_{12}^r=-\int_0^\tau dl\,\int_{-\infty}^0 dl'\, 
e^{-\ga_s(\tau-l)-\ga_3(l-l')+i\nu^{-1}\oms (l-l')}.
\ee
Computing the real part of $\cT_{12}$, we find that 
\be\lbl{Re T}
2\Re\cT_{12}=2\, \frac{1-e^{-\ga_s \tau}}{\ga_s}
  \frac{ \ga_3}{\ga_3^2+(\nu^{-1}\om^{12}_{3s})^2}
=\cZ^3,
\ee
where $\cZ^3$ is defined in \eqref{Z-edi}.
On the other hand,
 \be\label{n_esti}
\left| \hat N_s^1 -2 \ssum_{1,2}
	\dep \cT_{12}
	\,
	n_1^{(0)} n_2^{(0)} n_s^{(0)}\right| = 2\left|
	 \ssum_{1,2}
\dep
\cT^r_{12}\,
n_1^{(0)} n_2^{(0)} n_s^{(0)}\right|
 \ee
by \eqref{T_12^r} equals to
 \be\lbl{T_12^r+}
2 \Big| \ssum_{1,2} \dep  
  \int_0^\tau dl\,\int_{-\infty}^0 dl'\,  
  e^{l' + i\nu ^{-1} \oms (l-l')} F_s(s_1,s_2,l,l',\tau)
 \Big|
 \ee
 with $F_s=e^{-\ga_s(\tau-l)-\ga_3 l + (\ga_3-1)l'} \,n_1^{(0)} n_2^{(0)} n_s^{(0)}$.
Since by \eqref{n_s_0} $n_r^{(0)}=n_r^{(0)}(0)$ is a Schwartz function of $r\in\R^d$ and $\ga_3-1\geq 0$, then Theorem~\ref{c_6.2} applies and 
implies that \eqref{T_12^r+}
 is bounded by $ C^\#(s) \big(\nu^2\chi_d(\nu) + \nu^{-2}L^{-2}\big)$, 
 so  the l.h.s. of \eqref{n_esti} also is. 

Thus, due  to  \eqref{N221-hat} and  \eqref{Re T}
 we arrive at the relation
\be\lbl{N221}
\Big| 2\Re N_s^1-
2\ssum_{1,2}
\dep \cZ^3
\,
n_1^{(0)} n_2^{(0)} n_s^{(0)} \Big|
\leq  C^{\#}(s)(\nu^2\chi_d(\nu)+\nu^{-2}L^{-2}+\tau^2\nu).
\ee

{\it Term $N_s^2$.} Finally, we study the term $N_s^2$ by  literally repeating the argument  we have applied to  $N_s^1$. 
We find that
\begin{align}\non
	N_s^2&=
	i\,L^{-d}\sum_{1,2}\dep \int_0^\tau dl\, e^{-\ga_s(\tau-l)+i\nu^{-1}\oms l} 
	\EE \Del \big(a_1^{(1)}(l)a_2(l)\bar a_3(l)\bar a_s(\tau)\big)
	\\\lbl{N_s-dva}
	&=-L^{-2d}\sum_{1,2}\sum_{1',2'}\dep\de'^{1'2'}_{3'1}
	\int_0^\tau dl\,\int_0^l dl'\, 
	e^{-\ga_s(\tau-l)+i\nu^{-1}\oms l} e^{-\ga_1(l-l')+i\nu^{-1}\om^{1'2'}_{3'1} l'} 
	\\\non
	& \qu\times\EE \big(a_{1'}(l')a_{2'}(l')
	\bar a_{3'}(l')a_{2}(l) \bar a_{3}(l)
	\bar a_s(\tau)\big).
\end{align}
By the Wick theorem we should  take  summation either 
under the condition
$s_{1'}=s_3$, $s_{2'}=s$, $s_{3'}=s_2$ or $s_{1'}=s$, $s_{2'}=s_3$, $s_{3'}=s_2$.
Since in both cases 
$\de'^{1'2'}_{3'1}=\dep$
and 
$\om^{1'2'}_{3'1}=-\oms$,
then 
\begin{align}\non
	N_s^2=
	-2&\ssum_{1,2}\dep
	\int_0^\tau dl\,\int_{0}^l dl'\, 
	e^{-\ga_s(\tau-l)-\ga_1(l-l')+i\nu^{-1}\oms (l-l')} 
	\\\lbl{N222-start}
	&\times
	\EE a_2(l)\bar a_{2}(l')\,
	\EE a_3(l')\bar a_{3}(l)\,
	\EE a_s(l')\bar a_{s}(\tau).
\end{align}
We set 
$$
M_{12}:=
\int_0^\tau dl\,\int_0^l dl'\, 
e^{-\ga_s(\tau-l)-\ga_1(l-l')+i\nu^{-1}\oms (l-l')}
$$
  and note that $M_{12}$ equals to $T_{12}$, defined in \eqref{T_12}, if replace $\ga_3$ by $\ga_1$.
Then, as in \eqref{Tbeau}-\eqref{Re T} we get $M_{12}=\cM_{12}+\cM^r_{12}$, where
 $2\Re \cM_{12}=\cZ^1$.
 Arguing again as in Section~\ref{ss_2}, we replace in \eqref{N222-start} the processes $a_k^{(0)}(l)$, $a_k^{(0)}(l')$ and $a_k^{(0)}(\tau)$ by their value at zero $a_k^{(0)}(0)$. Next, using Theorem~\ref{c_6.2} we show that the input to the resulting sum of the term corresponding to $\cM^r_{12}$ is small. 
 Finally, similarly to \eqref{N221}, we get
\begin{comment}
and write $\int_0^l = \int_{-\infty}^l - \int_{-\infty}^0$, so that 
$M_{12}=	\cM_{12}+\cM^r_{12}$, where
$$
\cM_{12}=e^{-\ga_s \tau} \int_0^\tau dl e^{l(\ga_s  -\ga_1 +i\nu^{-1} \oms)}
\int_0^l dl' e^{l'(\ga_1-i\nu^{-1} \oms)}=
\frac{ 1- e^{-\ga_s\tau} } {\ga_s(\ga_1 -i\nu^{-1} \oms)}
$$
and
$$
\cM^r_{12}=-\int_0^\tau dl\,\int_{-\infty}^0 dl'\, 
e^{-\ga_s(\tau-l)-\ga_1(l-l')+i\nu^{-1}\oms (l-l')}.
$$
Arguing as in Section~\ref{ss_2}, we replace in \eqref{N222-start} the processes $a_k^{(0)}(l)$, $a_k^{(0)}(l')$ and $a_k^{(0)}(\tau)$ by their value at zero $a_k^{(0)}(0)$. Next, using Theorem~\ref{c_6.2} we show that the input to the resulting sum of the term corresponding to $\cM^r_{12}$ is small, 
and  note that $2\Re \cM_{12}=\cZ^1$.
Thus, similarly to \eqref{N221}, 
\end{comment}
\be\non
\Big| 4\Re N_s^2
+4 \ssum_{1,2}
\dep \cZ^1
\,
n_2^{(0)} n_3^{(0)}n_s^{(0)}\Big| 
\leq  C^{\#}(s)(\nu^2\chi_d(\nu)+\nu^{-2}L^{-2}+\tau^2\nu),
\ee
 where the sign "+"  in the l.h.s. is due to the sign "-"  in \eqref{N222-start} (cf. \eqref{N221-start}).

\subsection{The last two terms of $Q_s$}\label{ss_1}

In Sections~\ref{ss_2} and \ref{ss_3} we have seen that the first two terms of $Q_s$ approximate the sum $\cZ_s$. So to get assertion of the proposition it suffices to show that
\be\lbl{N^r-est}
|\EE \Del a^{(1)}_s(\tau) \bar c^{(1)}_s (\tau) |,\,|\wt N_s|\leq  C^\#(s) \big(\nu^2\chi_d(\nu) + \nu^{-2}L^{-2}\big).
\ee
We have
$$
\EE \Del a^{(1)}_s \bar c^{(1)}_s(\tau)
=\EE \int_{0}^\tau e^{-\ga_s(\tau-l)} \cY_s(a,\nu^{-1}l)\,dl 
\int_{-T}^0 e^{-\ga_s(\tau-l')} \ov{\cY_s(a,\nu^{-1}l')}\,dl'.
$$
This expression coincides with \eqref{E Del a} in which the integral $\int_0^\tau dl'$ is replaced by  $\int_{-T}^0 dl'$.
Then, $\EE \Del a^{(1)}_s \bar c^{(1)}_s(\tau)$ has the form \eqref{N^2,1_wick} where the same replacement is done.
Since the correlations $\EE a_j(l)\bar a_j(l')$ are given by \eqref{corr_a_in_time}, Theorem~\ref{c_6.2} applies  (see a discussion after its formulation)
 and we get the first inequality from \eqref{N^r-est}.

Expressing the function $\cY^{sym}_s$ through $\cY_s$, for the term $\wt N_s$  we find
$\wt N_s=\wt N^1_s+2\wt N^2_s$,
where
$$
\wt N^1_s=i\EE \Big(\bar a_s(\tau) \int_0^\tau e^{-\ga_s(\tau-l)} 
\cY_s(a,a, c^{(1)};\nu^{-1}l)\, dl\Big)
$$
and 
$$
\wt N^2_s=i\EE \Big(\bar a_s(\tau) \int_0^\tau e^{-\ga_s(\tau-l)} 
\cY_s(c^{(1)},a,a;\nu^{-1}l)\, dl
\Big).
$$
Expressing $c^{(1)}$ through $a^{(0)}$, we see that the terms $\wt N^1_s$ and $\wt N^2_s$
have the forms \eqref{N_s-odin} and \eqref{N_s-dva} correspondingly, where the integral $\int_0^ldl'$ is replaced by $\int_{-T}^0dl'$. 
Then the second inequality from \eqref{N^r-est} again follows from Theorem~\ref{c_6.2}.
\qed

\section{Energy spectra of quasisolutions and  wave kinetic equation}\label{s_WKE}

Everywhere in this section in addition to \eqref{r_scale} we assume that 
\be\label{L_assump}
L\ge  \nu^{-2-\epsilon}, \qquad \epsilon>0.
\ee
Let us consider equation \eqref{w_k_e} with $f(s) =2{b}(s)^2$ for $\tau\ge-T$:
\be\label{WKE}
\dot \zz(\tau ,s) = -\Lc \zz+ \eps K(\zz) + 2{b}(s)^2, \quad s \in \R^d, 
\ee 
 with the initial condition
\be\label{WK}
 \zz(-T) =0.
\ee
 Fix any $r_0>d$. Denoting % $r_0= \max(2r_*+1, d+1)$ and
$
\eps(b) = \eps_*\big(2\| b^2(s)\|_{ r_0},  r_0 \big)>0,
$
where $\eps_*$ is the constant from Theorem \ref{t_kin_eq}, we get from Corollary~\ref{c_kin_eq} that for %$r\ge 2r_*+d$ and 
$0\le \eps\le\eps(b)$ a  solution $\zz$ of \eqref{WKE}, \eqref{WK}  exists, is unique and $\zz\in X_r$ for each $r$,
	so that
 \be\lbl{bound2-r}
 |\zz(\tau)|_r\leq C_r, \quad \forall  r,
 \ee 
 uniformly in $\tau$.
Our goal is to compare $n_s^{\le2}(\tau) $ (extended to a function on  $ \R^d =\{s\}$) 
 with  $\zz$ in the spaces $\Cc_r(\R^d)$, 
 and we recall that, due to  
\eqref{n_est}, 
 \be\lbl{bound1-r}
 |n^{\leq 2}(\tau)|_r\leq C_r \quad \forall r,
 \ee 
 uniformly in $\nu, L, \tau$.
 The constants $C_r$ below vary from formula to formula and we often 
skip the dependence on $r$ writing simply $C$.

Since both curves  $n^{\le2}$ and   $\zz$
vanish at $-T$, then their difference $w= n^{\le2} - \zz$ also does. Let us estimate 
 the increments of $w$.

\begin{proposition}\label{p_increment_esti} If $\rho, \nu, L$ satisfy  \eqref{r_scale},
\eqref{L_assump},  $r>d$  and    $\eps \le C_{1r}^{-1} \le \eps(b)$,  then 
\be\label{z9}
|w(\tau' + \tau)
|_r \le (1-\tau/2) |w(\tau')|_r    +C_{2r}\tau W \quad \forall \tau'\ge -T, \  0\le \tau\le1/2,
\ee
where $W= \eps\big( \tau +\eps  +
 \tau^{-1} \nu^{1-\aleph_d} 
	+ \tau^{-1}\nu^{-3}L^{-2}\big)$
			and $\aleph_d$ is defined as in Theorem~\ref{t_singint}. The constants $C_{1r},\, C_{2r}$ 
 do not depend on $\tau,\tau',T$ and 
 	$\nu, L, \eps$,	 but $C_{2r}$ depends on $\aleph_d$.
 \end{proposition} 

\begin{proof}
The calculation below does not depend on $\tau'\ge -T$ and to simplify presentation we take $\tau'=0$. Then 
$$
\zz(\tau) = \etl \zz(0) +2 \int_0^\tau e^{-t\Lc} {b}^2\,dt + \eps   \int_0^\tau \etlt K(\zz(t))\,dt.
$$
From here  and  
\eqref{n_increment} we have
\begin{equation}\label{00}
\begin{split}
w(\tau) &= \etl w(0) + \eps\Delta +
%\eps \\& +\eps \left[K^\tau(n^{(0)} (0)) - K^\tau (n^{\le2} (0)) \right] 
\Rc, %\quad \Delta =  \left[K^\tau(n^{\le2} (0)) - \int_0^\tau \etlt K(\zz(t))\,dt\right].
\end{split}
\end{equation}
where $ \Delta =  \left[K^\tau(n^{\le2} (0)) - \int_0^\tau \etlt K(\zz(t))\,dt\right].$
Now we will  estimate $\Delta$. 
By \eqref{I_tau}, it is made by the sum of four terms
$$
\ \int_0^{\tau} e^{-t\Lc} K_j(n^{\le2}(0))\,dt  - \int_0^\tau \etlt K_j (\zz(t))\,dt, \quad 1\le j\le4.
$$
 Their estimating is very similar. Consider, for example, the term 
with $j=1$ and write it as 
\begin{equation*}
\begin{split}
& \int_0^\tau e^{-(\tau-t)\Lc} \big( K_1(\nd) - K_1( \zz(0))\big)\,dt\\
&+  \int_0^\tau e^{-(\tau-t)\Lc} \big( K_1(\zz(0)) - K_1( \zz(t))\big)\,dt =: \Delta^1 +\Delta^2. 
\end{split}
\end{equation*}
In view of Corollary \ref{c_7.1} and \eqref{normL}, 
$|\Delta^1|_{r} \le C\tau  |w(0)|_{r}$,
where  $C$ depends on the norms 
$|n^{\leq 2}(0)|_{r}$ and $ |\zz(0)|_{r}$. 
By \eqref{bound1-r} and \eqref{bound2-r} 
the two 
norms are bounded by constants,  so the constant $C=C_r$ is absolute.
Consider $\Delta^2$. 
From \eqref{bound2-r} and the equation on $\zz$ 
 we have  that $|\dot\zz|_r \le C_r$,
  so  $|\zz(0) - \zz(\tau)|_r \le C_r\tau$. Hence, by Corollary~\ref{c_7.1} and \eqref{normL},    $|\Delta^2|_r \le C_r\tau^2$. 
 
 We have seen that 
\be\label{z6}
|\Delta |_r \le C_{1r}' \tau (\tau + |w(0)|_r). 
\ee
Since $0\leq \tau\leq 1/2$, then by \eqref{normL} 
we  have $|e^{-\tau\Lc}w(0)|_r \leq (1-\tau)|w(0)|_r$.
So  using the bound \eqref{R_s} for $|\Rc|_r$, 
by \eqref{00} we get  
$$
|w(\tau)|_r \le (1-\tau) |w(0)|_r  +C_{1r}' \eps\tau |w(0)|_r  +C_{2r}\tau W
$$
with $W$ 	as in the proposition. This relation implies the assertion with  
$
C_{1r}^{-1} = \min({(2C_{1r}')}^{-1}, \eps(b))
$.
\end{proof}

Denote $\tilde w(\tau) = e^{-\tau\Lc} w(0)$. Then $\frac{d}{d\tau} \tilde w(\tau) =-\Lc \tilde w(\tau)$, so 
$\ 
|\tilde w(\tau)- w(0)|_r \le \int_0^\tau |\Lc \wt w(t)|_r\,dt\leq C\tau, 
$
where in the latter inequality we  used 
	\eqref{bound2-r} and \eqref{bound1-r}.
Note also that, due to \eqref{R_s}, 
$
|\Rc|_r\leq C\tau
$
provided that  \eqref{L_assump} holds and  $\nu^{1-\aleph_d}\leq C_1\tau$.
These relations, equality \eqref{00}  and estimate \eqref{z6} on the term $\Delta$  imply

\begin{corollary}\label{c_easy_esti}
 	Assume that  $\nu^{1-\aleph_d}\leq C\tau$ 
 for some constant $C>0$. Then, under the assumptions of Proposition \ref{p_increment_esti}, 
 \be\label{z10}
 |w(\tau' +\tau) -w(\tau')|_r \le 
  C_{r,\aleph_d}\tau , \quad \forall\, \tau'\ge-T. 
 \ee
 \end{corollary}

 Now we state and prove the main result of our work.  Let us again   fix any $r_0>d$ and set 
 $$
 \eps_r= C^{-1}_{1\, \max(r,  r_0)}, 
% \left\{  \begin{array}{cl}  C_{1r}^{-1} \qu\mbox{if}\qu r\geq r_0, \\
% C_{1r_0}^{-1} \qmb{if}\qu r\le r_0, \end{array}\right.
 $$ 
  where $C_{1r}$ is the constant from Proposition~\ref{p_increment_esti}.
  We recall that the energy spectrum $n_s(\tau)=n_s(\tau; \nu, L)$ of a quasisolution $A(\tau)$ 
 naturally extends to a Schwartz function of $s$, see \eqref{schw_ext}. 
 
 \begin{theorem}\label{t_principal}
 Let $\nu$ and $L$ satisfy \eqref{L_assump}, let 
  $A(\tau)$ be the quasisolution \eqref{2nd trunc},
 corresponding to $\rho^2 =\eps\nu^{-1}$,  let  $n_s(\tau)=\EE|A_s(\tau)|^2$ be its energy spectrum and 
  $\zz(\tau,s)$ be the  (unique) solution of \eqref{WKE}, \eqref{WK}.   Then for any $r$ and for  $0<\eps  \le \eps_r$ 
   there exists  $\nu_\eps(r)>0$ such that for $\nu\le\nu_\eps(r)$ we have 
 \be\label{main_est}
 | n(\tau) - \zz(\tau)|_r \le C_r \eps^2\quad \forall\, \tau\ge -T, \quad \forall\, r. 
 \ee
  The constant  $ C_r$ does not depend on  $\tau$ and $ T$. 
  \end{theorem}
  Note that $n(\tau)$ has the form \eqref{n_s-decomp}, where $n^{(1)}=0$ and the first nontrivial term $\rho^2 n^{(2)}$ is of
  order $\eps$, which is much bigger than the r.h.s. of \eqref{main_est} if $\eps\ll1$. 
   Similarly, $\zz(\tau)$ has the form \eqref{u_present}, where the first nontrivial term $\eps u^1$ is also of the size $\eps$.

  \begin{proof}
  		Since $|\cdot|_{r'}\leq |\cdot|_r$ for $r'\leq r$, then estimate \eqref{main_est} for $r<  r_0$ follows from \eqref{main_est} with $r= r_0	$. Assume now that $r\geq r_0$.
  Since $w(t) = n^{\le2}(t) - \zz(t)$, then in  view of \eqref{n_s^3-est} it suffices to establish that 
  \be\label{z20}
  | w(\tau')|_r \le C\eps^2 \quad \forall\, \tau'\ge -T
  \ee
  (we assume $\nu_\eps\ll1$). Let us 	fix some $0<\aleph_d\ll1$  and any time-step    $\tau$, satisfying
  \be\lbl{ass_tau}
  C_0\eps^{-1}\nu^{1-\aleph_d} 
  \leq \tau\leq C^{-1}_0\eps^2,
  \ee 
  for a sufficiently large constant $C_0>0$.  
  We claim that 
  \be\label{z19}
w_n:=  | w(-T +n\tau)|_r \le C C_{2r} \eps^2 \quad \forall n\in \N\cup\{0\}, 
  \ee
  where  $C_{2r}=C_{2r,\aleph_d}$ 
   is the constant from Proposition \ref{p_increment_esti}.  Indeed, let us fix any $N\in\N$ and let $w_n$, $n=0, \dots, N$, attains 
   its maximum at a point $n$, which we write as $n=k_0+1$. If $k_0+1=0$, then $w_k \equiv 0$. Otherwise, in view of \eqref{z9},
   $$
   w_{k_0} \le  w_{k_0+1} \le (1-\tau/2)  w_{k_0}  +C_{2r} \tau W;
   $$
   so $ w_{k_0} \le 2 C_{2r}  W$. From here 
   $\ 
   \max_{0\le k\le N}  | w(-T +n\tau)|_r  = w_{k_0+1}  \le 3 C_{2r}  W
   $
   for any $N$, and \eqref{z19} follows in view of \eqref{ass_tau} and \eqref{L_assump}.

  Since for any $\tau'>-T$ we can find an  $n$ such that $\tau' \in [ -T+n\tau, -T+ (n+1)\tau]$, then \eqref{z20} follows from 
  \eqref{z19} and \eqref{z10}. 
   \end{proof}

   By Proposition \ref{L to infty},  when $L\to\infty$ and $\nu$
   stays fixed,    the energy spectrum $n_s(\tau; \nu,L)$ admits a limit $n_s(\tau; \nu,\infty)$,  which is a Schwartz function of $s\in\R^d$.   Since  estimate \eqref{main_est} is uniform in $\nu, L$, we immediately get 
   
   \begin{corollary}\lbl{main L-to-infty}
   For $0<\eps \le \eps_r$, $\rho =\eps^{1/2} \nu^{-1/2}$ and $\nu\le\nu_\eps(r)$ ($\nu_\eps(r)$ as in Theorem~\ref{t_principal}),
   the limiting energy spectrum  $n_s(\tau; \nu,\infty)$ satisfies estimate \eqref{main_est}. 
  \end{corollary}
	
Jointly with Proposition \ref{p_stab}, Theorem \ref{t_principal} implies exponential convergence of $n(\tau)$ to an 
equilibrium, modulo $\eps^2$: 

 \begin{corollary}\lbl{c_stab}
 For $r> d$ there 
  exists $\eps' =\eps'( r, |b^2|_r)$ such that if $0<\eps \le \min(\eps', C_{1r}^{-1})$, then eq.~\eqref{WKE}  has a unique 
 steady state $\zz^\eps$ close to $\zz^0 = 2\Lc^{-1} b^2$, and 
 $$
 | n(\tau) - \zz^\eps |_r \le e^{-\tau-T} | \zz^\eps |_r +C_r \eps^2,\qquad \forall \tau\ge-T.
 $$
   \end{corollary}

In view of  Corollary~\ref{c:a-goa}, Theorem C from the introduction follows from Theorem~\ref{t_principal} and Corollary~\ref{main L-to-infty}. Similarly, the asymptotic \eqref{time_ass}  follows from Corollary~\ref{c_stab}. 
   
   \smallskip
   
   By \eqref{u_present}, the solution $\zz$ of \eqref{WKE}, \eqref{WK}  may be written as $\zz =\zz^0+ \eps\zz^1 +O(\eps^2)$, where 
   $$
   \dot\zz^0 =-\Lc \zz^0 +2 b(s)^2, \qquad \zz^0(-T)=0,
   $$
    $$
   \dot\zz^1 =-\Lc \zz^1 +K(\zz^0) , \qquad \zz^1(-T)=0. 
   $$
  From the first equation we see that $\zz^0(s) = n_s^{(0)}$ (see \eqref{n_s_0}), and from the second -- that 
  $$
  \zz^1(\tau) = \int_{-T}^\tau e^{-(\tau-l)\Lc} K( n_s^{(0)}(l))\,dl. 
  $$
Since, by \eqref{n_s^3-est}, $n(\tau) =  n^{(0)}(\tau)  +\eps (\nu^{-1}   n^{(2)}(\tau)) + O(\eps^2)$, then \eqref{main_est} implies that 
\be\label{n2_form}
 n^{(2)}(\tau) =\nu  \int_{-T}^\tau e^{-(\tau-l)\Lc} K( n_s^{(0)}(l))\,dl + O( \nu \eps), 
\ee
where $|O( \nu\eps)|_r \le \nu \eps C_r$ for every $r$. Other way round,  now, after the exact form of the operator $K$ in 
\eqref{WKE}  is established, the validity of the presentation \eqref{n2_form} (obtained by some direct calculation), jointly
with estimate \eqref{n_s^3-est} would imply Theorem~\ref{t_principal}.

\section{Energy spectra of solutions  \eqref{decomp}}
\lbl{s_en_spectra}

Results of Sections \ref{sec:quasisol} and \ref{s_WKE} concerning the quasisolutions 
suggest a natural question if they extend to higher order truncations of  the complete decomposition \eqref{decomp}  in series in $\rho$.
 In this section we  discuss the corresponding positive and 
negative results, obtained in \cite{DK}. For a complex number $z$ we denote by $z^*$ either $z$ or $\bar z$.

Firstly let us return to Section \ref{s_series}. Iterating the  Duhamel integral in the r.h.s. of \eqref{an}  and expressing there iteratively $a^{(n_j)}(l)$ with $1\le n_j <n$ via
integrals \eqref{an} with $n:= n_j$,  we will eventually represent each $a_s^{(n)}(\tau)$ as a sum of 
iterated integrals of the form
\be\label{I1}
J_s(\cT) = J_s(\tau; n,\cT) =\int\dots\int dl_1\dots dl_n L^{-nd}  \sum_{s_1, \dots,  s_{3n} } (\dots).
\ee
The zone of integrating in $l=(l_1,\ldots,l_n)$ is a convex polyhedron in the cube $[-T, \tau]^n$, and the summation is taken over
the vectors  $(s_1, \dots, s_{3n}) \in \big(\Z^d_L\big)^{3n}$ subject to certain linear relations which follow from the 
factor $\dess$ in the definition \eqref{ku4} of $\cY_s$. The summand in  brackets in \eqref{I1} is a monomial of exponents 
$e^{-\ga_{s'}(l_k-l_j)}$, $e^{\pm i\nu^{-1} \omega^{s'_1 s'_2}_{s'_3 s'_4}}$ and the processes $\big(a^{(0)}_{s''}\big)^*(l_r)$,
 which has degree $2n+1$ with respect to the processes. Each integral $J_s(\tau; n,\cT)$ 
 corresponds to an oriented rooted tree $\cT$ from a class $\Ga(n)$ of trees with the root at $a_s^{(n)}(\tau)$, with random variables 
  $\big( a^{(0)}_{s''}\big)^* (l_r)$ at its leaves  and with vertices labelled by $\big(a^{(n')}_{s'} \big)^* (l_r')$ ($1\le n'<n$).   To any vertex  
  $\big(a^{(n')}_{s'} \big)^* (l_r')$  enters one edge of the tree and
  three edges outgo from it to the vertices or leaves, corresponding to some three specific terms 
   $\big(a^{(n_1)} \big)^* $,   $\big(a^{(n_2)} \big)^* $,  $\big(a^{(n_3)} \big)^* $ in the decomposition \eqref{an} of 
    $\big(a^{(n')}_{s'} \big)^* (l_r')$.
So 
\be\label{A_formula}
a_s^{(n)}(\tau) = \sum _{\cT\in \Gamma(n)} J_s(\tau; n, \cT).
\ee
\smallskip

 Now let us consider the formal series 
\be\label{N_sa}
N_s(\tau) =n_s^0(\tau) + \rho n_s^1(\tau) + \rho^2 n_s^2(\tau)+\dots 
\ee
 for the energy spectrum $N_s=\EE|a_s|^2$  of  a solution $a_s(\tau)$, when the latter is 
  written as the  formal series \eqref{decomp}.
There $n_s^0\sim1$, $n_s^1 =0$ and $n_s^2$ are the same as  $n_s^{(0)}$, $n_s^{(1)}$ and $n_s^{(2)}$ in the decomposition \eqref{n_s-decomp} of quasisolutions, but
$n^3_s$  and $n^4_s$ are different. This small ambiguity should not cause a problem, and we will see below that the new 
$n^3_s$  and $n^4_s$ still meet the estimates \eqref{n_s_est}. Let us consider any $n_s^k (\tau)$. It equals 
\be\label{k_formula}
n_s^k(\tau) = \EE \sum_{k_1+k_2=k} a_s^{(k_1)}(\tau) \bar a_s^{(k_2)}(\tau). 
\ee
We analyse each expectation $\EE a_s^{(k_1)} \bar a_s^{(k_2)}$ separately. Due to \eqref{A_formula}, 
\be\lbl{a corr}
\EE a_s^{(k_1)}(\tau)  \bar  a_s^{(k_2 )}(\tau) = \sum_{\cT_1 \in \Gamma(k_1), \cT_2 \in \Gamma(k_2)} 
\EE J_s(\tau; k_1, \cT_1) \overline{ J_s(\tau; k_2, \cT_2)},
\ee
with 
\be\label{I2}
\EE J_{s}(\tau; k_1,\cT_1) \overline{J_{s}(\tau; k_2,\cT_2)}
 = \int\dots\int dl_1\dots dl_k  L^{-kd}\sum_{s_1, \dots, s_{3k}}  \EE (\dots),
\ee
where $k=k_1+k_2$ and under the expectation sign stands a product of the terms in the brackets in the presentations \eqref{I1} for 
$ J_{s}(\tau; k_1,\cT_1) $ and $\overline{ J_{s}(\tau; k_2,\cT_2)}$.

Since  $a_s^{(0)}(l)$ are Gaussian random variables whose correlations are given by \eqref{corr_a_in_time}, then by the Wick
theorem (see \cite{Jan}) each expectation \eqref{I2} is a finite sum over different Wick pairings between the conjugated and 
non-conjugated variables  $a_{s'}^{(0)}(l'_r)$ and $\bar a_{s''}^{(0)}(l''_r)$, 
labelling the leaves of $\cT_1 \cup \overline{\cT_2}$. 
 By \eqref{corr_a_in_time}, for each Wick pairing, in the sum $\sum_{s_1,\ldots,s_{3k}}$ from \eqref{I2} only those summands do not vanish for which indices $s'$ and $s''$ of the Wick paired variables are equal, $s'=s''$. So we take the sum only over vectors $s_1,\ldots,s_{3k}$ satisfying this restriction.
	The Wick pairings in \eqref{I2} can be parametrised by Feynman diagrams $\cF$ from a class ${\frak F}(k_1, k_2)$ of 
diagrams, 
obtained from the union $\cT_1\cup\ov\cT_2$ of various trees $\cT_1 \in \Ga(k_1)$ and $\overline{\cT_2} \in \ov{\Ga(k_2)}$ by all paring of conjugated leaves with non-conjugated leaves in  $\cT_1\cup\ov\cT_2$. We denote the summands, forming the r.h.s. of 
\eqref{I2}, by $I_s(\cF)$, where $\cF\in \fF(k_1, k_2)$, and accordingly  write \eqref{a corr} as 
\be\lbl{Mom a-a}
\EE a_s^{(k_1)}(\tau)  \bar  a_s^{(k_2 )}(\tau)  =  \sum_{\cF \in\fF(k_1, k_2)} I_s(\cF),\qquad
 I_s(\cF)=I_s(\tau;k_1, k_2,\cF). 
\ee
Then, \eqref{k_formula} takes the form
\be\label{I3}
n_s(\tau)  =  \sum_{k_1+k_2=k}\sum_{\cF \in\fF(k_1, k_2)} I_s(\cF). 
\ee
See \cite{DK} for a detailed explanation of the formulas \eqref{I1}--\eqref{I3}. 

Resolving all the restrictions on the indices $s_1, \dots, s_{3k}$ in \eqref{I2} which follow from the rules, used to construct the
trees and the diagrams, we find that among those  indices  exactly $k$ are independent. Suitably parametrizing them by
vectors $z_1, \dots, z_k \in \Z^d_L$ we write the sum in \eqref{I2} as 
$
 \ssum_{z_1, \dots, z_k \in \Z^d_L}. 
$
Approximating the latter by an integral over $\R^{kd}$ using Theorem~\ref{t_sum_integral},  we get for integrals $I_s(\cF)$
with $\cF\in \fF(k_1, k_2)$ an explicit formula: 
\begin{equation}\lbl{I_F}
	I_s(\cF)
=\int_{\R^k}dl \, \int_{\R^{kd}}dz \,F_{\cF}(\tau,s,l,z) e^{i\nu^{-1}\sum_{i,j=1}^k \al_{ij}^\cF(l_i-l_j)z_i\cdot z_j} + O(L^{-2}\nu^{-2}),
\end{equation}
where $k=k_1+k_2$, $z=(z_1,\ldots,z_k)$, the function $F_\cF$ is smooth in $s,z$ and fast decays in $s,z$ and $l$,
while $\al^\cF=(\al^\cF_{ij})$ is a skew-symmetric matrix without zero lines and rows. Its rank is $\ge2$ and for any $k$ 
it may be equal to 2. 
 In particular, \eqref{Mom a-a} together with \eqref{I_F} implies Proposition~\ref{L to infty} since the integrals in
 \eqref{I_F} are independent from $L$ and are Schwartz functions of $s$.

Now let us go back to the series \eqref{N_sa}. We know that $n_s^0\sim1$, $n_s^1=0$ and $n_s^2\sim C^\#(s)\nu$. 
By a direct (but long) calculation, 
similar to that in Appendix~\ref{a_stat_phase}, it is possible to verify that 
$
|n^3_s|\leq C^\#(s) \nu^2\chi_d(\nu)\leq C^\#(s) \nu^{3/2}$ and  $|n^4_s|\leq C^\#(s) \nu^2.$  This
 suggests to assume that 
\be\lbl{sharp est n^k}
|n_s^k|\leq C^{\#}(s;k)\nu^{k/2} \qmb{for any}\qu k.
\ee
If this is the case, then under the "kinetic" scaling $\rho =\nu^{-1/2} \eps^{1/2}$ the series 
\eqref{N_sa} becomes  a formal series in $\sqrt\eps$, uniformly in $\nu$ and $L$, and  its truncation of any order $ m\ge2$ in $\rho$  still is 
$O(\eps^2)$--close to the solution ${\frak z}(\tau)$ of the wave kinetic equation \eqref{WKE}, \eqref{WK}. 
On the contrary, if this is not the case in the sense that $\|n_\cdot^k\|\geq C_k\nu^{k'}$ with $k'<k/2$ for some $k$, 
 then \eqref{N_sa} even is not a formal series, uniformly in $\nu,L$  under the kinetic scaling above.

Analysing formula \eqref{I_F}, we obtain  estimate 
\eqref{main} of Theorem~\ref{l:est_a^ia^j} for every integral $I_s(\cF)$, which implies estimate \eqref{main} itself (see Theorem 1.2 in \cite{DK}).
 By \eqref{k_formula}, estimate \eqref{main} implies 
\begin{theorem}\label{t2}
	For each $k$, 
	\be\label{A_est}
	|n_s^k|\le C^\#(s;k) \max(\nu^{\lceil k/2\rceil }, \nu^d)\, \chi^k_d(\nu),
	\ee
	provided that $L$ is so big that $C^\#(s;k) L^{-2}\nu^{-2}$, 	 is smaller than the r.h.s. of \eqref{A_est}.
\end{theorem}
Below in this section 
we always  assume that $L$ is as big as it is required in  the theorem. 
The theorem implies estimate \eqref{sharp est n^k} only for $k\leq 2d$.
In particular, for $k\leq 4$ since $d\geq 2$. 

To get estimate \eqref{A_est}, we establish it for each integral $I_s(\cF)$ separately.
Our next result shows that estimate \eqref{sharp est n^k}  can not be obtained by improving  \eqref{A_est} for every integral 
$I_s(\cF)$, since inequality \eqref{A_est} is sharp for some of the integrals.

  Let $\gF^B(k)$  be the set of Feynman diagrams in $\cup_{k_1+k_2=k} \fF(k_1, k_2)$ for which 
   the matrix $\al^\cF$ from \eqref{I_F} satisfies $\al_{ij}^\cF=0$ if
    $i\neq p$ and $j\neq p$  (so, only the $p$-th line and column of the matrix $\al^\cF$  are non-zero).  
 This set  is not empty.

\begin{proposition}\label{t3}
 If $k > 2d$, then for any $\cF\in\gF^B(k)$ the corresponding integrals $I_s(\cF) $ satisfy 
 $
 I_s(\cF) \sim C^{\#}(s;k) \nu^d.
 $
\end{proposition}

Proposition \ref{t3} shows that the estimate \eqref{sharp est n^k} can be true for $k>2d$ only if in the sum \eqref{I3} the large terms 
cancel each other. 
 And indeed, we observe strong cancellations among the integrals from the set $\gF^B(k)$:

\begin{proposition}\lbl{l:cancellation} 
	For any $k$,
	$
	\big|\sum_{\cF\in\gF^B(k)} I_s(\cF)\big| \leq C^{\#}(s;k)\nu^{k-1}.
	$
\end{proposition}

Since for $k\geq 2$ we have $\nu^{k-1}\leq \nu^{k/2}$, the  estimate in the proposition 
 agrees with  \eqref{sharp est n^k}.
In the proofs of Propositions~\ref{t3} and \ref{l:cancellation} we use special structure of the set $\gF^B(k)$ and we do not have similar results for a larger set of diagrams, nevertheless we find it plausible 
that \eqref{sharp est n^k} is true.
Namely,  that the decomposition of  integrals $I_s(\cF)$ in  asymptotic sum in $\nu$ is such that a few main order terms of the decomposition
 of the sum \eqref{I3}   vanish due to a cancellation, so 
  that  \eqref{sharp est n^k} holds  (see Problem~\ref{c_1}).
We understand the mechanism of this cancellation, but do not know if  it  goes till the order $\nu^{k/2}$.

\section{Proof of Theorem \ref{t_singint}}
\lbl{s_singint}
	\subsection {Vicinity of the point $(s,s)$. }\label{s_71}
Off the quadric  ${}^s\!\Sigma$ the  integral $I_s$ is small -- of order $\nu^2$ -- and the main task is 
to  examine it in the vicinity of ${}^s\!\Sigma$.  First
we will study $I_s$ near the locus  $(s,s)$, and next -- near the smooth part 
$ {}^s\!\Sigma_*$.

Passing to the variables $(x,y,s)=(z,s)$ (see \eqref{p01}) we write $F_s$ and $ \Gamma_s$ as $F_s(z)$ and $\Gamma_s(z)$. 
The functions still satisfy \eqref{B_condition}  and \eqref{Ga_condition} with $(s_1,s_2)$ replaced by $(x,y)$ since 
the map $(s_1, s_2, s) \mapsto(x,y,s)$ is a linear isomorphism.

Consider the domain
\be\label{K_delta}
K_\delta = \{|x|\le\delta, |y|\le\delta\}\subset \R^d\times \R^d,\qquad 0<\delta\le 1\,,
\ee
and the integral $  \lan I_s, {K_\delta}\ran$ (see {\it Notation}, formula \eqref{<I,M>}): 
\be\label{I_delta}
\lan I_s, {K_\delta}\ran = \nu^2 \int_{K_\delta} \frac{ F_s(z)\,dz}{(x\cdot y)^2 +(\nu \Gamma_s(z))^2}\,.
\ee
Obviously, everywhere in $K_\delta$,
$
|F_s(z)|\le C^\#(s)$ and $  \Gamma_s(z) \ge K^{-1}.
$
 Then, since the volume of $K_\delta$ is bounded by $C\de^{2d}$, 
\be\lbl{nearsing}
\big|\lan I_s, K_\delta\ran \big|  \le C^\#(s)\de^{2d}.
\ee

Next we pass to the global study of integral $I_s$,  written  similar to \eqref{I_delta} as
\be\label{p1}
I_s=
\nu^2 \int_{\R^{2d}} \frac{F_s(z)\,dz}{(x\cdot y)^2 +\big( \nu\Gamma_s(z)\big)^2}\,.
\ee
%where % $\Gamma_s$ and $F_s$ are functions $\Gs$ and $F(s_1,s_2,s_3,s)$, written in the variables $(z,s)$:
%$$ \Gamma_s(z)=\Gs,\quad F_s(z) = \Gs\, \frac2{\ga_s} \bg\,. $$
%We recall the relation \eqref{p2} which implies that $F_s(z)$ is a Schwartz function of $(z,s)$.

\subsection{The manifold $\Sigma_*$ and its vicinity.}\label{s_13.3} 
We denote by $\Sigma$ the quadric $\Sigma=\{ z=(x,y): x\cdot y=0\}$ and 
$\Sigma_*=\Sigma\setminus (0,0)$.
The set $\Sigma_*$ is a smooth  manifold of dimension $2d-1$.
Let $\xi\in\R^{2d-1}$ be a local coordinate on $\Sigma_*$ with the coordinate mapping
$\xi\mapsto  z_\xi=(x_\xi,y_\xi)\in \Sigma_*$. Abusing notation we
write $|\xi|= |(x_\xi,y_\xi)|$. The vector
$
N(\xi)= (y_\xi, x_\xi)
$
is a normal to $\Sigma_*$ at $\xi$ of  length $|\xi|$, and
$$
\lan N(\xi), (x_\xi, y_\xi)\ran = 2 x_\xi \cdot y_\xi=0\,.
$$
For any $0<R_1<R_2$ we denote
\be\label{nota}
\begin{split}
	S^{R_1} = \{z\in\R^{2d}: |z|=R_1\}\,,\quad
	&\Sigma^{R_1} = \Sigma\cap S^{R_1} \,,
	\\
	S^{R_2}_{R_1} = \{z: R_1< |z|<R_2\}\,,\quad
	&
	\Sigma^{R_2}_{R_1}  = \Sigma\cap S^{R_2}_{R_1}\,,
\end{split}
\ee
and for  $t>0$  denote by $D_t$ the dilation operator
$$
D_t:  \R^{2d}\to \R^{2d},\quad z \mapsto tz\,.
$$
It preserves $\Sigma_*$, and we write
$\
D_t\xi =t\xi$ for $ \xi\in\Sigma_*\,.
$

The following lemma, specifying the geometry of $\Sigma_*$ and its vicinity, is proved in \cite{K}: 
\begin{lemma}\label{l_p1}
	1)
	There exists $\theta_0^*>0$ 
	such that for any $0<\theta_0\le \theta_0^*$ a
	suitable \nbh\  $\Sv= \Sv( \theta_0)$ of $\Sigma_*$ in $\R^{2d}\setminus\{0\}$ may be uniquely parametrised
	as
	\be\label{par}
	\Sv = \{ %(x_\xi,y_\xi) +\theta N_\xi = (x_\xi,y_\xi) + \theta(y_\xi, x_\xi)=:
	\pi(\xi,\theta),\ 
	\xi\in\Sigma_*,\; |\theta| < \theta_0\}\,,
	\ee
	where
	\be\non
	\pi(\xi,\theta) =
	(x_\xi,y_\xi) +\theta N_\xi = (x_\xi,y_\xi) + \theta(y_\xi, x_\xi).
	\ee
	2) For any vector $\pi:= \pi(\xi, \theta)\in\Sv$  its length equals
	\be\non
	|\pi| =|\xi|\sqrt{1+\theta^{2}}.
	\ee
	The distance from $\pi$ to
	$\Sigma$ equals $|\xi| |\theta|$, and the shortest path from $\pi$ to $\Sigma$  is the segment
	$S:= \{\pi(\xi, t\theta): 0\le t\le 1\}$. \\
	3) If $z=(x,y)\in S^R$ is such that $\dist (z, \Sigma)\le \frac12 R\theta_0$, then $z=\pi(\xi,\theta)\in\Sv$,
	where $|\xi|\le R\le |\xi|\sqrt{1+\theta_0^2}$. \\
	4) If $(x,y)=\pi(\xi,\theta)\in\Sv$, then
	\be\label{p0}
	x\cdot y =|\xi|^2\theta.
	\ee
	5) 		If $(x,y)\in S^R\cap (\Sv)^c$, then $|x\cdot y| \ge c|\theta_0|^2 R^2$ for some $c>0$.
\end{lemma}

The
coordinates \eqref{par} are known as the normal coordinates, and their existence  follows easily
from the implicit function theorem. The assertion~1) is a bit more precise than the general result since it
specifies the width of the \nbh\ $\Sv$.

Let as fix any $0<\theta_0\le\theta_0^*$, and consider the manifold  $\Sv(\theta_0)$. Below we provide it with  additional
structures and during the corresponding constructions decrease $\theta_0^*$, if needed.
Consider the set $\Sigma^1$ (see \eqref{nota}). 
It equals
$$
\Sigma^1= \{(x,y): x\cdot y=0, |x|^2+|y|^2=1\}\,,
$$
and  is a
smooth compact  submanifold of $\R^{2d}$ of codimension 2. Let us cover it by some finite system of charts
$\cN_1,\dots, \cN_{\tilde n}$, $\cN_j =\{\eta^j=  (\eta_1^j,\dots,\eta^j_{2d-2})\}$.  Denote by $m(d\eta)$ the volume element
on $\Sigma^1$, induced from $\R^{2d}$, and denote the  coordinate maps as
$\cN_j\ni \eta^j\to(x_{\eta^j},y_{\eta^j})\in\Sigma^1$. We will write points of $\Sigma^1$ both as $\eta$ and $(x_{\eta},y_{\eta})$.

The mapping
$$
\Sod\times \R^+ \to \Sigma_*,\quad (\eta, t)\to D_t (x_{\eta},y_{\eta}) \in \Sigma^t, 
$$
is 1-1,  and is a local diffeomorphism; so this is a global diffeomorphism. Accordingly we can cover $\Sigma_*$ by the $\tilde n$
charts $\cN_j\times \R_+$, with the coordinate maps
$
(\eta^j, t) \mapsto D_t(x_{\eta^j},y_{\eta^j})$, $ \eta^j\in\cN_j,\;\; t>0\,.
$
In these coordinates the volume element on $\Sigma^t$ is $t^{2d-2}m(d\eta)$.  Since $\p/\p t \in T_{\eta,t}\Sigma_*$ is
a vector of unit length, perpendicular to $\Sigma^t$, then the volume element on $\Sigma_*$ is
\be\label{vol_on_*}
d_{\Sigma_*} = t^{2d-2}m(d\eta)\,dt\,.
\ee

The coordinates $(\eta,t,\theta)$ with $\eta\in\cN_j, t>0, |\theta|<\theta_0$, where $1\le j\le\tilde n$,  make  coordinate systems on
$\Sv$. Since the vectors $\p/\p t$ and $t^{-1} \p/\p \theta$ form an orthonormal base of the orthogonal complement in  $\R^{2d}$
to $T_{(\eta,t,\theta=0)} \Sigma^t$, then in the domain  $\Sv$ the volume element $dz=dx\,dy$ may be written as
\be\label{p4}
dz =t^{2d-1} \mu(\eta,t,\theta) m(d\eta)dt\,d\theta\,, \quad \mu(\eta,t,0)=1\,.
\ee
The transformation $D_r:(\eta,t,\theta)\mapsto (\eta, rt,\theta)$, $r>0$, multiplies the form in the l.h.s. by $r^{2d}$, preserves $d\eta$ and $d\theta$, and multiplies
$dt$ by $r$. Hence, $\mu$ does not depend on $t$, and we have got

\begin{lemma}\label{l_p2}
	The coordinates
	$
	(\eta^j, t,\theta)$, where $\eta^j  \in\cN_j,
	\; t>0, \; |\theta|<\theta_0\,,
	$
	and   $1\le j \le\tilde n$,   define on $\Sv$ coordinate systems, jointly covering $\Sv$.
	In these coordinates 
	the dilations $D_r$, $r>0$,  reed as
	$
	D_r: (\eta, t,\theta) \mapsto (\eta, rt, \theta)\,,
	$
	and the volume element has the form \eqref{p4}, where $\mu$ does not depend on $t$.
\end{lemma}

Consider the mapping 
$$
\Pi: \Sv \to \Sigma_*,\qquad  z=\pi(\xi,\theta) \mapsto \xi.
$$
By Lemma \ref{l_p1}.2,  $|z|\le |\Pi(z)| \le 2 |z|$.% and  for $z=(x,y) = \pi(\xi, \theta)$,
%$\p/ \p\theta =  N_\xi \cdot\partial_z. %\quad \frac{\p}{\p t} = \Pi(z) \cdot\p_z.$	

For $0<R_1<R_2$ we will  denote
\be\lbl{Sigma_1^2}
\big(\Sv\big)^{R_2}_{R_1} = \Sv\cap\big(\Pi^{-1} \Sigma^{R_2}_{R_1}\big)\,.
\ee
In  coordinates $(\eta^j, t,\theta)$  this domain is $\{(\eta^j, t,\theta): R_1<t< R_2, |\theta|<\theta_0\}$.

Let us consider functions $\Gamma$ and $F$ in the variables $(\eta, t,\theta)$. Consider first 
$\Gamma_s(z)$ with $z = \pi(\xi,\theta) \in \Sv$. Since $\pi(\xi, \theta) =z_\xi+ \theta N_\xi$, then
$
\frac{\p^k \Gamma_s(z)}{ \p\theta^k} = d_z^k \Gamma_s(z)(N(\xi), \dots, N(\xi)).
$
As  $|N(\xi)| = |\xi| \le|z|$, then by \eqref{Ga_condition}, 
\be\label{p9}
| \Gamma_s| \ge K^{-1},\quad
\Big| \frac{\p^k}{\p \theta^k} \Gamma_s\Big| \le C_1 K\lan(z,s)\ran^{r_1-k}  |N_\xi|^k \le   C_2 K\lan(z,s)\ran^{r_1}, 
\ee
for  	$k\le3$.
Similar, since $F$ satisfies \eqref{B_condition}, then for $z = \pi(\eta, t, \theta)$ we have 
\be\label{p10}
F_s(\eta,t,\theta) \in C^2
\;\; \text{ and }\;\; \big|\frac{\p^k}{\p \theta^k}  F_s \big|\le C^\# (t,s)\,, \quad k\le2.
\ee

\subsection{Integral over the complement to a \nbh\ of $\Sigma$. }\label{s_over_complement}
On $\R_+ \times \R^d$ let  us define the function 
$$
\Theta=\Theta(t,s) = \ \lan(t,s)\ran^{-r_1} \le1
$$
($r_1\ge0$ is the exponent in \eqref{Ga_condition}), 
and consider a \nbh\ of $\Sigma$:
$$
\Sigma^{nbh}(\theta_0) = \{\pi(\xi, \theta): |\theta| \le \theta_{0m}\}\subset\Sv(\theta_0),
 \quad \theta_{0m} = \theta_0\Theta(t,s). 
$$
Consider the integral over its complement,
$
\Upsilon_s^m(\theta_0) = \lan |I_s|, \R^{2d}\setminus \Sigma^{nbh}(\theta_0)\ran. 
$
Using the polar coordinates in $\R^{2d}$, we  have 
\begin{equation*}
\begin{split}
%\lan I_s, \R^{2d}\setminus \Sv\ran
\Upsilon_s^m(\theta_0) %:= \lan| I_s|,& { \R^{2d}\setminus \Sv(\theta_0)}\ran\Big| \le 
&\le \lan |I_s|, { \{|z|\le \nu^b\} } \ran   \\
+&\nu^2 C_d \int_{\nu^b}^\infty dr\,r^{2d-1}
\int_{S^{1}\setminus \Sigma^{nbh}(\theta_0)} \frac{|F_s(z)|\, d_{S^{1}}}{ (x\cdot y)^2 +(\nu\Gamma_s(z))^2}\,,
\end{split}
\end{equation*}
where we choose $b=\tfrac1{d}$  and denote
 $d_{S^{1}}$ is the normalised Lebesgue measure on  $S^{1}$. 
By item 5) of Lemma~\ref{l_p1} with $\theta_0$ replaced by $\theta_0 \Theta(t,s)$, 
the divisor of the integrand  is 
$\ge C^{-2} r^4\theta_0^4\Theta^4.
$
So the internal integral is bounded by 
$
C^\#(r,s) r^{-4}  \lan(t,s)\ran^{r_1} 
$.
Due to this, 
\eqref{p10} and  \eqref{nearsing} 
the r.h.s.    is bounded by
$$
C_1^\#(s)  \nu^{2bd} +
\nu^2 C^\#(s) \int_{\nu^b}^\infty C^\#(r) r^{2d-5}\,dr
\le C_1^\#(s) \nu^2 \chi_d(\nu)
$$
(we recall the notation \eqref{chi_d}). Accordingly, 
\be\label{int_ext}
\Upsilon_s^m(\theta_0)   \le  C^\#(s)   \nu^2 \chi_d(\nu)\,.
\ee

\subsection{Integral over the vicinity of $\Sigma$.}\label{s_over_domain}
Let   $\Svm$ be a \nbh\ of $\Sigma_*$ such that 
\be\label{Smod}
\Sigma^{nbh} (\tfrac12 \theta_0) \subset \Svm \subset \Sigma^{nbh}(2\theta_0)\,.
\ee
Then in view of \eqref{int_ext} 
\begin{equation*}
\begin{split}
| \lan I_s, \Sigma^{nbh}(2\theta_0)\ran  -   \lan I_s, \Svm \ran | \le \lan |I_s|, 
\R^{2d}\setminus \Sigma^{nbh}(\tfrac12\theta_0)\ran \\
=\Upsilon_s^m\big(\tfrac{\theta_0}2\big)    \le C_1^\#(s) \nu^2 \chi_d(\nu)\,.
\end{split} 
\end{equation*}
In view of this estimate and \eqref{int_ext}, 
\be\label{zamena_oblasti}
\left|  \lan I_s, \R^{2d}\ran-
\lan I_s, \Svm\ran \right|\le C_1^\#(s) \nu^2 \chi_d(\nu)\,.
\ee
So	 to prove the theorem 
it suffices to calculate  integrals $ \lan I_s, \Svm\ran$ over any domains $\Svm$ as in \eqref{Smod}. 
Below we  will do this for domains $\Svm$ of the form
$$
\Svm =\{(\eta, t, \theta) : -\theta^-( t,s)  <\theta <  \theta^+( t,s)\},
$$
with suitably defined  functions $\theta_0^\pm( t,s)\in [\tfrac12
 \theta_{0m} , 
2
 \theta_{0m}  ]$. 

We estimate the integrals 
$ \lan I_s, \Svm(\theta_0) \ran$ in 4 steps.

\subsection{Step 1: Disintegration of $I_s$}\label{s_step1}
For any $0<R_1<R_2$ we define domains 
$
(\Sigma^{nbh}(\theta_0))^{R_2}_{R_1} = \Sigma^{nbh}(\theta_0) \cap (\Pi^{-1}\Sigma^{R_2}_{R_1} )
$
and using \eqref{p4}
write integral $ \lan I_s, {\big(\Sigma^{nbh}(\theta_0)\big)^{R_2}_{R_1}}\ran$ as
$$
\nu^2   \int_{\Sod} m(d\eta)  \int_{R_1}^{R_2}  dt \, t^{2d-1} \int_{-\theta_{0m}}^{\theta_{0m}} d\theta \, 
\frac{F_s(\eta, t,\theta) \mu(\eta,\theta)}{(x\cdot y)^2 +(\nu \Gamma_s(\eta,t,\theta))^2}
$$
$$
=  \nu^2   \int_{\Sod} m(d\eta)  \int_{R_1}^{R_2}  J_s(\eta,t)  t^{2d-1} dt\,
$$
(cf. \eqref{p1}),  	where by \eqref{p0}
$$
J_s(\eta,t) = \int_{-\theta_{0m}}^{\theta_{0m}} d\theta\,
\frac{F_s(\eta, t,\theta) \mu(\eta,\theta)}{t^4\theta^2 +(\nu \Gamma_s(\eta,t,\theta))^2}\,.
$$

To study $J_s$ let us  write $\Gamma_s$ as
$$
\Gamma_s(\eta,t,\theta) = h_{\eta, t, s}(\theta) \Gamma_s(\eta, t,0)\,.
$$
The function  $ h_{\eta, t, s}(\theta) =: h(\theta)$ is  $C^3$--smooth and in   view of \eqref{p9} satisfies
\be\label{p11}
h(0)=1, 
\quad
\Big| \frac{\p^k}{\p \theta^k} h(\theta) \Big| \le C \Theta^{-1}
\quad  %\forall\, s,\eta, t, \theta, \; 
\forall\,1\le k\le3,
\ee
for all $ \eta, t, \theta,s$. 
Denoting $\eps= \nu t^{-2} \Gamma_s(\eta, t,0)$, we write $J_s$ as
\be\label{Js}
J_s = t^{-4} \int_{-\theta_{0m}}^{\theta_{0m}}
\frac{F_s(\eta,t,\theta)\mu(\eta,\theta) h^{-2}(\theta)\,d\theta}{\theta^2 h^{-2}(\theta) +\eps^2}\,.
\ee

\subsection{Step 2:  Definition of domains $\Svm$. }\label{s_step2}
On the segment 
$
\I= [-\theta_{0m},  \theta_{0m} ]  \subset [-\theta_0,   \theta_0 ]
$
consider the function $h$ and the function 
$$
f=f_{\eta, t, s}:
\I \ni \theta \mapsto \thb = \theta  / h(\theta) \,,
$$
By \eqref{p11}, $ \tfrac23 \le h\le\tfrac32$ on $\I$, if $\theta_0$ is small enough. From here and \eqref{p11}, for $\theta\in \I$
we have   $ \tfrac12 \le f'(\theta) \le 2$ (if $\theta_0$ is small), and 
\be\label{f_est}
\left| \frac{\p^k f}{\p \theta^k}\right|  \le C \Theta^{-(k-1)}(t,s) \quad \forall\, 1\le  k\le3. 
\ee
So $f$ defines a $C^3$--diffeomorphism of $\I$ on $f(\I)$, 
$\tfrac12 \I\subset f(\I) \subset 2\I$, such that $f'(0)=1$ and $f^{-1}$ also 
satisfies estimates \eqref{f_est} (with a modified constant $C$). 

Let us set $\zeta^+ =f(\theta_0\Theta)$ and 
$\zeta^- =-f(-\theta_0\Theta)$. Then
$
2^{-1} \theta_0\Theta \le\zeta^\pm\le 2 \theta_0\Theta,
$
and  passing in  integral \eqref{Js} from variable $\theta$ to $\zeta=f(\theta)$ we find that 
$$
J_s = t^{-4} \int_{-\zeta^-}^{\zeta^+}
\frac{F_s(\eta,t,\theta)\mu(\eta,\theta) h^{-2}(\theta) (f^{-1}(\thb))' \,d\thb}{\thb^2 +\eps^2}\,.
$$
Denoting the nominator of the integrand as $\Phi_s(\eta, t,\thb)$ and using \eqref{p10} 
we see
that this is a $C^2$--smooth function, satisfying
\be\non
| \frac{\p^k}{\p\theta^k}
\Phi_s| \le C^\#(t,s)\quad \forall\, 0\le k\le2\,.
\ee
Moreover, since $h(0)=1$ and $(f^{-1}(0))'= f'(0)=1$, then in view of \eqref{p4} we have that
\be\label{p13}
\Phi_s(\eta, t, 0) = F_s(\eta, t, 0)\,.
\ee

Now 	denote  
\be\lbl{zeta 0}
\thb_0 = \min(\zeta^+, \zeta^-)\in[ \tfrac12 \theta_0  \Theta, 2\theta_0 \Theta],
\ee
set
$$
\theta_0^-(\eta, t, s) =- f^{-1}_{\eta, t,s}(-\zeta_0), \quad \theta_0^+(\eta, t, s) = f^{-1}_{\eta, t,s}(\zeta_0),
$$
and   use these functions $\theta_0^\pm$ to define the domains $\Svm$. Then 
\be\label{p001}
\lan I_s, {\big(\Svm\big)^{R_2}_{R_1}}\ran =  \nu^2   \int_{\Sod} m(d\eta)  \int_{R_1}^{R_2}  J_s^m(\eta,t)  t^{2d-1} dt\,, 
\ee
where 
\begin{equation*}
\begin{split}
J_s^m(\eta,t)  &=   \int_{-\theta_0^-}^{\theta_0^+} d\theta\,
\frac{F_s(\eta, t,\theta) \mu(\eta,\theta)}{t^4\theta^2 +(\nu \Gamma_s(\eta,t,\theta))^2}\\
&= t^{-4} \int_{-\zeta_0}^{\zeta_0}
\frac{F_s(\eta,t,\theta)\mu(\eta,\theta) h^{-2}(\theta) (f^{-1}(\thb))' \,d\thb}{\thb^2 +\eps^2}\,.	
\end{split}
\end{equation*}

\subsection{Step 3: Calculating integrals   $\lan I_s, \Svm\ran$}\label{s_step3}

Clearly, 
\be\non
|J_s^{m}| \le
  \frac{C^\#(t,s)}{t^{4}} \int_{-\thb_0}^{\thb_0} \frac{d\thb}{{\thb}^2+\eps^2}\le 
  \frac{
C_1^\#(t,s)\theta_0\Theta}{\nu^{2} \Gamma_s^{2}   }(\eta, t, 0)\le \nu^{-2} C^\#(t,s)
\ee
(here and below the constants may depend on $\theta_0$). So if $\alpha>0$, then 
\be\label{first_est}
\lan| I_s|, {\big(\Svm\big)_{\nu^{-\alpha}}^{\infty}}\ran \le
\nu^{-2} \nu^2 C^\#(s) \int^\infty_{\nu^{-\alpha}} t^{2d-1} C^\#(t) dt
\le C_\alpha^\#(s) \nu^2. 
\ee
Similar, if $|s| \ge \nu^{-\beta}$, $\beta>0$, then
\be\label{second_est}
\lan| I_s|, {\big(\Svm\big)_{0}^{\infty}}\ran \le
C^\#(s) \int^\infty_{0} t^{2d-1} C^\#(t) dt
\le C_\beta^\#(s) \nu^2. 
\ee
\medskip

To estimate $J^m_s(\eta,t)$ for $t$ and $s$ not very large, let us  consider  the integral $J_s^{0m}$,
obtained from $J_s^m$ by frosening $\Phi_s$ at $\thb=0$:
$$
J_s^{0m}= t^{-4}  \int _{-\thb_0}^{\thb_0}\frac{ \Phi_s(\eta,t,0)\,d\thb}{\thb^2+\eps^2}
=2 t^{-4} F_s(\eta,t,0) \eps^{-1} \tan^{-1}\frac{\thb_0}{\eps}
$$
(we use \eqref{p13}).
As $0<\frac{\pi}2 -\tan^{-1}\frac1{\bar\eps}<\bar\eps$ for $0<\bar\eps\le\frac12$, then 
\be\label{p17}
0<\pi \nu^{-1} t^{-2} (F_s/\Gamma_s)\mid_{\theta=0}
% \frac2{\ga_s}\bga 
- J^{0m}_s <\frac2{\thb_0} t^{-4} F_s(\eta,t,0)\,,
\ee
if $\nu t^{-2} \Gamma_s(\eta,t,0)\le \tfrac12 \thb_0$, which holds if 
\be\label{p18}
\nu \le C  t^2 \Theta^2(t,s),
\ee
in view of \eqref{p9} and \eqref{zeta 0}. 
Now we estimate the difference between $J_s^m$ and $J_s^{0m}$. We have:
$$
(J^m_s-J_s^{0m})(\eta, t) = t^{-4} \int_{-\thb_0}^{\thb_0} \frac{\Phi_s(\eta,t,\thb)-\Phi_s(\eta,t,0)}{\thb^2+\eps^2}d\thb\,.
$$
Since $\Phi_s$ is a smooth function and its $C^2$--norm is bounded by $C^\#(t,s)$, then
$$
\Phi_s(\eta,t,\thb)-\Phi_s(\eta,t,0) = A_s(\eta,t)\thb +B_s(\eta,t,\thb)\thb^2\,,
$$
where $|A_s|, |B_s| \le C^\#(t,s)$. From here
$$
|J^m_s - J_s^{0m}|\le C_1^\# (t,s) t^{-4} \int_0^{\thb_0} \frac{\thb^2\,d\thb}{\thb^2+\eps^2}\le C_2^\#(t,s)t^{-4}.
$$
Denote
\be\label{p20}
\cJ_s(\eta,t)= \pi t^{-2} (F_s/\Gamma_s)(\eta, t, 0).
\ee
Then, jointly with \eqref{p17}, the last estimate tell us that
\be\label{p19}
|J_s^m - \nu^{-1}\cJ_s(\eta,t)|\le C^\#(t,s)t^{-4}
\quad\text{if \eqref{p18} holds}. 
\ee

\subsection{Step 4: End of the proof}\label{s_step4}
Let us  write $\lan I_s, \Svm\ran$ as 
$$
\lan I_s, \big(\Svm\big)^{\nu^{b}}_0 \ran +
\lan I_s, {\big(\Svm\big)^{\nu^{-a}}_{\nu^{b} }} \ran+
\lan I_s, {\big(\Svm\big)^{\infty }_{\nu^{-a}}} \ran.
$$
We will  analyse the three terms,  choosing properly positive constants $a,b$. 

\noindent {\bf
	1.} By \eqref{first_est}, 
$$
\Big| \lan I_s, (\Svm)^{\infty}_{\nu^{-a}} \ran \Big|  % \nu \int_{\nu^{-b}}^\infty t^{2d-3} C^\#(t,s)\,dt
\le  C_a^\#(s) \nu^2.
$$
Similar,
$$
\Big | \nu \int_{\Sod}m(d\eta) \int_{\nu^{-a}}^{\infty} dt\, t^{2d-1} \cJ_s(\eta,t) \Big|\le
C^\#(s) \nu  \int_{\nu^{-a}}^\infty t^{2d-3} C^\#(t)\,dt \le  C_a^\#(s) \nu^2. 
$$
\medskip

\noindent {\bf
	2.}  Since $(\Svm)^{\de}_0\subset K_{2\de}$, estimate \eqref{nearsing} with $\delta=\nu^b$ implies
$$
\lan | I_s |,  (\Svm)^{\nu^b}_0\ran
\le \lan |  I_s| ,  K_{2\nu^b}  \ran 
\le C^\#(s) \nu^{2bd}. 
$$ 
Besides,
$$
\nu \int_{\Sod}m(d\eta) \int_{0}^{\nu^{b}} dt\, t^{2d-1} \cJ_s(\eta,t)\le \nu C^\#(s) \int_0^{\nu^b} t^{2d-3} dt = 
C_1^\#(s)\nu^{1+2b(d-1)}\,.
$$
\medskip

\noindent {\bf
	3.} Now consider
\be\label{X_s}
X_s := 
\Big|  \lan I_s, (\Svm)^{\nu^{-a}}_{\nu^{b}}\ran - \nu \int_{\Sod}m(d\eta) 
\int_{\nu^{b}}^{\nu^{-a}} dt\, t^{2d-1} \cJ_s(\eta,t)\Big|\,.
\ee
We claim that 
\be\label{we_claim}
X_s \le \nu^2 \chi_d(\nu) 
C^\#(s),
\ee
where $\chi_d(\nu)$ was defined in \eqref{chi_d}.
Indeed, if $|s|\ge \nu^{-{a}}$, then by \eqref{second_est}
the modulus of the first term in the r.h.s. of \eqref{X_s} is
$\le C_{a}^\#(s) \nu^2$. The second term also satisfies this estimate with some other $ C_{a}^\#(s)$. So
\eqref{we_claim} is established if $|s|\ge \nu^{-{a}}$. By a similar (and even easier) argument the claimed estimate 
holds if $\nu_1 \le\nu\le1$ for any fixed constant $\nu_1>0$. 

 Now let us consider the case
$
\nu\le\nu_1$, $ |s| \le \nu^{-{a}}. 
$
%Note that $X_s\le X_s^1 +X_s^2$, where $X_s^1$ is obtained from $X_s$ by replacing $\nu^{-a}$ by 1, and 
%$X_s^2$ -- by replacing $\nu^b$ by 1. Let $\nu^b\le t\le1$ and 
%$
%\nu\le\nu_1$, $ |s| \le \nu^{-{a}}. 
%$
Then for $\nu^b\le t\le \nu^{-a}$ the r.h.s. of \eqref{p18} is no smaller than 
$\ 
Y:= C \nu^{2b} (1 +2\nu^{-2{a}})^{-r_1}. 
$
If 
\be\label{b_cond}
2b +2{a} r_1 <1, 
\ee
then $Y\ge\nu$ if $\nu\le\nu_1$ and  $\nu_1>0$ is small enough. 
Then assumption \eqref{p18} holds, and \eqref{p001} together with \eqref{p19} implies \eqref{we_claim}:
\begin{equation*}
\begin{split}
X_s\le&\nu^2 
\int_{\Sod}m(d\eta) \int_{\nu^{b}}^{\nu^{-a}} dt\, t^{2d-1} C^\#(t,s)    t^{-4}\\
&\le \nu^2  C^\#(s)  \int_{\nu^{b}}^{\nu^{-a}} dt\, C^\#(t) t^{2d-5} 
\le C_1^\# (s) \nu^2  \chi_d(\nu)\,.
\end{split}
\end{equation*}
%So $X^1_s$ satisfies \eqref{we_claim}. 

%If $1\le t\le \nu^{-a}$, then the r.h.s. of   \eqref{p18}  is 
%$
%\ge C(2\nu^{-2a})^{-r_1}. 
%$
%So if \eqref{b_cond} holds, then assumption  \eqref{p18}  is fulfilled and $X_s^2$ also meets \eqref{we_claim}. 

%5\medskip

\noindent {\bf
	4.}
In the same time, in view of \eqref{p20}, for any $-\infty\le A<B\le\infty$ we have
\begin{equation}\label{p21}
\begin{split}
\Big| \nu \int_{\Sod}m(d\eta) \int_{\nu^{B}}^{\nu^{A}} dt\, t^{2d-1}\cJ_s(\eta,t)\Big|
\le  \nu \int_{\nu^{B}}^{\nu^{A}} dt\, t^{2d-1-2} C^\#(t,s)\le C_1^\#(s)\nu,
\end{split}
\end{equation}
since $d\ge2$.
\medskip

Now we get from {\bf 1-4} that
\begin{equation}\label{p30}
\begin{split}
&\Big|  \lan I_s, {\Svm} \ran - \nu \int_{\Sod}m(d\eta) \int_{0}^{\infty} dt\, t^{2d-1} \cJ_s(\eta,t)\Big|\\
&\le \Big| \lan I_s,  (\Svm)_{\nu^{-a}}^\infty\ran \Big|
+  \nu \int_{\Sod}m(d\eta) \int_{\nu^{-a}}^{\infty} dt\, t^{2d-1} |\cJ_s(\eta,t) |\\
&+\Big|   \lan I_s,  (\Svm)_{\nu^{b}}^{\nu^{-a}} \ran-  \nu \int_{\Sod}m(d\eta) 
\int_{\nu^{b}}^{\nu^{-a}} dt\, t^{2d-1} \cJ_s(\eta,t) \Big|\\
&+  \Big| \lan I_s, (\Svm)_{0}^{\nu^b}\ran\Big| +  \nu \int_{\Sod}m(d\eta) \int_{0}^{\nu^b} dt\, t^{2d-1} |\cJ_s(\eta,t)|\\
& \le\,  C^\#(s)\big( \nu^2 +\nu^2 \chi_d(\nu) + \nu^{2bd}+ \nu^{1+2b(d-1)} \big)=:Z,
\end{split}
\end{equation}
if condition \eqref{b_cond} holds for some $a,b >0$. If $d\ge3$, we choose  $b=\tfrac1{d} <1$. Then 
\eqref{b_cond} holds for some ${a}>0$,  and $Z\le  C^\#(s) \nu^2\chi_d(\nu)$. 
 If $d=2$, then $\nu^{2bd}  = \nu^{4b}$ and $\nu^{1+2b(d-1)}=\nu^{1+2b}$. We choose $b=1/2-\aleph/4$, $\aleph>0$. 
Then again \eqref{b_cond} holds 
for some ${a}(\aleph)>0$, so $Z\le C_\aleph^\#(s) \nu^{2-\aleph} $.

% If $d=2$, then the factor $\ln \nu^{-1}$ should be added to the r.h.s.
By \eqref{p21} the improper integral $\int_{\Sod}m(d\eta) \int_{0}^{\infty} dt\, t^{2d-1} \cJ_s$
converges absolutely 
and is bounded by $C_1^\#(s)$. In view of \eqref{vol_on_*} and \eqref{p20}
it  may be written as
$$
\nu\int_{\Sod}m(d\eta) \int_{0}^{\infty} dt\, t^{2d-1} \cJ_s(\eta,t)=
\pi\nu  \int_{\Sigma_*} \frac{F_s(z)}{\Gamma_s(z)   \sqrt{ |x|^2+|y|^2}} \, dz\!\mid_{\Sigma_*}\,.
$$

This result jointly with the estimates \eqref{p30} and \eqref{zamena_oblasti} imply the assertion of Theorem \ref{t_singint}. 	
\qed

\subsection {Proof of  Proposition \ref{p_prop}}\label{s_proofprop}

Denoting by $B_r^R$
the spherical layer  $\{ r\le  |z|\le R\}$, in view of \eqref{vol_on_*} for $R\ge0$ we have
$$
\int_{B_R^{R+1}} \mu^\Sigma (dz) \le C\int_R^{R+1} t^{2d-3} dt \le C_1 (R+1)^{2d-3}. 
$$
So   for any function $f\in \cC_m(\R^{2d})$ with $m>2d-2$ 
the integral $ \int  f(z)\,  \mu^\Sigma(dz)$ is bounded by 
\begin{equation*}
\begin{split}
% \int  f(z)\,  \mu^\Sigma(dz) \le 
|f|_m  \sum_{R=0}^\infty \int_{B^{R+1}_R} \lan z\ran^{-m} \,  \mu^\Sigma(dz)
%& \le C_1 |f|_m  \sum_{R=0}^\infty \lan R\ran^{-m}   \frac{(R+1)^{2d-2} -R^{2d-2}}{2d-2}\\
% & \le 
\le C_2 |f|_m  \sum_{R=0}^\infty \frac{(R+1)^{2d-3} }{\lan R\ran^{m}}
= C_3 |f|_m.
\end{split}
\end{equation*}
This proves the proposition.

\section{Oscillating sums under the limit \eqref{assumption}} \label{s_9}

	\subsection{Correlations between increments of $a_s^{(1)}$ 	}\label{s_oscill}
	
	 Correlations between the increments of $a_s^{(1)}$ and some similar quantities, treated in Section~\ref{sec:Q_s-prop} lead to sums of the following form: 
\be\label{Sis_0}
	\begin{split}
\Sigma_s^0 (\tau) =\ssum_{1,2}{F}_s(s_1, s_2) \dess \int_0^\tau  \int_0^\tau d\theta_1\,d\theta_2 e^{-\ga_s\big( 2\tau-\theta_1 -\theta_2)\big) +i \nu^{-1} \oms(\theta_1-\theta_2)},
\end{split}
	\ee
	where $s\in\R^d$ and
$
\tau  \in( 0,1]
$
is a parameter. Our goal is to study their asymptotical behaviour as $\nu\to0$, $L\to\infty$,
uniformly in $\tau$. 
In \eqref{Sis_0} 	 ${F}_s$ is a  real  $C^2$--function on $\R^{2d}$, satisfying \eqref{B_condition} 
(for example, ${F}_s$ may be the function, defined in \eqref{F}). Integrating over $d\theta_1\, d\theta_2$ we find that 
\be\label{Sis_1}
\Sigma_s^0 (\tau) =\ssum_{1,2} {F}_s(s_1, s_2) \dess\,  \frac{|e^{i\nu^{-1} \oms\tau} -e^{-\ga_s\tau}|^2 }{\ga_s^2+(\nu^{-1} \oms)^2}.
\ee
Note that the nominator in the fraction above equals
\be\label{chisl}
|e^{i\nu^{-1} \oms\tau} -e^{-\ga_s\tau}|^2 = 1+e^{-2\ga_s \tau} -2 e^{-\ga_s \tau} \cos(\nu^{-1} \tau \oms).
\ee
We wish to study the asymptotical behaviour of the sum  $\Sigma^0_s$ as $\nu\to0,\, L\to\infty$ and, as before, 
 will do this by comparing \eqref{Sis_0} 
 with  the integral 
$$
I_s^0= \int\int  ds_1\, ds_2 \int_0^\tau  \int_0^\tau d\theta_1\,d\theta_2 \, \des
M_s(s_1, s_2,s_3; \theta_1, \theta_2),
$$
where 
$$
M_s=  {F}_s(s_1, s_2)   e^{-\ga_s\big( 2\tau-\theta_1-\theta_2\big) +i \nu^{-1} \oms(\theta_1-\theta_2)}.
$$
By Theorem \ref{t_sum_integral}, 
\be\label{Iso_diffe}
|I_s^0 -\Sigma_s^0| \le C^\#(s)\nu^{-2}L^{-2} .
\ee
So if $L\gg \nu^{-1}$, then 
 to calculate the asymptotical behaviour of $\Sigma_s^0$ it suffices to calculate that of $I_s^0$. Integrating over $d\theta_1\,d\theta_2$
 in the expression for $I_s^0$ we obtain \eqref{Sis_1} with $\ssum_j$ replaced by $\int ds_j$. In view of \eqref{chisl} we get that
 $I_s^0= \nu^2 I_s^{0,1} + \nu^2 I_s^{0,2} $, where  
 \be\non
	\begin{split}
	& I_s^{0,1}  =  (1+e^{-2\ga_s\tau}) \int\int ds_1\,ds_2\, \frac{ \des {F}_s(s_1,s_2)}{\nu^2 \ga_s^2 +(\oms)^2},\\
  &I_s^{0,2}  = -2  e^{-\ga_s\tau} \int\int ds_1\,ds_2\, \frac{ \des {F}_s(s_1,s_2) \cos(\nu^{-1} \tau \oms) }{\nu^2 \ga_s^2 +(\oms)^2}.
\end{split}
	\ee 
Consider first $I_s^{0,1} $.  Applying 
Theorem~\ref{t_singint} with $\Ga_s=\ga_s/2$, we find
\be\label{Is_0_as}
I_s^{0,1} = \frac{1+e^{-2\ga_s\tau}}{2\ga_s} \pi \nu^{-1}
 \int_{{}^s\!\Sigma_*} \frac{{F}_s(z)}{|(s-s_1, s-s_2)| }\,dz\!\mid_{{}^s\Sigma_*} +  O(1)\nu^{-\aleph_d} C^\#(s;\aleph_d),
\ee
where $z=(s_1,s_2)$.

It remains to study the asymptotical behaviour of $I_s^{0,2} $. It is described by the following result,
proved in Section~\ref{l_s_oscill} (also see \cite{K1}): 

\begin{lemma}\label{l_osc} 
For any $s\in\R^d$ and $\tau\in(0,1]$,
\be\label{pr88}
\Big| I_s^{0,2}  + \ga_s^{-1}  \nu^{-1} e^{-2\ga_s\tau} \pi 
 \int_{{}^s\!\Sigma_*} \frac{{F}_s(z)}{|(s-s_1, s-s_2)|  }\,dz\!\mid_{{}^s\!\Sigma_*} \Big| \le  C^\#(s) \chi_d(\nu).
\ee 
\end{lemma}

Relations  \eqref{Is_0_as} and  \eqref{pr88} imply the main result of this section:

\begin{theorem}\label{t_osc} For any $s\in\R^d$ and $\tau\in(0,1]$,
	\be\non
	\begin{split}
	\left| I_s^0 -  \frac{1-e^{-2\ga_s\tau}}{2\ga_s}  \nu \pi 
 \int_{{}^s\!\Sigma_*} \frac{{F}_s(z)}{|(s-s_1, s-s_2)|  }\,dz\!\mid_{{}^s\Sigma_*}
\right|
 \le C^\#(s;\aleph_d) \nu^{2-\aleph_d}.
\end{split}
	\ee
\end{theorem}

Jointly with \eqref{Iso_diffe} this gives an asymptotic for the sum $\Sigma_s^0$.

	\subsection{ Correlations between solutions and their increments
	}\label{s_correlation}
	
	In this section we analyse sums similar to \eqref{Sis_0} in which the integral $\int_0^\tau d\theta_2$ is replaced by the integral $\int_{-\infty}^0 d\theta_2$. 
	Sums of such form arise in Section~\ref{sec:Q_s-prop}, when studying the correlation $$\EE\Del a_s^{(1)}(\tau) \bar c_s^{(1)}(\tau)=e^{-\ga_s\tau}\EE\big(a_s^{(1)}(\tau)-e^{-\ga_s\tau} a_s^{(1)}(0)\big)a_s^{(1)}(0)$$ and similar quantities. 
	We will show that the considered sums are negligible. The reason for this is that for $\oms\ne 0$ 
	the "fast" frequency $\nu^{-1}\oms(\theta_1-\theta_2)$ is of the size $\lesssim1$ only if $|\theta_1|, |\theta_2|\lesssim \nu$, so the Lebesgue measure of such resonant vectors $(\theta_1,\theta_2)$ is only of order $\nu^2$ (if $\tau\sim1$). That is why such sums are much smaller that those of the form \eqref{Sis_0}, where the measure of the resonant vectors $(\theta_1,\theta_2)$  is of order $\nu\gg\nu^2$.
		
We consider the sum
\begin{equation*}
S_s(\tau)=
\ssum_{1,2}\dep  \int_0^\tau dl \int_{-\infty}^{0} dl' \, e^{\ga l'}  e^{i\nu^{-1}(l-l') \oms}   
   F_s(s_1, s_2, l,l',\tau), \qu s\in\R^d,
\end{equation*}
where $\ga>0$ and  $F_s$ is a  real function,  measurable in  $s_1,s_2,l,l'$, $C^2$--smooth in $s_1, s_2$, and such that 
\be\label{F_estim}
  e^{\ga l'} | \p^{\al}_{s_1,s_2} F_s  | \le
 C^{\#} (s)  C^{\#}(s_1) C^{\#} (s_2) \qquad \forall\ 0\leq |\al| \le d+1, 
 \ee
 uniformly in $l'\le0, \ 0\le\tau\le1,\  0\le l\le\tau$.

\begin{theorem}\label{c_6.2}
Let $0\leq\tau\leq 1$.	Then under the assumption \eqref{F_estim}
 the sum  $S_s$ meets    the estimate 
$$
|S_s(\tau)| \le C_1^\#(s) \big(\nu^2 \chi_d(\nu) + \nu^{-2}L^{-2}\big),
$$
where $ C^\#(s)$ depends only on $d, \ga$ and the function $  C_1^{\#}(s) $ from \eqref{F_estim}.
\end{theorem}

Note that since we assume no smoothness in $l'$ for function $F_s$, then the theorem also applies to the sums
$$
S_s^T(\tau) = \ssum_{1,2}\dep  \int_0^\tau dl \int_{-T}^{0} dl' \, e^{\ga l'}  e^{i\nu^{-1}(l-l') \oms}   
   F_s
$$
since we may extend $F_s$ by zero for $l'<-T$ and regard $S_s^T$ as the sum $S_s$. 

\begin{proof}
It is convenient to change the variable $\theta= -l',$ so that the sum $S_s$ takes the form
\begin{equation*}
S_s(\tau)=
\ssum_{1,2}\dep \int^{\infty}_{0} d\theta  \int_0^\tau dl \, e^{-\ga\theta}  e^{i\nu^{-1}(l+\theta) \oms}   
F_s(s_1, s_2, \theta, l,\tau).
\end{equation*}
As usual, to estimate  $S_s$  we first approximate it by the integral  
$$
I_s(\tau) =  \int^\infty_{0} d\theta \int_0^\tau dl \, \int\int ds_1 ds_2\, \des
e^{-\ga\theta }  e^{i\nu^{-1}(l+\theta) \oms} F_s.
$$
Applying  Theorem \ref{t_sum_integral} we get that 
\be\label{Iso_di}
|I_s (\tau) - S_s(\tau) | \le C^\#(s)\nu^{-2}L^{-2} .
\ee

To estimate $I_s$ we write it as a sum of  integrals over the domains $\{ \nu^{-1}(l+\theta)\ge1\} = \{  l\ge \nu-\theta\}$ 
 and
$\{ \nu^{-1}(l+\theta)\le1\} = \{ 0\le l\le \nu-\theta \}$.
\smallskip

{\it Integral over  $\{  l\ge \nu - \theta \}$.}
Let us denote this integral  $I^1_s$, and for any fixed $\theta\ge0$  and 
$ l\ge \nu - \theta$ consider the internal integral over $ds_1\,ds_2$:   
$$
I^1_s(l,\theta) = e^{-\ga  \theta} \int\int ds_1 ds_2\,  \des  e^{i\nu^{-1}(l+\theta) \oms}  F_s.
$$
Since $\oms = 2(s_1-s)\cdot(s-s_2)$ is a non-degenerate quadratic form and $\nu^{-1}(l+\theta)\ge1$ on the zone of 
integrating, then the integral $I^1_s(l,\theta)$ has the form \eqref{I_la}  with $\nu: = \nu (l+\theta)^{-1}$,
$\varphi = e^{-\ga\theta} F_s$ and $n=2d$. 
 So by \eqref{stph_est1},  \eqref{sob} and \eqref{F_estim}, 
 $$
| I^1_s(l,\theta) | \le C^{\#} (s) \nu^d (l+\theta)^{-d}  e^{-\ga \theta }, 
$$
where $ C^{\#} (s)$ is as in the theorem. 
Accordingly, 
$$
| I^1_s  | \le  {\frak C}  \int_{0}^\infty d\theta e^{-\ga \theta}  \int_{0}^\tau dl\, (l+\thet)^{-d}\chi_{\{l\geq \nu-\theta\}}, \qquad {\frak C}= C^{\#} (s) \nu^d .
$$
We split the integrating zone
$
\{ \theta \ge 0, \,0\le l\le \tau, \,l\ge \nu-\theta\}
$
to the part where $\theta\ge\nu$ and its complement,
and accordingly  split the integral above as 
 $  I^{1,1}_s +   I^{1,2}_s$, where 
  \begin{equation*}
\begin{split}
& I^{1,1}_s = \frak C \int^{\infty}_{\nu} d\theta e^{-\ga \theta} 
 \int_{0}^\tau   dl\, (l+\thet)^{-d}, %\quad \frak C=C^{\#} (s) \nu^d,
 \\
 & I^{1,2}_s  = \frak C \int_{(\nu-\tau)\vee0}^{\nu} d\theta e^{-\ga \theta} 
 \int_{\nu-\theta}^\tau     dl\, (l+\thet)^{-d}.
\end{split}
 \end{equation*}
Consider first $  I^{1,1}_s$. Computing the internal integral and replacing $\tau$ by $\infty$ we find
$$
I_s^{1,1}\leq C\frak C \int_\nu^\infty \frac{e^{-\ga\theta}}{\theta^{d-1}}\,d\theta \leq \frac{C_1\frak C}{\nu^{d-2}}\chi_d(\nu)=C_1^\#(s)\nu^2 \chi_d(\nu).
$$ 
Now consider $I_s^{1,2}$. Replacing in the external integral $(\nu-\tau)\vee0$ by $0$ and in the internal one $\tau$ by $\infty$ we find
$$
I_s^{1,2}\leq C\frak C \int_0^\nu d\theta \,\nu^{-d+1}= C_1^{\#}(s)\nu^2. 
$$

{\it Integral   over   $\{0\le  l\le  \nu - \theta\}$}.
Denoting this integral as $I^2_s$, we find 
$$
|I_s^2|\leq \int_0^\nu d\theta \,\int_0^{(\nu-\theta)\wedge \tau}dl\, \int\int ds_1ds_2\, |F_s|\leq C^{\#}(s)\nu^2,
$$
since $\int\int ds_1ds_2\, |F_s|\leq C^\#(s)$ and the area of integration over $d\theta dl$ is bounded by  $\nu^2$.

We have seen that 
$
| I_s |\le  C^\#(s) \nu^2\chi_d(\nu).
$
	This  and \eqref{Iso_di} jointly imply the assertion of the theorem. 
\end{proof}

	\subsection{Proof of  Lemma \ref{l_osc}}\label{l_s_oscill}

	Let us denote $\nuu = \tfrac12\nu\ga_s$. 
		If $\nuu>1$, then $|s| \ge C\nu^{-1/2r_*}$, so
	taking into account assumption \eqref{B_condition} we see that the both summands in the l.h.s. of \eqref{pr88} are bounded by $C^\#(s)$ and the result follows. So we may assume that $\nuu\le1$.

	 Let us write $ I^{0,2}_s$ as
	\be\non
	I^{0,2}_s= -\frac12 \,% \ga_s^{-2} {\nuu}^2 
	e^{-\ga_s\tau} \int\int ds_1\,ds_2\, \des
	\frac{F_s(s_1,s_2) \cos({\nu}^{-1} \tau \oms) }{{\nuu}^2  +(\oms/2)^2}.
	\ee
	We will examine the integrals $\lan I^{0,2}_s, K_{2r}\ran, r\ll1$,  $\lan I^{0,2}_s, \Sv \ran$ and 
	$\lan I^{0,2}_s, \R^{2d}\setminus\Sv \ran$ (see Notation, \eqref{K_delta}  and Lemma~\ref{l_p1}), 
	 and will derive the lemma from this analysis. The constants below do not depend on $\tau, s$ and $\nu'$. 
	 \medskip
	 
	  Let us re-write $I^{0,2}_s$, using the variables $(x,y)=z$, see \eqref{p01}.
	Disregarding for a moment the pre-factor
	$-\frac12  e^{-\ga_s\tau} $ 
	 we  examine the integral
	$$
	J_s = \int_{\R^{2d}} dz\, \frac{F_s(z) \cos \la x\cdot y}{(x\cdot y)^2 +\nuu^2 }, \qquad  \la = \nuu^{-1}\tau\ga_s.
	$$
	
	\noindent {\it Step 1.}  Since $|\cos \la x\cdot y|
		\leq 1$, an upper bound for $\lan  J_s, \ K_{2r}\ran$, where $ r\ll1$, (see \eqref{K_delta})  follows from \eqref{nearsing}:
	\be\lbl{pr1}
	\lan |J_s |, K_{2r}\ \ran \le C^\#(s) {\nu'}^{-2} r^{2d}. 
	\ee

	\noindent {\it Step 2.} Integral over $\Sv$. We recall that $(\Sigma^v)_{R_1}^{R_2}$ is defined in \eqref{Sigma_1^2}.
	Passing to the variables $(\eta, t,\theta)$ (see Lemma~\ref{l_p2}) 
	and using \eqref{p0} we disintegrate $\lan J_s, (\Sv)_r^\infty\ran =: J_s^r$ as 
follows:
 \be\label{3.9}
 \begin{split}
 J_s^r=
 \int_{\Sod} m(d\eta) \int_{r}^{\infty} dt\,  t^{2d-1}   \int_{-\theta_0}^{\theta_0} d\theta\,
 \frac{F_s(\eta, t,\theta) \mu(\eta,\theta) \cos(\la x\cdot y) }{(t^2\theta)^2 +\nuu^2}\\
=  \int_{\Sod} m(d\eta) \int_{r}^{\infty} dt \, t^{2d-1} \Upsilon_s(\eta,t)\,,
 \end{split}
 \ee
 where 
 $$
 \Upsilon_s(\eta,t) %= \int_{-\theta_0}^{\theta_0} d\theta 
 %\frac{F_s(\eta, t,\theta) \mu(\eta,\theta)  \cos(\la t^2\theta)}{t^4\theta^2 +\nuu^2}=
 = t^{-4} \int_{-\theta_0}^{\theta_0} 
  \frac{F_s(\eta,t,\theta)\mu(\eta,\theta)   \cos(\la t^2\theta)  \,d\theta}{\theta^2  +\eps^2}\,, \quad \eps =   \nuu t^{-2}\,.
 $$

To estimate  $\Upsilon_s$, 
%treat the case of  $\eps$ smaller than $\thb_0/2$, 
consider first the integral $\Upsilon_s^{0}$,
 obtained from $\Upsilon_s$ by frozening $F_s \mu$ at $\theta=0$. Since $\mu(\eta, 0)=1$, then $\Upsilon_s^{0}$ equals 
$$
  2 t^{-4}    F_s (\eta,t,0) \int _{0}^{\theta_0}\frac{  \cos(\la t^2\theta) \,d\,\theta}{\theta^2+\eps^2}
 % = 2t^{-4} F_s(\eta,t,0) \eps^{-1} 
 % \tan^{-1}\frac{\theta_0}{\eps}% + \tan^{-1}\frac{-\theta_0}{\eps} )
 =2 \nuu^{-1} t^{-2}    {F_s (\eta,t,0)}  \int _{0}^{\theta_0/\eps}\frac{  \cos( \ga_s\tau w) \,dw}{w^2+1}\,.
  $$
  Consider the integral 
  $$
 2  \int _{0}^{\theta_0/\eps}\frac{  \cos(\ga_s\tau  w) \,dw}{w^2+1} =  2  \int _{0}^{\infty }\frac{  \cos(\ga_s\tau  w) \,dw}{w^2+1}
 - 2  \int _{\theta_0/\eps}^{\infty}\frac{  \cos(\ga_s\tau  w) \,dw}{w^2+1} =: I_1 - I_2\,. 
  $$
  Since
  $$
   2  \int _{0}^{\infty}\frac{  \cos(\xi  w) \,dw}{w^2+1} =   \int _{-\infty}^{\infty}\frac{  e^{i\xi  w} \,dw}{w^2+1} = \pi e^{-|\xi|},
  $$
  then $I_1 = \pi e^{-\ga_s\tau}$. For $I_2$ we have an obvious estimate
  $
  |I_2| \le 2\eps/\theta_0  = C_1\nuu t^{-2}  \,.
  $
   So
  \be\label{4.9}
  \begin{split}
  \Upsilon_s^0 (\eta,t)=\nuu^{-1}   \pi t^{-2}F_s(\eta, t,0) (e^{-\ga_s\tau} +\Delta_t)\,,\qquad
  |\Delta_t | \le C \nuu t^{-2}  \,.
  \end{split}
  \ee

  Now we estimate the difference between $\Upsilon_s$ and $\Upsilon_s^{0}$. Writing 
  $(F_s\mu)(\eta, t,\theta) - (F_s\mu)(\eta, t,0)$ as $A_s(\eta, t)\theta + B_s(\eta, t,\theta)\theta^2$, where 
  $ |A_s|, |B_s| \le  C^\#(s,t)$,  we have 
  $$
  \Upsilon_s - \Upsilon_s^0 =   t^{-4} \int_{-\theta_0}^{\theta_0} 
  \frac{(A_s\theta +B_s \theta^2 )  \cos(\la t^2\theta)  \,d\theta}{\theta^2  +\eps^2}\,.
  $$
  Since the first integrand is odd in $\theta$, then its integral vanishes, and 
  $$
 | \Upsilon_s - \Upsilon_s^0| \le C^\#(s,t)   t^{-4} \int_{-\theta_0}^{\theta_0} 
  \frac{ \theta^2    \,d\theta}{\theta^2  +\eps^2}\le 2 C^\#(s,t)   t^{-4}\theta_0 \,.
  $$
  So by  \eqref{4.9}
   \be\non
  \begin{split}
 |\Upsilon_s (\eta,t) &- \nuu^{-1}   \pi t^{-2}F_s(\eta, t,0) e^{-\ga_s\tau}|\\
 &\le   C^\#(s,t) \big( t^{-4} +\nuu^{-1} t^{-2} \, \nuu t^{-2}\big)  = C_1^\#(s,t)  t^{-4}\,.
  \end{split}
  \ee
  Then, by \eqref{3.9},
   \begin{equation*}
 \begin{split} 
 J_s^r&=
 \int_{\Sod} m(d\eta) \int_{r}^{\infty} dt\,  t^{2d-1} \big[ \nuu^{-1} \pi t^{-2} e^{-\ga_s\tau} F_s(\eta, t,0)
 +O(C^\#(s,t)t^{-4}) \big]\\
&= \int_{\Sod} m(d\eta) \int_{r}^{\infty} dt \,
\big[\nuu^{-1} \pi  e^{-\ga_s\tau}  t^{2d-3}  F_s(\eta, t,0) \big] +  O\big(C^\#(s)\chi_d(r)\big),
 \end{split}
 \end{equation*}
 since $ \int_{r}^{\infty} t^{2d-1}C^\#(s,t)t^{-4}\,dt\leq  C^\#(s)\chi_d(r)$.
  Noting that 
  $$
|  \nuu^{-1} \pi  e^{-\ga_s\tau} \int_{\Sod} m(d\eta) \int_{0}^{r} dt \, t^{2d-3}  F_s(\eta, t,0) | \le C^\#(s) \nuu^{-1} r^{2d-2},
  $$
  we arrive at the inequality 
   \be\label{pr5}
   \begin{split}
  |J_s^r -  \nuu^{-1} \pi  e^{-\ga_s\tau}& \int_{\Sod} m(d\eta) \int_{0}^{\infty} dt \, t^{2d-2}(t^{-1}  F_s(\eta, t,0))|\\
 & \le C^\#(s) \big( \nuu^{-1} r^{2d-2} +  \chi_d(r)  \big).
  \end{split}
  \ee
  Here in view of \eqref{vol_on_*}
   \be\label{pr6}
   \int_{\Sod} m(d\eta) \int_{0}^{\infty} dt \, t^{2d-2}(t^{-1}  F_s(\eta, t,0)) = \int_{\Sigma_*} F_s(z) |z|^{-1}.
     \ee
     
     Since $(\Sv)_0^r\subset K_{2r}$ by Lemma \ref{l_p1}.3, where $\theta_0$ is assumed to be sufficiently 
      small, then by \eqref{pr1},
 \be\label{pr7}
  | J_s^r -\lan J_s, \Sv\ran | \le  C^\#(s)  \nuu^{-2} r^{2d}.
       \ee
\smallskip
	 
	\noindent {\it Step 3.} Final asymptotic.
Consider the integral over the complement to $ \Sv$:
           $$
       \big|  \lan J_s, \R^{2d}\setminus  \Sv  \ran  \big|
       \le    \big\lan |J_s|, \{ |z| \le r\} \big\ran 
     %  \int_{\{ |z| \le r\}}\frac{|F|\,dz}{\om^2+\nuu^2}
       + C \int_{r}^\infty dt\, t^{2d-1}  \int_{S^1}\setminus\Sv \frac{|F_s(z)|\, d_{S^1}}{\omega^2/4 + \nuu^2}.
      $$
      By item 5) of Lemma \ref{l_p1}, $|\omega(z)|\ge Ct^2$ in ${S^t\setminus\Sv}$. Jointly with \eqref{pr1} it   
        implies that 
      \be\label{pr8}
      \begin{split}
        \big|  \lan J_s, \R^{2d}\setminus  \Sv  \ran  \big| &\le C^\#(s)  \nuu^{-2} r^{2d}
        +  C^\#(s) \int_{r}^\infty t^{2d-1} t^{-4}  C^\#(t)
        \, dt
        \\
         &\le  C_1^\#(s) \big(  \nuu^{-2} r^{2d} + \chi_d(r) \big).
        \end{split}
      \ee
   Finally by \eqref{pr7}, \eqref{pr8}, \eqref{pr5} and \eqref{pr6} with  $r=\sqrt\nu$, 
    we have 
   $$
   \left| J_s  - \nuu^{-1}
    e^{-\ga_s\tau}  \pi  \int_{\Sigma_*} F_s(z) |z|^{-1}  \right| \le
    C^\#(s)  \chi_d(\nu). 
   $$
   That is, 
    $
   \left|I_s^{0,2} +\ga_s^{-1} \nu^{-1} e^{-2\ga_s\tau}  \pi  \int_{\Sigma_*} F_s(z) |z|^{-1}  \right| \le
    C^\#(s)  \chi_d(\nu),
   $
   and the  lemma is proved.

\section{Wave kinetic integrals and  equations: proofs}\label{s_proof_kinetic}

\subsection{Proof of Lemma \ref{l_kin_int} }
\label{s_proof_kin_int} 

To prove the lemma we may assume that 
$$
|u^1|_r=\dots = |u^4|_r =1. 
$$
Consider first the integral $J_4$. By the above,
$$
|J_4(s)| \le \int_{\Sigma_*} d\mu^\Sigma(z) \lan x+s\ran^{-r}\! \lan y+s\ran^{-r}\! \lan  x+y+s\ran^{-r}=: \cJ_4(s), 
\quad z=(x,y).
$$
We should verify  that
\be\label{l111}
 \cJ_4(s) \le C\lan s\ran^{-r-1} \quad \forall\, s\in \R^d. 
\ee
If $|s|\le2$, then it is not hard to check that
$$
 \lan x+s\ran \lan y+s\ran  \lan  x+y+s\ran \ge C^{-1}\lan z\ran^2
$$
for a suitable constant $C$ independent from $s$. So the integrand for $\cJ_4(s)$  is bounded by $C_1\lan z\ran^{-2r}$.
Since $r>d$, then  by Proposition~\ref{p_prop}  $  \cJ_4(s) \le C$ if $|s|\le2$. This proves  \eqref{l111}  if $|s|\le2$, and it
remains to consider the case when 
 $$
  R:= |s|\ge2. 
  $$

  Assuming that a vector $s$ as above is fixed, let us consider the sets
  $$
  O_1 =\{\xi \in \R^d : |\xi| \le \tfrac9{10} R\},\quad O_2 =\{\xi  :   \tfrac9{10} \le |\xi| \le \tfrac{11}{10} R\},\; 
  $$
  $$
   O_3 =\{\xi  : |\xi| \ge \tfrac{11}{10} R\},\quad 
    O_{ij} = O_i \cup O_j, \quad i,j=1,2,3,
   $$
  and define 
  $$
  \Sigma^{i, j} =\Sigma_* \cap (O_i\times O_j), \quad 
    \Sigma^{12, j}  =\Sigma_* \cap (O_{12} \times O_j), \quad \text{etc}. 
  $$
Next we denote by $J_4^{i,j}(s)$ the part of the integral in   \eqref{J4} which comes from the integrating over 
 $\Sigma^{i, j}$:
$$
J_4^{i,j}(s)=  \int_{\Sigma^{i, j}}  d\mu^\Sigma(z) \big(u^1(x+s) u^2(y+s) u^3(x+y+s)\big),
$$
 and define similar integrals  $J_4^{12, j}(s)$, etc. Then  $J_4(s) = \sum_{i, j \in\{1,2,3\}}  J_4^{i,j}(s)$ and 
$
|J_4^{i,j}(s)| \le \cJ_4^{i,j}(s)
$,
where 
$$
 \cJ_4^{i,j}(s) = \int_{\Sigma^{i, j}} d\mu^\Sigma(z) \lan x+s\ran^{-r}\! \lan y+s\ran^{-r}\! \lan  x+y+s\ran^{-r}.
$$
Clearly
\be\label{k3}
\cJ_4^{i,j}(s) = \cJ_4^{j,i}(s)\quad \forall\,(i,j),
\ee
and 
 $\cJ_4^{12, j}(s)=  \cJ_4^{1, j}(s) + \cJ_4^{2, j}(s)$, etc. To verify 
 \eqref{l111} it remains to check that
\be\label{k33}
\cJ_4^{i,j} (s) \le C R^{-r-1}\quad \forall\, (i,j),
\ee
when $R= |s|\ge2$. 
  Applying  Theorem \ref{t_disintegr} we find that 
  \be\label{k4'}
  \cJ_4^{i,j} (s) = \int_{O_i} dx\, |x|^{-1} \lan x+s\ran^{-r} \!
  \int_{x^\perp\cap O_j} \dy\,   \lan y+s\ran ^{-r}\! \lan x+y+s\ran^{-r}.
  \ee
  By $s_x$ (by $s_y$) we will denote the projection of $s$ on the space $x^\perp$ (on $y^\perp$), and by $\hat s, \hat x, \hat y$
  -- the vectors $s/R, x/R, y/R$; so $| \hat s|=1$.

  We will estimate the r.h.s. of \eqref{k4'} for each set $\{i,j\} \subset\{1,2,3\}$, using the following elementary inequalities, related to the domains 
  $O_j$:
  \be\label{k6}
  \lan y+s\ran \geq |y+s| 
   \ge 
   \tfrac{1}{10}
   R \quad \forall\, y\in O_{13},
  \ee
 % \be\label{k7}
 % \lan y+s\ran \geq |y+s| 
 %  \ge 
 %  \tfrac{1}{11}
 %  	 |y| \quad \forall\, y\in O_{3},
 % \ee
  \be\label{k8}
  \lan x+y+s\ran \ge C^{-1}(|x|+ |y|) \quad \forall\, z \in (O_{23}\times O_{23})\cap\Sigma.
  \ee
%   \be\label{k9}
%  \lan x+y+s\ran \ge C^{-1} |y| \quad \forall\, z \in (\R^d\times O_{3})\cap\Sigma. 
%  \ee
  Proof of \eqref{k8} uses that $|x+y|\ge (|x|+|y|)/\sqrt2$  for $x,y \in\Sigma$. 
  We will also use the integral 
  inequalities below, where $|s| =R\ge2$ and 
   $\hat\chi_l=\hat\chi_l(R)$ equals 1 if $l\ne0$ and equals $\ln R$ if $l=0$:
   \be\label{k10}
  \int_{\R^{d}} \lan y+\xi\ran^{-p} dy \le C\quad \forall\, \xi\in \R^{d}\quad\text{ if}\;\; p>d,
  \ee
%   \be\label{k102}
%  \int_{|y|\le2} \lan  R(y+\xi)  \ran^{-\hat r} dy \le C \hat\chi_{d-\hat r} R^{ \max(-d, -\hat r)}
%\quad\text{ if}\; \;  |\xi| \le3, 
%  \ee
   \be\label{k100}
  \int_{|x| \le \frac{11}{10}R}  |x|^{-1} \lan x+ s \ran^{-\hat r} dx \le  C \hat\chi_{d-\hat r} R^{ \max(-1, d-1-\hat r)}, 
  % C R^{-1},
    \ee
     \be\label{k101}
  \int_{|x| \ge \frac{9}{10}R}  |x|^{A} \lan x+ s \ran^{-\hat r} dx \le   C \hat\chi_{d-\hat r} R^{ \max(A,A+d -\hat r)}
  %\quad\text{ if}\;\;  A<\hat r-d. 
    \ee
    if $  A<\hat r-d.$ We will use these relations with $d:=d$ or $d:=d-1$. 
    The first inequality is obvious. 
    To prove \eqref{k100} note that there the  integral equals 
    $$
    R^{d-1} \int_{|y| \le \frac{11}{10}}  |y|^{-1}
    \lan R(y+\hat s)\ran^{-\hat r} dy, \quad \hat s = s/R. 
    $$
     It is straightforward to see that this   integral  (with the pre-factor), taken 
  over the ball $ \{   |y + \hat s |\le 1/10 \}$ is 
 $\le   C \hat\chi_{d-\hat r} R^{ \max(-1, d-1-\hat r)}$, while the integral over the ball's complement is 
 $\le CR^{d-1-\hat r}.  $
 It implies \eqref{k100}.   Proof of  \eqref{k101}  is similar. 
 \bigskip
 
 \noindent
 {\it Estimating the integral $\cJ_4^{23,23}(s)$.} \ By \eqref{k4'} and \eqref{k8} the integral is bounded by
 $$
 %\cJ_4^{23,23} \le
 C \int_{|x|\ge \frac{9}{10} R} dx\, |x|^{-1}\lan x+s\ran ^{-r} \int_{|y| \ge \frac{9}{10} R} \dy\,  \lan y+s \ran^{-r} (|x|+|y|)^{-r}.
 $$
 Since $\lan y+s \ran \geq \lan y+s_x\ran$ and $r>d$, 
 the internal integral  is less than 
 $$
 C R^{-r}   \int_{|y|\ge \frac{9}{10}R } d_{x^\perp} y\, \lan y+s_x \ran^{-r} \le  C_1  R^{-r} ,
 $$
 where we used \eqref{k10}   with $d:=d-1$.   So 
 $$
 \cJ_4^{23,23}(s) \le  C R^{-r}
 \int_{|x|\ge \frac{9}{10} R} dx\, |x|^{-1}\lan x+s\ran ^{-r} \le R^{-r-1}
 $$
 by \eqref{k101}  with $A=-1$ and $\hat r=r$. This implies \eqref{k33} if $i,j \in \{2,3\}$. 
 \medskip

 \noindent
 {\it Estimating the integral $\cJ_4^{12,13}(s)$.} \ By \eqref{k4'} and \eqref{k6}, 
 $$
 \cJ_4^{12,13}(s) \le CR^{-r}  \int_{|x|\le \frac{11}{10} R} dx\, |x|^{-1}\lan x+s\ran ^{-r} 
 \int  \dy\, \lan x+y+s\ran ^{-r} .
 $$
 Since $\lan x+y+s \ran \ge \lan y+s_x \ran$, by \eqref{k10} with $d:=d-1$ 
 the internal integral is bounded by a constant. So in view of \eqref{k100}
  with $\hat r=r$, we get
 $\cJ_4^{12,13}(s) \le CR^{-r-1} $. This implies \eqref{k33} if $i \in\{1,2\}$, $j \in\{1,3\}$. 
 \medskip

 Evoking the symmetry \eqref{k3} we see that we 
 have checked \eqref{k33} for all $(i,j)$, and thus have proved \eqref{l3} for $l=4$. 
 \bigskip
 
 Other three integrals are easier. Let us first  consider   $J_3(s)$,
$$
J_3(s) = u^4(s)\int_{\Sigma_*}d\mu^\Sigma(z)  u^1(x+s) u^2(y+s), \;\; %\ga_3 =\ga_{s_3} =\lan x+y+s\ran^{2r_*},
$$
where   $|u^j|_r=1$ for each $j$. Then
$$
|J_3(s)| \le  \lan s \ran^{-r}  \int_{\Sigma_*} d\mu^\Sigma(z)  
\lan x+s\ran^{-r}\! \lan y+s\ran^{-r}  =: \cJ_3(s).
$$
If  $|s|\leq 2$, then $\cJ_3(s)$ is bounded by
$$
C  \int_{\Sigma_*} d\mu^\Sigma(z)\lan x+s\ran^{-r}\! \lan y+s\ran^{-r}
=C\int dx\, |x|^{-1} \lan x+s\ran^{-r} \!
\int_{x^\perp} \dy\,   \lan y+s\ran ^{-r}. 
$$
Applying \eqref{k10} to the both internals, 
using that $r>d$ and the  integrability of $|x|^{-1}$ over a neighbourhood of  $0\in\R^d$, we see that  $|\cJ_3(s)| \le C$, 
which implies $|\cJ_3(s)| \le C\lan s\ran^{-r-1}$ for $|s|\leq 2$.

Now we pass to the case $R=|s|\geq 2$.   Then   %$|J_3(s)| \le\sum_{i,j}\cJ_3^{i,j}(s)$, 
% the integrals $\cJ_3^{i,j}(s)$ are symmetric in $i,j$, and   %$\cJ_3^{i,j} (s)$ equals 
  $$
   \cJ_3 (s)  \le
    { \lan s \ran^{-r}}  \! 
   \int dx\, |x|^{-1} \lan x+s\ran^{-r} \!
  \int_{x^\perp} \dy \, \lan y+s\ran ^{-r} .%\! \lan x+y+s\ran^{-2r_*}. 
  $$
   Applying  \eqref{k10} we see that the internal integral is bounded by a constant, while by  \eqref{k100} and  \eqref{k101} 
   with $A=-1$ and $\hat r=r$ the external is $\le CR^{-1}$.    This proves  \eqref{l3} for $l=3$. 
 
   \bigskip
  
  The integrals $J_1$ and $J_2$ are very similar and we only consider $J_1$:
  $$
J_1(s) =   u^4(s) \int_{\Sigma_*}d\mu^\Sigma(z)   u^2(y+s) u^3(x+y+s) .% \quad \ga_1= \lan x+s\ran^{2r_*}.
$$
Assume first $|s|\leq 2$. Using the disintegration of the measure $d\mu^\Sigma(z)$ in the form \eqref{Pii} 
%which follows from Theorem~\ref{t_disintegr}, 
we get  
$$
|J_1(s)|\leq C \int dy\, |y|^{-1} \lan y+s\ran^{-r} \!
\int_{y^\perp} \ddy x\,   \lan x+y+s\ran^{-r}.
$$
Since $\lan x+y+s\ran^{-r}\leq \lan x+s_y\ran^{-r}$ and $r>d$, 
 then using \eqref{k10}   we see that $|J_1(s)|\leq C$ for $|s|\leq 2$. 

If  $R=|s|\geq 2$, then,  in view of  \eqref{Pii}, 
  %\be\label{k**}
  $$
  |J_1 (s)| \le
     { \lan s \ran^{-r}}   \! 
   \int dy\, |y|^{-1} \lan y+s\ran^{-r} \!
  \int_{y^\perp}  \ddy x\,    \lan x+y+s\ran^{-r}.
 $$%  \ee 
   Again, using  \eqref{k10} we estimate the internal integral, using  \eqref{k100} 	and \eqref{k101} estimate the external and get that   $ |J_1 (s) | \le CR^{-r-1}$.

  Thus the lemma's assertion also holds for $l=1,2$, and Lemma \ref{l_kin_int}   is proved.

\subsection{Proof of Theorem \ref{t_kin_eq}} \label{proof_t_kin_eq}
Let us denote by $\{S_t, t\ge0\}$ the semigroup of the linear equation, i.e. $S_t = e^{-t\Lc}$. Obviously, 
$
| S_t|_{\cC_r, \cC_r} \le e^{-2t}.
$
The solution $u^0(t,x)$ of the linear problem \eqref{w_k_e}${}_{\eps=0}$, \eqref{wk1} is given by the Duhamel 
integral
$$
u^0(t) =\frak C (u_0, f)(t) := S_t u_0 +\int_0^t S_{t-s} f(s)\, ds,
$$
so 
\be\label{wk3}
\| u^0\tr  =\| \frak C (u_0, f) \tr
\le |u_0|_r + \frac12 \| f\tr.
\ee
A solution $u$ for the problem \eqref{w_k_e}, \eqref{wk1}  is a fixed point of the operator
$$
u \mapsto \frak C(u_0, f+\eps K(u)) =: \frak B(u). 
$$
Assuming \eqref{wk4} and using \eqref{k2} we see from \eqref{wk3} that
$$
\|  \frak B(u)\tr \le \frac32 C_* +\frac{\eps}2 C_r^K\|u\tr^3,
$$
where $C_r^K$ is the constant from \eqref{k2}.
So
 the operator $\frak B$ preserves the ball
$$
B_R = \{ u\in X_r, \| u\tr \le R\}, \quad R= \frac32 C_* +1, 
$$
if $\eps\ll1$. Using Corollary \ref{c_7.1} we see that 
the operator $\frak B$ defines a contraction of the ball $B_{R}$ if
 $\eps\ll1 $. 
 Accordingly   eq. \eqref{w_k_e} has a unique solution $u\in B_R$. 
 
 Finally, 
 if $(u_{01}, f_1)$ and  $(u_{02}, f_2)$ are  two sets of initial data, satisfying  \eqref{wk4},  
 and $u^1, u^2$ are the corresponding solutions, then 
 $$
 \|u^1 - u^2 \tr  \le  |u_{01} -u_{02}  |_r  + \frac12 \| f_1 - f_2\tr +  \frac32 \eps\, C^K_r R^2   \|u^1-u^2\tr
 $$
 for all $t\geq 0$, so 
 $$
 \|u^1-u^2\tr \le (1-  \frac32 \eps\, C^K_r R^2  )^{-1}( |u_{01} -u_{02}  |_r  + \frac12 \| f_1 - f_2\tr). 
 $$
 This implies \eqref{wk6} if $\eps\ll1$ and   the theorem is proved.

	\section{Addenda}\label{s_append}
	\subsection{Discrete turbulence} \lbl{app_discr_turb}
	In order to study the double limit \eqref{limit'}  it is natural to examine first the limit
	$\nu\to0$ (with $L$ and $\rho$ fixed), known as the {\it limit of discrete turbulence}, see \cite{Naz11}. To do this consider the following
	{\it effective equation}:
	\begin{equation}\label{ku5}
	 \dot \goa_s+\gamma_s \goa_s = i\rho L^{-d} \Big({\sum}_{s_1} {\sum}_{s_2} \dess \delta(\omega) 
	\goa_{s_1} \goa_{s_2} \bar \goa_{s_3} - \goa_s|\goa_s|^2\Big) 
	+b(s) \dot\beta_s
	\,,\quad s\in{{\mathbb Z}}^d_L  \,,
	\end{equation}
	where $\delta(\omega)$ is the delta-function of $\omega=\omega^{12}_{3s}$ (equal 1 if $\omega^{12}_{3s}=0$ and equal 0 otherwise).
	The following result is proven in \cite{KM16, HKM}.
	
	\begin{theorem}\label{t_discrturb}
	If $r_*$ is sufficiently big in terms of $d$, 	 then 
		eq. \eqref{ku5} is well posed and mixing. When $L$ and $\rho$ are  fixed and  $\nu\to0$, then\\
		i) solutions of \eqref{ku4a} converge in distribution, on time intervals of order 1, to solutions of \eqref{ku5} with the same initial
		data at $\tau=0$;\\
		ii) the  unique  stationary measure $\mu_{\nu, L}$
		of \eqref{ku4a} weakly converges to the unique stationary measure of eq.  \eqref{ku5}.
	\end{theorem}

	\subsection{Estimates for integrals with 
	 fast oscillating exponents, given by quadratic forms}
	\label{s_stat_ph}
	
	Let $\phi$  be a complex  $ L_1$--function on $\R^n$ such that its Fourier transform $\hat\phi\in L_1(\R^n)$, where in this appendix 
we define   $\hat \phi(\xi) $ as $ \int e^{-ix\cdot\xi} \phi(x)\,dx $. Let  $Q$ be a symmetric non-degenerate real
 $n\times n$-matrix.
	Consider the integral
	\be\label{I_la}
	I(\nu)= \int_{\R^n}e^{i\nu^{-1} x\cdot Qx/2 } \bar \phi(x) 
	\,dx,\quad 0<\nu\le 1.
	\ee
	The Fourier transform of the function $e^{i \nu^{-1} x\cdot Qx/2}=: F_0$ is
	$$
	\hat F_0=
	(2\pi\nu)^{n/2} \zeta\, |\det Q|^{-1/2}\, e^{-i\nu \xi\cdot Q^{-1}\xi/2}, %\qquad \zeta\in\C,\qu |\zeta|=1
	$$
	where $\zeta$ is some complex number of unit norm 
	(see \cite{DS, Hor}).  Formally applying  Parseval's  identity we get 
	\be\label{pars}
	I(\nu)=   (2\pi)^{-n}\int \hat F_0 \bar{\hat \phi}\, d\xi =
	\big(\frac{\nu}{2\pi}\big)^{n/2} \zeta\, |\det Q|^{-1/2}\,  \int_{\R^n} e^{-i\nu \xi\cdot Q^{-1}\xi/2} \ov{\hat \phi(\xi)}\,d\xi.
	\ee
	To justify the validity of \eqref{pars} in out situation, we approximate $F_0$ by functions 
	$$
	F_\eps =  e^{i \nu^{-1}  x\cdot (Q+ i\eps I)x /2}, \qquad \eps>0.
	$$
	These are  Schwartz functions whose Fourier transforms
$$
\hat F_\eps=
	(2\pi\nu)^{n/2}  (\det iQ - \eps I)^{-1/2}\, e^{-i\nu   \xi\cdot (Q+i\eps I)^{-1}\xi/2}
$$
converge to $\hat F_0(\xi)$ for each $\xi$, as $\eps\to0$ (see \cite{DS, Hor}).  For every $\eps>0$ Parseval's identity holds 
for $F_0$ replaced by $F_\eps$. Passing there to the limit as $\eps\to0$ using Lebesgue's theorem we recover \eqref{pars}. 
	So 
	\be\label{stph_est1}
	| I(\nu) | \le  \big(\frac{\nu}{2\pi}\big)^{n/2} |\det Q|^{-1/2} | \hat \phi|_{L_1}. 
	\ee
	%Note that $I(\nu)$ also satisfies the trivial estimate 
	%$\ 	| I(\nu) | \le  |  \phi|_{L_1}. $
	We recall that 
	\be\label{sob}
	| \hat \phi|_{L_1} \le C_m \| \phi\|_{H^m(\R^n)} \quad \text{for any } \; m>n/2,
		\ee	
	where $H^m(\R^n)$ denotes the Sobolev space on $\R^n$.

	\subsection{Direct proof of estimate \eqref{int_est}}\label{a_stat_phase}
	
	In this appendix we estimate directly  integral $I_s$ in  \eqref{I_s1} with $T=\infty$ (indirectly this integral in the form \eqref{I_s} was 
	 estimated  in  \eqref{int_est} via Theorem~\ref{t_singint}).
	Setting $x=\sqrt2 (s_1-s)$ and $y=\sqrt2(s_2-s)$, in view of \eqref{omega} we get
	$\oms=-x\cdot y$.  
	Then, denoting $z=(x,y)$, we obtain  
	\be\lbl{I_s(tau)}
	I_s=\int_{\R^{2d}}dz\,\int_{\R^2_-} dl\;F_s(l,z) e^{i\nu^{-1}x\cdot y(l_2-l_1)},
	\qquad l=(l_1,l_2),
	\ee
	where %$\R^2_-=[-\infty,0]\times[-\infty, 0]$ and 
	 $F_s(l,z)$ is the function 
	$$
	2B(s_1,s_2,s_3)e^{-|l_1-l_2|(\ga_1+\ga_2+\ga_3)+\ga_s(l_1+l_2)},\qu s_3=s_1+s_2-s,
	$$
	written in the coordinates $l,z$. 
Since $B$ is a Schwartz function and $\ga_s\geq 1$,  for $l\in\R^2_-$ the function $F_s$ satisfies the estimate
	\be\lbl{F_s-est}
	|\p_{z^\al}F_s(l,z)|\leq C_{\al}^\#(s)C_{\al}^\#(z)e^{(l_1+l_2)},
	\qquad \forall \al.
	\ee
	Let	$\cN$ be the subset of $\R^2_- =\{ l=(l_1,l_2)\}$ where $|l_1-l_2|\geq\nu$, and $\cN^c$ be its complement.
	Then, bounding the complex exponent in \eqref{I_s(tau)} by one, we find 
	\be\lbl{N^c}
	|\lan I_s,\cN^c\ran|\leq C^\#(s)\nu,
	\ee
	 where we recall that the notation $\lan I_s,\cN^c\ran$ was introduced in \eqref{<I,M>}.
	
	To estimate the term $\lan I_s,\cN\ran$, we note that the integral over $dz$ in \eqref{I_s(tau)} has the form \eqref{I_la} with $\nu:=\nu|l_1-l_2|^{-1}$ and $n=2d$. Then, due to \eqref{stph_est1} and \eqref{sob}, 
	\be\non
	|\lan I_s,\cN\ran|\leq \frac{1}{(2\pi)^d}\int_\cN\frac{\nu^d}{|l_1-l_2|^d}
	\big| \hat F_s(l,\cdot)\big|_{L^1} \, dl
	\leq C^\#(s)\nu^d\int_\cN\frac{e^{(l_1+l_2)}}{|l_1-l_2|^d}\, dl,
	\ee
	where in the last inequality we  used \eqref{F_s-est} and the definition of the set $\cN$. 
	Since $d\geq 2$, this implies 
	$\ 
	|\lan I_s,\cN\ran|\leq C_1^\#(s)\nu^d/\nu^{d-1}=C_1^\#(s)\nu.
	$
	Combining the obtained  inequality with \eqref{N^c}, we get the desired estimate \eqref{int_est}.

		\subsection{Proof of Proposition~\ref{l:a-goa}.}
	\lbl{s:a-goa}
	
	We start by explaining the scheme of the proof.  Let us express the process $a_s^{(n)}$ through the processes $a^{(0)}_k$ by iterating formula \eqref{an} $n$ times and compute the expectation $\EE |a_s^{(n)}(\tau)|^2$ applying the Wick theorem. It can be shown that when $L\to\infty$  this expectation stays of size one. The reason for this is as follows: arguing by induction we see that 
	an affine space of variables $k\in\Z^d_L$ which serves as the set of indices 
over which we take summation in the expression for $|a_s^{(n)}(\tau)|^2$ through $a^{(0)}_k, \,\bar a^{(0)}_k$  has dimension $4nd$ (indeed,
for $n=0$ this dimension obviously is zero, while for $n=1$ due to \eqref{a1} and the conjugated formula it is $4d$, etc.). 

 When computing the expectation $\EE|a_s^{(n)}(\tau)|^2$ we make the Wick pairing of terms $a^{(0)}_k, \,\bar a^{(0)}_{k'}$ and it is possible to see that
 the dimension  of the corresponding 
  space of variables becomes $2nd$ (since  in view of \eqref{corr_a_in_time} 
 for Wick-coupled terms $a^{(0)}_k$ and $\bar a^{(0)}_{k'}$ we should have $k=k'$).
 \footnote{	  This is true for the most of Wick-pairings, while for some of them the affine space of variables may become empty because of the restrictions of the type $\{s_1,s_2\}\ne\{s_3,s\}$ imposed by $\dep$.}
  At the same time, after $n$ iterations of \eqref{an} we get a factor $L^{-2nd}$ in the formula for  $\EE |a_s^{(n)}(\tau)|^2$; this leads to the claim. For an example of such computation see Section~\ref{s:2nd moments}. Similarly,  $\EE |\goa_s^{(n)}(\tau)|^2$ stays of order one when $L\to\infty$.
	 
	 Now we express the difference $\Delta^n_s(\tau):=\goa_s^{(n)}(\tau)-a_s^{(n)}(\tau)$ through the processes $a^{(0)}_k$ and write $\Delta^n_s(\tau)$ as a finite sum $\Delta^n_s=\sum\Delta_s^{n,j}$. Each term $\Delta^{n,j}_s$ is obtained by iterating  \eqref{an}, where at
 at least one iteration  $\dep$ is replaced by $-\de^{s_1}_{s_2}\de^{s_2}_{s_3}\de^{s_3}_{s}$ (this corresponds to the "diagonal" term $-\goa^{(n_1)}_s \goa^{(n_2)}_s \bar \goa^{(n_3)}_s$ in \eqref{ana}).
 % is taken into account, i.e. $\dep$ is replaced by zero. 
 This implies that dimension of the affine space of variables over which we take the summation  when computing $\EE|\Delta_s^{n,j}(\tau)|^2$ drops at least by $2d$ in comparison with that for $\EE|a_s^{(n)}(\tau)|^2$ and $\EE|\goa_s^{(n)}(\tau)|^2$. At the same time we still have the factor $L^{-2nd}$ in the formula for $\EE |\Delta_s^{n,j}|^2$. This implies the desired estimate  $\EE |\Delta_s^{n,j}|^2\leq C^\#(s) L^{-2d}$.
	
	We give a complete  proof only for the cases $n=1,2$   since in Theorem C which is  the main result of this paper we only
	deal with $a^{(n)}$ and  $\frak a^{(n)}$ such that $n\le2$.	 A general case can be considered similarly, by analyzing 
	 the dimensions of the affine spaces of variables over which we take the summation in the formula for $\Delta^n$. In particular, no 
	 delicate cancellation argument is used.   However, the proof  is cumbersome due to the 
 notational difficulty,    arising when expressing $\Delta^n$ through $a^{(0)}$.

	By \eqref{a1a} and \eqref{a1}, 
	\be\lbl{De1}
	\Delta^1_s(\tau)=-iL^{-d}\int_{-T}^\tau e^{-\ga_s(\tau-l)}|a_s^{(0)}(l)|^2a_s^{(0)}(l)\,dl,
	\ee
	so in view of \eqref{corr_a_in_time} the Wick theorem implies 
	$
	\EE|\Delta_s^1(\tau)|^2\leq C^\#(s)L^{-2d}.
	$
	Next,
	$\Delta^2_s=2\Delta^{2,1}_s+\Delta_s^{2,2}$, where the term
	\begin{align}\non
	\Delta^{2,1}_s(\tau)=iL^{-d}\int_{-T}^\tau e^{-\ga_s(\tau-l)}
		\Big( &\sum_{s_1,s_2}\delta'^{12}_{3s} (\Delta^{1}_{s_1} a^{(0)}_{s_2}{\bar a}^{(0)}_{s_3})(l)	e^{i\nu^{-1} l \oms} \\ \lbl{De21}
	& - \goa_s^{1}(l)|a_s^{(0)}(l)|^2\Big)\,dl=\Delta_s^{2,1;1}(\tau)-\Delta_s^{2,1;2}(\tau)
	\end{align}
	corresponds to the choice $n_1=1$ and $n_2=n_3=0$ in \eqref{ana} with $n=2$. The term $\Delta_s^{2,2}$  corresponds to the choice $n_1=n_2=0,\,n_3=1$ and has a similar form. 
	Below we only discuss  $\Delta^{2,1}$. 
	
 Because of the factor $L^{-d}$ in \eqref{De21} it is straightforward to see that $\EE|\Delta_s^{2,1;2}(\tau)|^2\leq C^\#(s)L^{-2d}. $ Let us study the term $\Delta_s^{2,1;1}$. By \eqref{De21} joined with \eqref{De1},
	\begin{align}\non
	\EE|\Delta_s^{2,1;1}(\tau)|^2&=
	L^{-4d}\int_{-T}^\tau dl\,\int_{-T}^l dl_1\,\int_{-T}^\tau dl'\,\int_{-T}^{l'} dl'_1\, e^{-\ga_s(2\tau-l-l')} \\ \non
	&\sum_{s_1,s_2,s_1',s_2'}\delta'^{12}_{3s}\depp 
	 e^{-\ga_{1}(l-l_1)-\ga_{1'}(l'-l_1')} 	e^{i\nu^{-1} (l \oms-l'\om^{1'2'}_{3's})}\\ \lbl{EDe2}
	\EE \Big(a_{1}^{(0)}&(l_1)|a_{1}^{(0)}(l_1)|^2\, a_{2}^{(0)}(l)\bar a_{3}^{(0)}(l)\,
	\bar a_{1'}^{(0)}(l'_1)|a_{1'}^{(0)}(l'_1)|^2\, \bar a_{2'}^{(0)}(l') a_{3'}^{(0)}(l')\Big).
	\end{align}
Next we apply the Wick theorem. Due to \eqref{corr_a_in_time}, non-conjugated variables $a_k^{(0)}$ can be coupled with conjugated variables $\bar a_{k'}^{(0)}$ only, and $k$ should equals to $k'$. Consider e.g. the case when $a^{(0)}_1$ is coupled with $\bar a^{(0)}_1$, $a^{(0)}_1$ with $\bar a^{(0)}_{1'}$, $a^{(0)}_2$ with $\bar a^{(0)}_{2'}$,  $a^{(0)}_{3'}$ with $\bar a^{(0)}_3$ and $a^{(0)}_{1'}$ with $\bar a^{(0)}_{1'}$ (we write $|a_k|^2$ as $a_k\bar a_k$). Then, due to  \eqref{corr_a_in_time},
the corresponding to this Wick pairing term in the expression for $\EE|\Delta_s^{2,1;1}(\tau)|^2$ does not exceed
$$
CL^{-4d}\sum_{s_1,s_2}\dep\big(B(s_1)\big)^3B(s_2)B(s_3)\leq L^{-2d} C^\#(s),
$$
where $B(s)$ is defined in \eqref{n_s_0}.
The other Wick pairings can be considered similarly ~--- the dimension of the space of variables over which we take the summation always does not exceed $2d$, so the resulting estimate will be the same.

\begin{comment}

Below we list all possible Wick couplings; there we use \eqref{corr_a_in_time} and the definition of $\de'$ (we write $|a_k|^2$ as $a_k\bar a_k$). 
Since  the Wick coupling of variables $a_k$ and $a_{k'}$ does not vanish only if $k=k'$ we get that dimension of the space of variables over which we take the summation in \eqref{EDe2} is at most $2$, then

1) $a^{(0)}_1$ with $\bar a^{(0)}_1$; $a^{(0)}_1$ with $\bar a^{(0)}_{1'}$ or $\bar a^{(0)}_{2'}$; $a^{(0)}_2$ with $\bar a^{(0)}_{2'}$ or $\bar a^{(0)}_{1'}$ correspondingly; $a^{(0)}_{3'}$ with $\bar a^{(0)}_3$.

2)  $a^{(0)}_{1'}$ with $\bar a^{(0)}_{1}$ or $\bar a^{(0)}_{3}$ 
and $a^{(0)}_{3'}$ with $\bar a^{(0)}_{3}$ or $\bar a^{(0)}_{1}$ correspondingly; 
$a^{(0)}_{1}$ with $\bar a^{(0)}_{1'}$; $a^{(0)}_{1}$ with $\bar a^{(0)}_{1'}$ or $\bar a^{(0)}_{2}$
and $a^{(0)}_{2}$ with $\bar a^{(0)}_{2'}$ or $\bar a^{(0)}_{1'}$ correspondingly.
	
\end{comment}

\end{document}